\documentclass[11pt]{article}
\usepackage[T1]{fontenc}
\usepackage{mathrsfs}
\usepackage[centertags]{amsmath}
\usepackage{amsfonts}
\usepackage{amssymb}
\usepackage{amsthm}
\usepackage{cases}
\usepackage{epsfig}
\usepackage{graphicx}
\usepackage{color}
\usepackage{multirow}
\usepackage{tabularx}
\usepackage{booktabs}
\usepackage{threeparttable}
\usepackage{graphicx}
\usepackage{placeins}
\usepackage{array}
\usepackage{caption,float}
\usepackage{makecell}
\usepackage{microtype}         
\usepackage{bm}                
\usepackage{enumerate}         
\usepackage{tikz}
\usepackage{subcaption}
\allowdisplaybreaks[4]
\usepackage{MnSymbol}
\usepackage{hyperref}
\usepackage{bbm}
\usepackage[ruled,vlined]{algorithm2e}
\usepackage{rotating}
\usepackage{authblk}
\usepackage{geometry}
\usepackage{natbib}
\setcitestyle{authoryear,open={(},close={)}} 
\usepackage{csquotes}
\makeatletter

\patchcmd{\NAT@citex}
  {\@citea\NAT@hyper@{%
     \NAT@nmfmt{\NAT@nm}%
     \hyper@natlinkbreak{\NAT@aysep\NAT@spacechar}{\@citeb\@extra@b@citeb}%
     \NAT@date}}
  {\@citea\NAT@nmfmt{\NAT@nm}%
   \NAT@aysep\NAT@spacechar\NAT@hyper@{\NAT@date}}{}{}

\patchcmd{\NAT@citex}
  {\@citea\NAT@hyper@{%
     \NAT@nmfmt{\NAT@nm}%
     \hyper@natlinkbreak{\NAT@spacechar\NAT@@open\if*#1*\else#1\NAT@spacechar\fi}%
       {\@citeb\@extra@b@citeb}%
     \NAT@date}}
  {\@citea\NAT@nmfmt{\NAT@nm}%
   \NAT@spacechar\NAT@@open\if*#1*\else#1\NAT@spacechar\fi\NAT@hyper@{\NAT@date}}
  {}{}

\makeatother

\newtheorem{theorem}{Theorem}
\newtheorem{lemma}{Lemma}

\newtheorem{proposition}{Proposition}
\newtheorem{corollary}{Corollary}

\theoremstyle{definition}
\newtheorem{definition}{Definition}
\newtheorem{condition}{Condition}
\newtheorem{assumption}{Assumption}

\newtheoremstyle{remboldstyle}
  {}{}{}{}{\itshape}{.}{.5em}{{\thmname{#1 }}{\thmnumber{#2}}{\thmnote{ (#3)}}}
\theoremstyle{remboldstyle}
\newtheorem{remark}{Remark}
\newtheorem{example}{Example}

\newcommand{\argmax}{\operatornamewithlimits{\arg\max}}
\newcommand{\argmin}{\operatornamewithlimits{\arg\min}}
\newcommand{\pr}{\operatorname{\mathsf{P}}}
\renewcommand{\d}{\operatorname{d}}
\newcommand{\E}{\operatorname{\mathsf{E}}}
\newcommand{\tree}{\mathcal{T}} 
 
\newcommand{\idp}{\perp} 
\newcommand{\cdt}{\mid} 
\newcommand{\dto}{\xrightarrow{d}}
\newcommand{\pth}[2]{[{#1} \rightsquigarrow {#2}]}
\newcommand{\I}{\operatorname{\mathbbm{1}}}
\newcommand{\sgn}{\operatorname{sgn}}
\newcommand{\eps}{\varepsilon}

\definecolor{burntorange}{rgb}{0.8, 0.33, 0.0}
\definecolor{britishracinggreen}{rgb}{0.0, 0.26, 0.15}
\definecolor{darkcerulean}{rgb}{0.03, 0.27, 0.49}
\definecolor{blue}{rgb}{0.19,0.55,0.91}
\definecolor{magenta(dye)}{rgb}{0.79, 0.08, 0.48}
\definecolor{emerald}{rgb}{0.31, 0.78, 0.47}
\definecolor{brightlavender}{rgb}{0.75, 0.58, 0.89}

\newcommand{\mdf}[1]{}

\newcommand{\Address}{{
  \footnotesize
  Shuang Hu, \par
  \noindent\textsc{School of Science/Key Laboratory of Intelligent Analysis and Decision on Complex Systems, Chongqing University of Posts and Telecommunications, 400065 Chongqing, China.}\par\nopagebreak
  \noindent\textsc{LIDAM/ISBA, UCLouvain, B1348 Louvain-la-Neuve, Belgium.}\par\nopagebreak
  \noindent\textit{E-mail address:} \texttt{hushuang@cqupt.edu.cn}
}}

\begin{document}
\title{Modelling multivariate extreme value distributions via Markov trees\thanks{Running title: \emph{Modelling multivariate extremes}}}

\author[1,3]{SHUANG HU}
\author[2]{ZUOXIANG PENG}
\author[3]{JOHAN SEGERS}

\affil[1]{School of Science/Key Laboratory of Intelligent Analysis and Decision on Complex Systems, Chongqing University of Posts and Telecommunications, Chongqing, China}
\affil[2]{School of Mathematics and Statistics, Southwest University, Chongqing, China}
\affil[3]{LIDAM/ISBA, UCLouvain, Belgium}

\date{}
\maketitle

\begin{abstract}
Multivariate extreme value distributions are a common choice for modelling multivariate extremes.
In high dimensions, however, the construction of flexible and parsimonious models is challenging. We propose 
to combine bivariate max-stable distributions into a Markov random field with respect to a tree.
Although in general not  
max-stable itself, this Markov tree 
is attracted by a multivariate max-stable distribution.
The latter serves as a tree-based approximation to an unknown max-stable distribution with the given bivariate distributions as margins. Given data, we learn an appropriate tree structure by Prim's algorithm with estimated pairwise upper tail dependence coefficients 
as edge weights. The distributions of pairs of connected variables can be fitted in various ways. The resulting tree-structured max-stable distribution allows for inference on rare event probabilities, as illustrated on river discharge data from the upper Danube basin.
\smallskip

\end{abstract}

\section{Introduction}
\label{sec:motiv}
The modelling of tail dependence and extreme events has become a statistical problem of great interest. Through asymptotic theory, extreme value theory postulates models for sample extremes (max-stable distributions) and excesses over high thresholds (generalized Pareto distributions), both univariate and multivariate. Fields of application include meteorology and climatology (heat waves, wind storms, floods), finance (stock market crashes), reinsurance (large claims), and machine learning (anomaly detection). Text-book treatments can be found in \citet{C01}, \citet{B04} and \citet{HF06}, while \citet{DH15} provides a more recent survey of the field.

Multivariate max-stable distributions arise as the possible limiting distributions of normalized component-wise maxima of independent copies of a random vector. Their construction, however, is not straightforward because of their complicated dependence structure. Non-parametric inference on multivariate max-stable distributions is usually performed in terms of dependence functions such as the angular or spectral measure \citep{EPD01,ES09}, the Pickands dependence function \citep{D91} and the stable tail dependence function \citep{H92,D98}. Early reviews of parametric models can be found in \citet{C91}, \citet[Chapter~3]{K00}, and \citet[Chapter~9]{B04}. Well-known parametric models include the logistic and negative logistic families \citep{G60, T90}, the H\"{u}sler--Reiss model \citep{H89}, mixtures of Dirichlet distributions \citep{BD07}, the pairwise beta distribution \citep{CDN10}, max-linear factor models \citep{EKS12}, the extremal \emph{t} process \citep{DPR12}, and models based on Bernstein polynomials \citep{HdCC17}. Some generic construction principles yielding known and new examples are proposed in \citet{BS11} and \citet{KRSW19}. Although these methods perform quite well and have reasonable efficiency in low and moderate dimensions, their application in higher dimensions remains challenging stemming from the potentially complex dependence structure of multivariate max-stable distributions, computational challenges, and the need to balance flexibility with parsimony. 

Graphical models represent relatively simple probabilistic structures and are a popular way to model multivariate distributions in a sparse way. Lately, the intersection of extreme value theory and graphical models has attracted quite some attention. \citet{GK18} introduced  max-linear models on directed acyclic graphs, whose dependence structures were further investigated in \citet{GKO18}. \citet{S20} studied the limit theory of extremes of regularly varying Markov trees, with applications to structured H\"{u}sler--Reiss distributions in \citet{A21} and with extensions to Markov random fields on block graphs in \citet{AS21}. In parallel, \citet{E20} proposed a way to factor the density of multivariate Pareto distribution on graphs through a variation on the Hammersley--Clifford theorem related to a kind of redefined conditional independence, followed by \citet{EV20}, who studied a data-driven methodology of identifying the underlying graphical structures for such a factorization. An illuminating review of sparse models for multivariate extremes is given in \citet{EI21}.

The goal of this paper is to approximate a potentially unknown multivariate max-stable distribution by another multivariate max-stable distribution constructed from a Markov random field with respect to a tree, or Markov tree in short. As special cases of graphical models, Markov trees have simple and tractable structures. The conditional independence relations they satisfy are equivalent to a certain factorization of their joint distribution and eventually lead to an effective dimension reduction. The application of Markov trees in the modelling of multivariate probability distribution goes back to the seminal paper by \citet{C68}, who approximate a $d$-dimensional discrete distribution by a product of pairwise distributions linked up by a tree. It was demonstrated that in many cases, this method provides a reasonable approximation to the full distribution, capturing at least the marginal distributions and some important features of the joint dependence structure. The idea remains attractive in modern practice. In \citet{H20}, for instance, such Markov tree approximations are used in the context of anomaly detection in machine learning.

A Markov tree involves a set of conditional independence relations between random variables. However, \citet{P16} have shown that for a max-stable random vector with positive continuous density, conditional independence of two random variables given the other ones implies unconditional independence of the two variables. Except for trivial cases, a random vector with a multivariate max-stable distribution and a continuous density can therefore not satisfy any conditional independence relations. Hence, the method of \citet{C68} cannot be applied directly to multivariate max-stable distributions. Still, \citet{EV20} and \citet{S20} demonstrated that conditional independence relations can be incorporated into distributions of multivariate extremes through the lens of excesses over high thresholds. Based on these findings, we propose a method based on graphical models that consists of merging bivariate max-stable distributions along a tree. 

Given a multivariate max-stable distribution and a tree on the set of variables, we construct a Markov tree by letting each connected pair of variables have the same distribution as the corresponding bivariate margin of the max-stable distribution. The Markov tree constructed in this way is shown to belong to the domain of attraction of another max-stable distribution. This second multivariate max-stable distribution serves as a tree-based approximation to the first one, fully determined by the tree structure and certain bivariate margins. Explicit expressions of such a \emph{tree-structured multivariate max-stable distribution} (see Definition~\ref{def:maxstabletree} below) are derived from those of joint tails of Markov trees. 

In practice, the true max-stable distribution is unknown and needs to be estimated from the data. To do so, we implement a three-step procedure: \mdf{(\romannumeral1)} Select the maximum dependence tree which has the maximal sum of edge weights, taken as the pairwise estimates of the upper tail dependence coefficient $\lambda$; \mdf{(\romannumeral2)} Estimate the dependence structures of all the pairs of adjacent variables on the maximum dependence tree; \mdf{(\romannumeral3)} Measure the goodness-of-fit of the constructed model by comparing the distributions of nonadjacent pairs of variables in the original max-stable law on the one hand and in its tree-based approximation on the other hand. It is worth noting that, in the second step, different max-stable parametric families can be assumed for different adjacent pairs and different estimation methods can be applied for the parameters. All in all, the three steps combined provide a flexible and computationally feasible way to fit a flexible class of tree-based multivariate max-stable distributions. 

Some theoretical results corresponding to the modelling procedure are also obtained. If the original multivariate max-stable distribution has a tree-based dependence structure, then with probability tending to one, the set of maximum dependence trees weighted by the upper tail dependence coefficients contains the true tree structure. Moreover, the joint estimator of the parameters of all the edge-wise bivariate dependence structures is still asymptotically normal, provided all the edge-wise estimators have a certain asymptotic expansion. 

We illustrate the procedure to the average daily discharges of rivers in the upper Danube basin, estimating the probability of extreme flooding in at least one of several gauging stations. Specifically, for $j=1,\ldots,d$, let $\xi_{j}$ be the discharge at location $j$ and let $u_{j}$ denote a given water level at station $j$. The probability of interest is 
\begin{align}
\label{eq:rareevent}
    \pr\left(\xi_{1}>u_{1}~\text{or}~\ldots~\text{or}~\xi_{d}>u_{d}\right).
\end{align}
Under the assumption that $\bm{\xi}=(\xi_{1},\ldots,\xi_{d})$ belongs to the domain of attraction of a multivariate max-stable distribution, the above probability is estimated based on a tree-structured approximation of the attractor. 

The paper is organized as follows. In Section~\ref{sec:preliminaries}, we give some preliminaries related to multivariate max-stable distributions and Markov trees. The model construction and its properties are discussed in Section~\ref{sec:model}. The statistical procedure is detailed in Section~\ref{sec:estim}. A simulation study is conducted in Section~\ref{sec:Sim} and a data analysis of the discharges of rivers in the upper Danube basin is given in Section~\ref{sec:app}. \mdf{All proofs, discussions and examples of stable tail dependence function estimators are postponed to Appendix~\ref{sec:SM}.  Some examples and additional illustrations related to the simulations and the case study are given in Sections S1--S3 of the Supplementary Material.}

\section{Preliminaries}
\label{sec:preliminaries}

\subsection{Multivariate extreme value theory}
\label{subsec:MEVT}

Put $V=\{1,\ldots,d\}$ for some integer $d\ge 2$. Let $\bm{\xi}^{i}=(\xi_{i,v},v\in V)$, $i=1,\ldots,n$, be an independent random sample from $\bm{\xi}=(\xi_{v},v\in V)$ with univariate marginal distribution functions $F_{v}$ for $v\in V$. We say $\bm{\xi}$ is in the (maximum) domain of attraction of $H$, notation $\bm{\xi}\in D(H)$, if there exist sequences $\bm{a}_{n}=(a_{n,v},v\in V)\in (0,\infty)^{d}$ and $\bm{b}_{n}=(b_{n,v},v\in V)\in \mathbb{R}^{d}$ such that
\begin{equation}
\label{eq:G_M}
    \lim_{n \to \infty} \pr\left(\max _{i=1, \ldots, n} \xi_{i, v} \leq a_{n,v} x_{v}+b_{n,v}, \, v \in V\right)=H(\bm{x}), \qquad \bm{x}=(x_{v},v\in V)\in\mathbb{R}^{d},
\end{equation}
for a limiting distribution $H$ with non-degenerate margins. The distribution $H$ is called a multivariate extreme value or max-stable distribution. An max-stable distribution is called simple if its margins are unit-Fr\'{e}chet, that is, $H_{v}(x)=\exp(-1/x)$ for $x\in (0,\infty)$ and $v\in V$. 
By monotone marginal transformations, we can convert any max-stable distribution into a simple one without changing its copula. This allows us to deal with the marginal distributions and the dependence structure separately. To focus on the latter, we often consider simple multivariate max-stable distributions. 
For more background on multivariate max-stable distributions, see for instance \citet[Chapter~5]{R87} and \citet[Chapter~8]{B04}.

A multivariate max-stable distribution function $H$ with general margins can be represented by
\begin{equation}
    \label{eq:Maxstable_stdf}
    H(\bm{x})=\exp\left\{-\ell\left(-\log H_{v}(x_{v}), \, v\in V\right)\right\},
\end{equation}
for $\bm{x} \in \mathbb{R}^d$ such that $H_{v}(x_{v}) > 0$ for all $v \in V$, and where $\ell : [0, \infty)^d \to [0, \infty)$ is the so-called stable tail dependence function of $H$ and $\bm{\xi}$, from whose distribution it can be recovered via
\begin{equation}
\label{eq:stdf}
\ell(\bm{y})=\lim_{t \to \infty} t\left\{1-\pr\left(F_{v}(\xi_{v}) \leq 1-t^{-1}y_{v}, \, v \in V \right)\right\}, \qquad \bm{y} \in [0, \infty)^d,
\end{equation}
see \citet{H92} or \citet{D98}. The dependence structure of a max-stable distribution is completely described by the stable tail dependence function. In the bivariate case, $\bm{\xi}=(\xi_{1},\xi_{2})$, this dependence function is closely related to the upper tail dependence coefficient defined by
\begin{equation}
\label{eq:TDC}
  \lambda=\lim_{t\to \infty} t\pr\left(1-F_{1}(\xi_{1})\leq t^{-1}, 1-F_{2}(\xi_{2})\leq t^{-1}\right),
\end{equation} 
which captures the extremal dependence along the main diagonal of the positive plane. The identity $\lambda=2-\ell(1,1)$ is clear by the definitions.

\subsection{Regularly varying Markov trees}
\label{subsec:RVMarkovTree}

A graph is a pair $(V,E)$, where $V$ is a finite set of nodes and $E\subseteq\{(a,b)\in V\times V: a \ne b\}$ is a set of paired nodes called edges.  
A graph is called undirected if $(a,b)\in E$ implies $(b,a)\in E$. 
A path from $a$ to $b$, with notation $\pth{a}{b}$, is a collection $\{(u_{0},u_{1}),(u_{1},u_{2}),\ldots, (u_{n-1},u_{n})\}$ of distinct and directed edges such that $u_{0}=a$ and $u_{n}=b$. An undirected graph is connected if for any pair of distinct nodes $a$ and $b$, there is a path from $a$ to $b$. A cycle is a path whose start and end nodes are the same. A tree is an undirected, connected graph without cycles. For any distinct nodes $a$ and $b$ on a tree, there is a unique path from $a$ to $b$. 

Let $\tree=(V,E)$ be a tree and let $\bm{Y}=(Y_{v},v\in V)$ be a $d$-variate random vector indexed by the nodes of $\tree$. For disjoint and nonempty subsets $A$, $B$ and $C$ of $V$, we say that $C$ separates $A$ from $B$ in $\tree$ if any path from a node in $A$ to a node in $B$ passes through at least one node in $C$. We call $\bm{Y}$ a Markov tree on $\tree$ if it satisfies the global Markov property, i.e., whenever $C$ separates $A$ from $B$, the conditional independence relation 
\begin{equation}
\label{eq:gmp}
   \bm{Y}_{A}\idp \bm{Y}_{B}\cdt \bm{Y}_{C} 
\end{equation}
holds, where $\bm{Y}_{A}=(Y_{a}, a\in A)$ for $A\subseteq V$ and we write $\bm{Y}_{V}$ as $\bm{Y}$ when $A=V$.

We first recall some results on Markov trees and multivariate max-stable distributions that provide theoretical support and motivation for our approach. The following condition is adapted from \citet{S20} and ensures the weak convergence of the conditional distribution of the tree given that it exceeds a high threshold at a fixed variable. Condition~\ref{condition}(i) is an assumption about the high level transformation. Condition~\ref{condition}(ii) controls the probability of extremes caused by non-extremes. 

\begin{condition}
\label{condition}
For a non-negative Markov tree $\bm{Y}=(Y_{v},v\in V)$ on the tree $\tree=(V,E)$:
\begin{enumerate}[(i)]
    \item  for every $e=(a,b)\in E$, there exists a version of the conditional distribution of $Y_{b}$ given $Y_{a}$ and a probability measure $\mu_{e}$ on $[0,\infty)$ such that for every $y\in [0,\infty)$ in which the distribution function $y \mapsto \mu_e([0,y])$ is continuous, we have
\begin{equation}
\label{eq:mu_e}
    \pr(y_{a}^{-1} Y_{b}\le x \mid Y_{a}=y_{a})\to\mu_{e}([0, x]),\qquad y_{a}\to\infty;
\end{equation}
\item for every edge $e = (a, b) \in E$ such that $\mu_e(\{0\}) > 0$, we have
\begin{equation*}
\forall \eta>0, \quad \lim_{\delta\downarrow 0} \limsup_{y_{a}\to\infty} \sup_{\varepsilon\in[0,\delta]}\pr(y_{a}^{-1} Y_{b} >\eta \mid Y_{a}=\varepsilon y_{a})=0.
\end{equation*}
\end{enumerate}
\end{condition}

Assume $\bm{Y}=(Y_{v},v\in V)$ is a non-negative Markov tree on $\tree=(V,E)$ satisfying Condition~\ref{condition}. If for every $u\in V$ we have
\begin{align}
\label{eq:RVMar}
t\pr(Y_{u}>t)\to 1, \qquad   t \to \infty,
\end{align}
then by Theorem~2 and Corollary~5 in \citet{S20}, we have
\begin{equation}
\label{eq:W_i}
    \pr(t^{-1}\bm{Y}\le \bm{x} \mid Y_{u}>t) \to \pr(\bm{W}_{u}\le \bm{x}), \qquad t\to\infty,
\end{equation}
in all continuity points $\bm{x} \in [0, \infty)^d$ of the distribution function of $\bm{W}_u$. The weak limit $\bm{W}_{u}$ is equal in distribution to $W\bm{\Theta}_{u}$, where $W$ is a unit-Pareto random variable that is independent of the random vector $\bm{\Theta}_{u} = (\Theta_{u,v},v\in V)$ called tail tree and given by $\Theta_{u,u}=1$ together with
\begin{equation*}
\Theta_{u,v}=\prod_{e\in \pth{u}{v} } M_{e}, \qquad  v\in V\setminus\{u\},
\end{equation*}
in terms of independent random variables $M_e$, for $e\in E$, with marginal distributions $\mu_e$ in \eqref{eq:mu_e}. In fact, the random vector $\bm{\Theta}_u$ is the weak limit of the conditional distribution of $t^{-1} \bm{Y}$ given that $Y_u = t$ as $t \to \infty$. The random vector $\bm{Y}$ is called a regularly varying Markov tree.
The distributions $\mu_{e}$ of the random variables $M_{e}$ can be arbitrary except for the fact that always $\E[M_{e}] \le 1$, see the proof of Proposition~\ref{pro:TDC_ieq}.

By Theorem~2 in \citet{S20}, the convergence in \eqref{eq:W_i} implies that $\bm{Y}$ is in the domain of attraction of a simple multivariate max-stable distribution, i.e., for independent random samples $\bm{Y}^{i}=(Y_{i,v},v\in V)$, $i=1,\ldots,n$, generated from $\bm{Y}$, the convergence \eqref{eq:G_M} holds with $a_{n,v}=n$ and $b_{n,v}=0$ for $v\in V$. To stress the importance of such max-stable distributions, which emerge as the weak limits of normalized component-wise maxima of independent copies of Markov trees, we state the following definition.

\begin{definition}
\label{def:maxstabletree}
A $d$-variate max-stable distribution $H$ is said to \emph{have a dependence structure linked to the tree} $\tree=(V,E)$ if, up to increasing marginal transformations, it contains in its domain of attraction a non-negative regularly varying Markov tree $\bm{Y}=(Y_{v},v\in V)$ on $\tree$ satisfying Condition~\ref{condition} and equation~\eqref{eq:RVMar}. We call $H$ a \emph{tree-structured max-stable distribution}.
\end{definition}

\begin{remark}
As kindly pointed out by the Associated Editor, there are two possible tree-based models for multivariate extremes. One is the max-stable attractor of the Markov random fields on a tree along the lines of~\citet{S20}, i.e., the tree-structured max-stable distribution in Definition~\ref{def:maxstabletree}. The other one is the multivariate Pareto distribution factorizing with respect to a tree via a redefined concept of conditional independence studied in \citet{E20}, with the extension to cases without densities provided in \citet{EV20}. In fact, these two models are related in the sense that the latter are the multivariate Pareto distributions associated to the tree-based max-stable distributions with respect to the same tree; see Appendix~\ref{sec:TSMS2ETM}.
\end{remark}

There could be several trees associated with the same max-stable distribution, see the asymptotic logistic distribution with parameter $\bm{\Psi}_3$ in Example~\ref{eg:AsyLogZero}. We show that the stable tail dependence function and upper tail dependence coefficients of a tree-structured max-stable distribution $H$ on a certain tree $\tree$ are given by the coming propositions.

\begin{proposition}
\label{pro:G_M}
Let $H$ be a tree-structured max-stable distribution as in Definition~\ref{def:maxstabletree} with respect to $\tree=(V, E)$. Let $(M_e, e \in E)$ be independent random variables with marginal distributions $\mu_e$ in \eqref{eq:mu_e}. The stable tail dependence function of $H$ is
\begin{equation}
\label{eq:stdf_G}
\ell(\bm{y})=\sum_{i=1}^{d} \E\left\{\max_{j=i,\ldots, d} \left(y_{v_{j}}\prod_{e\in \pth{v_i}{v_j}}M_{e}\right)-\max_{j=i+1,\ldots, d} \left( y_{v_{j}}\prod_{e\in \pth{v_i}{v_j}}M_{e}\right)\right\},
\end{equation}
where $v_{1},\ldots,v_{d}$ is an arbitrary permutation of $1,\ldots,d$. The maximum over the empty set is defined as zero and $M_{uv}$ is interpreted as one if $u=v$. 
If, in addition, we have $\E(M_{ab}) = 1$ for all $(a,b)\in E$, then
\begin{equation*}
\ell(\bm{y})
    = \E \left\{\max_{j=1,\ldots, d} \left( y_{v_{j}}\prod_{e\in \pth{v_1}{v_j}}M_{e}\right)\right\}
    = \E \left\{ 
        \max_{v=1,\ldots,d} \left(
         y_v \prod_{e \in \pth{1}{v}} M_e
    \right)\right\}.
\end{equation*}

Conversely, let $\tree=(V,E)$ be an undirected tree, if $\{M_{ab},(a,b)\in E\}$ is a sequence of non-negative and independent random variables with $\E( M_{ab})\le 1$ for all $(a,b)\in E$ and satisfying the relation
\begin{equation}
\label{eq:changeDir}
\pr(M_{ba}>x)=\E\left[M_{ab}\I\{M_{ab}<1/x\}\right], \quad x\ge 0,
\end{equation}
then $H(\bm{x})$ defined in~\eqref{eq:Maxstable_stdf} with $H_{v}$, $v\in V$, being univariate max-stable distributions and $\ell(\bm{x})$ given in~\eqref{eq:stdf_G} is a tree-structured max-stable distribution with respect to $\tree$.
\end{proposition}

The stable tail dependence function of the bivariate marginal distribution of a tree-structured max-stable distribution is formulated in Proposition~\ref{pro:bi_TDC} below. It stresses the fact that the distribution of the pair $(X_a,X_b)$, $a,b\in V$ and $a\neq b$, only depends on the independent random variables $M_e$, $e\in\pth{a}{b}$, which explains why we say that the distribution inherits the dependence structure from the tree.

\begin{proposition}
\label{pro:bi_TDC}
Let the random vector $\bm{X} = (X_{v}, v \in V)$ have distribution $H$ as in Proposition~\ref{pro:G_M}.
For distinct $a,b\in V$, the stable tail dependence function of $(X_{a},X_{b})$ is
\[
    \ell_{ab}(x,y)=x+y-\E \left\{\min\left(x, y \prod_{e\in \pth{a}{b} } M_{e}\right)\right\},\qquad x,y\ge 0,
\]
while its upper tail dependence coefficient is
\begin{equation}
\label{eq:TDC_GM}
    \lambda_{ab}=\E\left\{\min\left(1,\prod_{e\in \pth{a}{b} }M_{e}\right)\right\}.
\end{equation}
\end{proposition}

For a tree-structured max-stable distribution, the dependence between variables becomes weaker as the path connecting them becomes longer. This is confirmed by the next proposition. A similar result can be found in Proposition~5 of~\citet{EV20}, where an inequality related to the upper tail dependence coefficients of extremal tree models is presented. The right-hand side of our inequality in Proposition~\ref{pro:TDC_ieq} and the inequality in Proposition~5 of~\citet{EV20} can be recovered from each other in view of the relation of tree-structured max-stable distribution and the extremal tree models discussed in Section~\ref{sec:TSMS2ETM}. However, we point out that we also give a lower bound for the upper tail dependence coefficient as well as a sufficient condition for strict inequality on the right-hand side which is weaker than the conditions given in~\citet{EV20}.

\begin{proposition}
\label{pro:TDC_ieq}
Let $\bm{X}$ be the random vector in Proposition~\ref{pro:bi_TDC} and let $\lambda_{uv}$ be the upper tail dependence coefficient of $(X_{u},X_{v})$. For any distinct nodes $a,b,u\in V$ such that $u$ is on the path from $a$ to $b$, we have 
\begin{equation}
\label{eq:inequality}
\lambda_{au}\lambda_{ub} \leq \lambda_{ab}\leq \min(\lambda_{au}, \lambda_{ub}).
\end{equation}
A sufficient condition for the inequality on the right-hand side to be strict is that, for any $e \in E$ and any $\eps > 0$, we have $\pr(1-\eps < M_{e} < 1) > 0$ and $\pr(1 < M_{e} < 1+\eps) > 0$.
\end{proposition}

The inequality on the right-hand side of \eqref{eq:inequality} can be an equality, even in non-trivial situations. Let for instance $\eps_1, Y_2, \eps_3$ be independent unit-Fr\'{e}chet random variables and define $Y_j = \max(Y_2, \eps_j)/2$ for $j \in \{1, 3\}$. Then $Y_1$ and $Y_3$ are conditionally independent given $Y_2$ and $(Y_1, Y_2, Y_3)$ is a regularly varying Markov tree with respect to the tree 1--2--3. Still, the tail dependence coefficients are $\lambda_{12} = \lambda_{23} = \lambda_{13} = 1/2$. Note that $\pr(M_{12} = 0) = \pr(M_{12} = 2) = 1/2$ while $\pr(M_{23} = 1/2) = 1$, so that indeed $\E\{\min(1, M_{12} M_{23})\} = \E( M_{12} M_{23} ) = 1/2$ and $\E\{\min(1, M_{12})\} = 1/2$.

The left-hand side inequality of \eqref{eq:inequality} can be an equality too. For instance, if $(a, u)$ and $(u, b)$ are edges and if $M_{au}$ and $M_{ub}$ are bounded by $1$ almost surely, then $\lambda_{ab} = \E\{\min(1,M_{au}M_{ub})\} = \E(M_{au}M_{ub}) = \E(M_{au})\E(M_{ub}) = \E\{\min(1,M_{au})\} \E\{\min(1,M_{ub})\} = \lambda_{au}\lambda_{ub}$.

\begin{remark}
For any pair $a,b\in V$ of distinct nodes, Proposition~\ref{pro:TDC_ieq} and a recursive argument imply
\begin{equation*}
\prod_{e\in \pth{a}{b} } \lambda_{e} \leq \lambda_{ab}\leq  \min_{e \in \pth{a}{b}} \lambda_{e}.
\end{equation*}
\end{remark}

\section{Model definition and properties}
\label{sec:model}

\subsection{Model construction}
\label{subsec:modelconstruction}

Let $G$ be a simple $d$-variate max-stable distribution to be approximated by a tree-structured max-stable distribution on a tree $\tree = (V, E)$. As before, write $V=\{1,\ldots,d\}$. Assume $\bm{Z}=(Z_{v},v\in V)$ is a random vector with distribution $G$ having stable tail dependence function $\ell$. We construct a Markov tree, denoted by $\bm{Y}=(Y_{v},v\in V)$, on $\tree$ by the following definition.

\begin{definition}
\label{def:Y}
Given an arbitrary $d$-variate simple max-stable distribution $G$ and a tree $\tree=(V,E)$, define the Markov tree $\bm{Y}=(Y_{v},v\in V)$ by the following two properties:
\begin{enumerate}[(1)]
    \item the random vector $\bm{Y}$ satisfies the global Markov property~\eqref{eq:gmp} with respect to $\tree$;
    \item for each pair $(a,b)\in E$, the distribution of $(Y_{a},Y_{b})$ is the same as the one of $(Z_{a},Z_{b})$.
\end{enumerate}
We construct the tree-structured model $G_M$, serving as an approximation of $G$, by the composite map 
\[ (G, \tree) \mapsto P_{\bm{Y}} \mapsto G_M \]
defined via
\[
    G_M(x_1,\ldots,x_d) 
    = \lim_{n \to \infty} \pr(Y_1 \le nx_1, \ldots, Y_d \le nx_d)^n, 
    \qquad (x_1,\ldots,x_d) \in (0, \infty)^d.
\]
Here $P_{\bm{Y}}$ denotes the distribution of the Markov tree $\bm{Y}$ on $\tree$.

\end{definition}

\begin{remark}
For edges $(a, b) \in E$, the bivariate distributions of $G_M$ and $G$ are the same. As a consequence, if we apply the construction with $G_M$ as input, then we get $G_M$ again. In other words, if $G$ is already tree-structured along $\tree$, then $G_M = G$.
\end{remark}

\begin{remark}
The Markov tree $\bm{Y}$ defined in Definition~\ref{def:Y} is different from the multivariate Pareto distribution 
of \citet{EV20}. The former is in the domain of attraction of $G_M$ which is associated with a multivariate Pareto distribution 
with respect to the same tree, see Section~\ref{sec:TSMS2ETM}.
\end{remark}

\begin{remark}
A similar construction method has been used in~\citet{A21} as a motivation for the structured H\"{u}sler--Reiss distribution. The constructed model here is more general where, except for the global Markov property, we only require the distributions of adjacent pairs to be bivariate max-stable distributions, with a possibility to be different from pair to pair. This gives a broad class of structured multivariate max-stable distribution which includes the structured H\"{u}sler--Reiss model considered in~\citet{A21} as a special case.
\end{remark}

By the construction, the distribution of $(Y_{a},Y_{b})$ for $(a,b)\in E$ is bivariate max-stable with unit-Fr\'{e}chet margins. Thus Condition~\ref{condition} is satisfied and $\bm{Y}$ is a regularly varying Markov tree \citep[Example~3]{S20}. Hence the conclusions in Propositions~\ref{pro:G_M}--\ref{pro:TDC_ieq} 
hold for $\bm{Y}$ and thus $G_M$ always exists. 
The representation of $G_{M}$ follows directly from equation~\eqref{eq:Maxstable_stdf} and Proposition~\ref{pro:G_M}.

\begin{corollary}
\label{pro:GM_Markov}
The Markov tree $\bm{Y}$ in Definition~\ref{def:Y} is in the domain of attraction of a simple $d$-variate max-stable distribution $G_{M}$ with dependence structure linked to the tree $\tree$. More precisely, for $\bm{x}=(x_{v},v\in V) \in (0, \infty]^d$, we have
\begin{equation*}
G_{M}(\bm{x})=\exp\left\{-\ell^M(1/x_{v},v\in V)\right\}
\end{equation*}
with stable tail dependence function $\ell^M$ given by the right-hand side in \eqref{eq:stdf_G}.
\end{corollary}

For a given max-stable distribution $G$, every different tree $\tree$ on $V$ leads to a possibly different tree-structured approximation $G_M$. We will come back to the choice of $\tree$ in Subsection~\ref{subsec:LOTS}.

The distribution $G_{M}$ is determined by the increments $M_{e}$, see Corollary~\ref{pro:GM_Markov}, whose laws $\mu_e$ can be calculated through the limit in Condition~\ref{condition}(i) if $G$ is known. Alternatively, the distribution of $M_{e}$ for $e=(a,b)$ can also be obtained from the stable tail dependence function of $(Z_{a},Z_{b})$ 
analogously as in Proposition~3.9 of \citet{A21}:
\begin{equation}
    \label{eq:derivation}
    \pr(M_{ab}\leq x)=\frac{\d \ell_{ab}(x,1)} {\d x}, \qquad 0\leq x<\infty,
\end{equation}
where $\ell_{ab}$ is the stable tail dependence function of $(Z_{a},Z_{b})$ and where the derivative is taken from the right (which always exists by the convexity of $\ell_{ab}$).

The stable tail dependence functions $\ell$ and $\ell^M$ of $G$ and its tree-based approximation $G_M$, respectively, are different in general. Still, they share the same margins for pairs of adjacent nodes in $\tree$, i.e., $\ell_{ab} = \ell_{ab}^M$ for $(a, b) \in E$. The difference between $\ell$ and $\ell^M$ becomes visible when studying pairs $(a, b)$ not in $E$ or when considering higher-order margins.

\subsection{Measuring the approximation error}
\label{subsec:ApproErr}

Given the max-stable distribution $G$, it is theoretically possible to get an explicit expression of $G_{M}$ in Corollary~\ref{pro:GM_Markov}. But how well can $G$ be approximated by $G_{M}$? This question is important because, like in \citet{C68}, the Markov assumption is not believed to be satisfied exactly, especially not in high dimensions. Still, it may provide a reasonable approximation, and the question is how large or small the approximation error really is.

By construction, the adjacent pairs on the tree have precisely the same dependence structure as the corresponding pairs of the original max-table distribution. But for pairs of random variables which are nonadjacent on the tree, the tail dependence is in general distinct since the Markov tree has enclosed conditional independence in the model and this may influence the dependence structure of the limiting tree-based max-stable distribution. Therefore, differences can be expected to arise in the nonadjacent pairs on the tree. For $a, b \in V$, let $\lambda_{ab}$ and $\lambda_{ab}^M$ denote the upper tail dependence coefficients of the variables with indices $a$ and $b$ in $G$ and in $G_M$, respectively. For the reason just explained, we focus on $a,b\in V$ with $a \ne b$ such that $(a,b)\notin E$ and quantify the approximation error by
\begin{equation}
\label{eq:D}
D_{\tree}=\sum_{a,b\in V,a\neq b,(a,b)\notin E}  \left|\lambda_{ab}^{M}-\lambda_{ab}\right|,
\end{equation} 
where $\left|\lambda_{ab}^{M}-\lambda_{ab}\right|$ and $\left|\lambda_{ba}^{M}-\lambda_{ba}\right|$ only count once in the summation. Note that $D_{\tree}$ also equals $\sum_{a,b\in V}|\lambda_{ab}^{M}-\lambda_{ab}|$ since $\lambda_{ab}^{M}=\lambda_{ab}$ for $(a,b)\in E$ or when $a = b$. The quantity $D_{\tree}$ can be interpreted as the cumulative difference in bivariate tail dependence between the two max-stable distributions.

If $G$ is already tree-structured along $\tree$, then $G_M=G$ and thus $D_{\tree}=0$. The converse is not true in general: even if $G_M\neq G$, it could still occur that $D_{\tree}=0$. A counterexample can be found in \mdf{Example~S1 in Section~S1 of the Supplementary Material}, where two different distributions have the same bivariate stable tail dependence functions and thus $D_{\tree}=0$. However, within parametric models, examples can also be found for which the converse statement holds, i.e., $G=G_{M}$ as soon as $D_{\tree}=0$, see the H\"{u}sler--Reiss distribution and the asymmetric logistic distributions given in the next section. Since $G_{M}$ is a tree-based model induced by the Markov tree which encodes conditional independence between nonadjacent pairs, an intuitive conjecture could be that $\lambda_{ab}^{M}$ is smaller than $\lambda_{ab}$ for $(a,b)\notin E$. But this is not true as shown by the H\"{u}sler--Reiss distribution in Example~\ref{eg:HR} with parameter $\bm{\Gamma}_2$: if $G_M$ is constructed based on tree 1--2--3--4 with structure $\tree^{(a)}$, then $\gamma_{14}=16$ and $\gamma^M_{14}=\gamma_{12}+\gamma_{23}+\gamma_{34}=12$, leading to $\lambda_{14}<\lambda_{14}^M$ by~\eqref{eq:TDCHR} and thus contradicting with the conjecture.

Given a max-stable distribution $G$, one could wonder which tree $\tree$ yields a minimal approximation error $D_{\tree}$ and thus a kind of best possible tree approximation $G_M$. A conjecture would be that the approximation is best for a maximum-spanning tree of $G$ with respect to the tail dependence coefficients as edge weights. Indeed, if $G = G_M$, then we know from Proposition~\ref{pro:uniqueness} below that the tree that minimizes $D_{\tree}$ must be one of such maximum-spanning trees, but a maximum-spanning tree with tail dependence coefficients as edge weights does not necessarily minimize $D_{\tree}$. For example see the asymmetric logistic distribution with parameter $\bm{\Psi}_3$ in Example~\ref{eg:AsyLogZero} and Table~\ref{tab:4HR_SDH}: the tree 3--1--2--4 with line structure $\tree^{(a)}$ is a maximum-spanning tree but it does not have the minimum $D_{\tree}$. 
The conjecture seems less likely for a general max-stable distribution $G$ as falsified by the $4$-variate H\"{u}sler--Reiss distribution with variogram matrix $\bm{\Gamma}_{3}$ in Example~\ref{eg:HR} and Table~\ref{tab:4HR_SDH}, where $D_{\tree}$ attains its minimum value $0.156$ at the tree 2--3--1--4 with line tree structure $\tree^{(a)}$ but the sum of edge weights takes its maximum value $1.087$ at the tree 3--1--2--4 based on $\tree^{(a)}$.

Another measure of the approximation error, suggested by the associate editor, is:
\[
    \tilde{D}_{\tree} = \sum_{a,b \in V; a \ne b; (a,b)\notin E} \sup_{(x_a,x_b) \in [0, 1]^2} \left| \ell^M_{ab}(x_a,x_b) - \ell_{ab}(x_a,x_b) \right|
\]
with $\ell^M_{ab}$ the $(a,b)$ bivariate marginal of the tree-structured stable tail dependence function $\ell^M$. Since $\lambda_{ab} = 2 - \ell_{ab}(1,1)$ and $\lambda_{ab}^M = 2 - \ell^M_{ab}(1, 1)$, we have $\tilde{D}_{\tree} \ge D_{\tree}$. If $G$ is already tree-structured along $\tree$, then $G_M = G$ and $\tilde{D}_{\tree} = 0$. The converse implication does not hold either, see again \mdf{Example~S1 in the Supplementary Material}.

\subsection{Examples}
\label{sec:Examples}
We explain the construction through two examples. Throughout the rest of this paper, the symbols $G$, $G_{M}$, $\bm{Y}$ and $\bm{Z}$ are defined in the same way as in Subsection~\ref{subsec:modelconstruction}. Further, $\bm{Z}_{M} = (Z_{v}^{M}, v \in V)$ denotes a random vector with distribution function $G_M$. For nonempty $U\subseteq V$, let the stable tail dependence functions of $\bm{Z}_{U} = (Z_{v},v\in U)$ and $\bm{Z}_{U}^{M} = (Z_{v}^{M},v\in U)$ be denoted by $\ell_{U}$ and $\ell_{U}^{M}$, respectively, where we write $\ell_{U}$ as $\ell_{ab}$ when $U=\{a,b\}$ and where the subscript $U$ is omitted when $U=V$.

\begin{example}[H\"{u}sler--Reiss model]
\label{eg:HR}
Assume $\bm{\Gamma}=(\gamma_{uv})_{u,v\in V}$ is a $d\times d$ dimensional symmetric, conditionally negative definite matrix with elements $\gamma_{uv}\in [0,\infty)$ satisfying $\gamma_{uu}=0$ for $u \in V$, i.e., $\bm{a}^\top \Gamma \bm{a}< 0$ for any nonzero $d$-dimensional column vector $\bm{a}=(a_{v},v\in V)$ such that $\sum_{v\in V}a_{v}=0$. For $\bm{x}\in (0,\infty)^{d}$, let $G(\bm{x})$ be a H\"{u}sler--Reiss distribution \citep{H89} defined by 
\begin{equation*}
G(\bm{x})=\exp \left\{-\sum_{u\in V} x_{u}^{-1} \Phi_{d-1}\left(\log \frac{x_{v}}{x_{u}}+ \frac{\gamma_{uv}}{2}, v \in V\setminus\{u\}; \bm{\Sigma}_{V,u}\right)\right\},
\end{equation*}
where $\bm{\Sigma}_{V,u}=(\sigma_{ab}^{2})_{a,b\in V\setminus \{u\}}$ is the $(d-1)\times (d-1)$ dimensional square matrix with elements $\sigma_{ab}^{2}=(\gamma_{au}+\gamma_{bu}-\gamma_{ab})/2$. The function $\Phi_{p}\left(\bm{x}; \bm{\Sigma} \right)$ is the $p$-dimensional normal cumulative distribution function with zero mean and covariance matrix $\bm{\Sigma}$, with subscript omitted when $p=1$. For nonempty  $U\subseteq V$, the stable tail dependence function of $(Z_{v},v\in U)$ is
\begin{equation*}
\ell_{U}(\bm{x}_{U})=\sum_{u\in U} x_{u} \Phi_{|U\setminus{\{u\}}|}\left(\log \frac{x_{u}}{x_{v}}+ \frac{\gamma_{uv}}{2}, v \in  V\setminus\{u\}; \bm{\Sigma}_{U,u}\right),\qquad \bm{x}_{U}=(x_{u},u\in U)\in [0,\infty)^{|U|},
\end{equation*}
where $|U|$ denotes the number of elements in $U$. Consequently, we have
\begin{equation}
\label{eq:TDCHR}
\lambda_{ab}=2\left\{1-\Phi\left(\sqrt{\gamma_{ab}}/2\right)\right\},\qquad a,b\in V.
\end{equation}

From Proposition~2.1 of \citet{A21}, $G_{M}$ is also a H\"{u}sler--Reiss extreme value distribution parameterized by the variogram matrix $\bm{\Gamma}_{M}=(\gamma^{M}_{uv})_{u,v\in V}$ with elements
\begin{align}
\label{eq:Gamma_M}
    \gamma^{M}_{uv} = \sum_{e\in \pth{u}{v}} \gamma_{e}.
\end{align}
Thus $\lambda_{ab}^{M}=\lambda_{ab}$ for $(a,b)\in E$ and 
\begin{equation*}
    \lambda_{ab}^{M}=2\left\{1-\Phi\left(2^{-1}\sqrt{\gamma_{ab}^M}\right)\right\}, 
    \qquad (a,b)\notin E. 
\end{equation*}

Consider the $4$-dimensional H\"{u}sler--Reiss max-stable distributions with variogram matrices
\begin{align}
\label{eq:Gamma123}
\bm{\Gamma}_{1}&=
\left(
  \begin{array}{cccc}
    0 & 4 & 4 & 4 \\
    4 & 0 & 8 & 8 \\
    4 & 8 & 0 & 8 \\
    4 & 8 & 8 & 0 \\
  \end{array}
\right),&
\bm{\Gamma}_{2}&=
\left(
  \begin{array}{cccc}
    0 & 4 & 8 & 16 \\
    4 & 0 & 4 & 8 \\
    8 & 4 & 0 & 4 \\
    16 & 8 & 4 & 0 \\
  \end{array}
\right),&
\bm{\Gamma}_{3}&=
\left(
  \begin{array}{cccc}
    0 & 11 & 10 & 15 \\
    11 & 0 & 10 & 4 \\
    10 & 10 & 0 & 10 \\
    15 & 4 & 10 & 0 \\
  \end{array}
\right).
\end{align}
For a tree having $4$ nodes, two types of tree structures, $\tree^{(a)}$ and $\tree^{(b)}$ in Figure~\ref{fig:treestructure}, can be used to create Markov trees. 

\begin{figure}
    \centering
    \begin{tikzpicture}[> = stealth, 
 	shorten > = 0.5pt, 
 	auto,
 	node distance = 3cm,
 	semithick 
 	,scale=0.8,auto=left]
 	\node[circle,fill=cyan!30] (n1) at (-2,-1)	{$v_{a}$};
 	\node[circle,fill=cyan!30] (n2) at (0,-1)  	{$v_{b}$};
 	\node[circle,fill=cyan!30] (n3) at (2,-1) 	{$v_{c}$};
 	\node[circle,fill=cyan!30] (n4) at (4,-1) 	{$v_{d}$};
 	\draw (n1)--(n2);
 	\draw (n2)--(n3);
 	\draw (n3)--(n4);

 	\node[circle,fill=cyan!30] (n1) at (6,0)	{$v_{b}$};
 	\node[circle,fill=cyan!30] (n2) at (8,0)  	{$v_{a}$};
 	\node[circle,fill=cyan!30] (n3) at (10,0) 	{$v_{c}$};
 	\node[circle,fill=cyan!30] (n4) at (8,-2) 	{$v_{d}$};
 	\draw (n1)--(n2);
 	\draw (n2)--(n3);
 	\draw (n2)--(n4);
 	\end{tikzpicture}
    \caption{Structures $\tree^{(a)}$ (left) and $\tree^{(b)}$ (right) for a tree with $4$ nodes.}
    \label{fig:treestructure}
\end{figure}

Given a tree, we can construct Markov trees and obtain the variogram matrix of $G_{M}$ through \eqref{eq:Gamma_M}. Table~\ref{tab:4HR_SDH} presents the value of $D_{\tree}$ defined in \eqref{eq:D} of the constructed models corresponding to different tree structures. As was expected, the value of $D_{\tree}$ is zero if $G$ is a tree-structured max-stable distribution itself and the model $G_{M}$ is constructed based on the real tree structure of $G$: this is the case for the H\"{u}sler--Reiss distribution with variogram matrix $\bm{\Gamma}_{1}$. Actually, although it is not the case in general, the converse statement also holds for H\"{u}sler--Reiss distributions: if $G$ is a H\"{u}sler--Reiss distribution and $D_{\tree}=0$, then we have $G=G_{M}$. To see this, note that the variogram matrix of a H\"{u}sler--Reiss distribution can be uniquely recovered from its pairwise upper tail dependence coefficients through~\eqref{eq:TDCHR}. Since $G_M$ is also a H\"{u}sler--Reiss distribution and the assumption $D_{\tree}=0$ implies $\lambda_{ab}=\lambda_{ab}^M$ for all $a,b\in V$, the distributions $G$ and $G_M$ share the same variogram matrix and are therefore equal.

For a non-tree structured H\"{u}sler--Reiss distribution, the tree which minimizes $D_{\tree}$ does not necessarily have the maximal sum of upper tail dependence coefficients along its edges, see the 4-variate H\"{u}sler--Reiss distribution with variogram matrix $\Gamma_3$. Still, we can show that for an arbitrary trivariate H\"{u}sler--Reiss distribution $G$, the set of maximum spanning trees with upper tail dependence coefficients as edge weights (see the definition of maximum dependence tree in~\eqref{eq:tree_star}) is exactly the set of trees which minimize the approximation error $D_{\tree}$. For more details, see the discussion in \mdf{Section~S1 of the Supplementary Material}.
\end{example}

\begin{example}[Asymmetric logistic model]
\label{eg:AsyLogZero}
We consider a special case of the asymmetric logistic model of \citet{T90}. For a parameter vector $\bm{\Psi}=(\psi_{v}, v\in V) \in [0, 1]^d$, let
\begin{equation}
\label{eq:AsyLogistDF}
G(\bm{x})=\exp\left\{-\sum_{u\in V}\frac{1-\psi_{u}}{x_{u}}-\max_{u\in V}\left(\frac{\psi_{u}}{x_{u}}\right)\right\},\qquad \bm{x}\in (0,\infty]^{d}.
\end{equation}
It appears as the distribution function of the max-linear model $\bm{Z}=(Z_{v},v\in V)$ with
\begin{equation}
    \label{eq:ML}
    Z_{v} = \max \left\{ \psi_{v} \eps, (1-\psi_{v}) \eps_{v} \right\},
\end{equation}
where $\eps$ and $\eps_{v}$, for $v\in V$, are independent unit-Fr\'{e}chet random variables. For $U\subseteq{V}$, the stable tail dependence function of $(Z_{v},v\in U)$ is
\begin{equation*}
    \ell_{U}(\bm{x}_{U})=\sum_{u\in U}(1-\psi_{u})x_{u}+\max_{u\in U}\left(\psi_{u} x_{u}\right),
\end{equation*}
which implies that
\begin{equation}
\label{eq:TDCAL}
    \lambda_{ab}=\min(\psi_{a},\psi_{b}), \qquad a,b\in V, \quad a\neq b.
\end{equation}

By \eqref{eq:derivation}, the distribution of $M_{ab}$ for $(a, b) \in E$ is discrete with at most two atoms:
\begin{align}
\label{eq:increase_ML}
\pr(M_{ab}=0)&=1-\psi_{a}, & \pr\left(M_{ab}=\frac{\psi_{b}}{\psi_{a}}\right)&=\psi_{a}.
\end{align}
Write $V_{i}=\{i,\ldots,d\}$ and $e=(e_{p},e_{s})$ for $e\in E$. 
By Corollary~\ref{pro:GM_Markov} and the maximum--minimums identity, the stable tail dependence function of $G_{M}$ is 
\begin{align}
\label{eq:stdf_G_M_ML}
\nonumber
\ell^{M}(\bm{y})
=& \sum_{i=1}^{d} \psi_{i}^{-1} \sum_{\emptyset \ne U\subseteq V_{i}} (-1)^{(|U|+1)} \left(\prod_{e\in \bigcup_{j\in U}\pth{i}{j}} \psi_{e_{p}}\right)  \left\{\min_{j\in U}\left(y_{j} \psi_{j}\right)\right\}\\
&\mbox{} - \sum_{i=2}^{d} \psi_{i-1}^{-1} \sum_{\emptyset \ne U\subseteq V_{i}} (-1)^{(|U|+1)} \left(\prod_{e\in \bigcup_{j\in U}\pth{(i-1)}{j}} \psi_{e_{p}}\right) \left\{\min_{j\in U}\left(y_{j} \psi_{j}\right)\right\};
\end{align}
see Appendix~\ref{subsec:LemProf} for details. Moreover, for $a,b\in V$, $a\neq b$ and $(a,b)\notin E$, 
it follows from Proposition~\ref{pro:bi_TDC} and the distribution of $M_{e}$ for $e\in E$ in~\eqref{eq:increase_ML} that
\begin{equation}
\label{eq:TDC_G_M_ML}
\lambda_{ab}^{M}=\E\left\{\min\left(1,\prod_{e\in \pth{a}{b}}M_{e}\right)\right\}
=\min\left(\prod_{e\in \pth{a}{b}}\psi_{e_{p}}, \prod_{e\in \pth{a}{b}}\psi_{e_{s}}\right).
\end{equation}
For the $4$-variate asymmetric logistic models with parameter vectors $\bm{\Psi}_{1}=(0.8,0.7,0.4,0.2)$, $\bm{\Psi}_{2}=(0.5,0.4,0.3,0.2)$ and $\bm{\Psi}_{3}=(0.5,1,1,0.2)$, the numeric results of $D_{\tree}$ in \eqref{eq:D} 
for different trees are given in Table~\ref{tab:4HR_SDH}.

As shown in \mdf{Lemma~S2 of the Supplementary Material}, if $G$ is an asymmetric logistic distribution in~\eqref{eq:AsyLogistDF}, then we have $D_{\tree}=0$ if and only if $G=G_M$. Moreover, for a general asymmetric logistic distribution, the star-shaped tree, denoted by $\tree^{\circledast}$, whose central node $v$ is the one having the largest coefficient $\psi_v$ is one of the maximum tail dependence trees and is also a tree that minimizes $D_{\tree}$. Note that in this example $\tree^{\circledast}$ is not necessarily the unique tree minimizing $D_{\tree}$ and the set of trees maximizing the sum of upper tail dependence coefficients along their edges is not a singleton in general either, as shown by the distribution with parameter vector $\bm{\Psi}_3$ (see Table~\ref{tab:4HR_SDH}). But for an arbitrary asymmetric logistic distribution $G$ given by~\eqref{eq:AsyLogistDF}, the tree minimizes $D_{\tree}$ must be one of the maximum-spanning tree with upper tail dependence coefficients as edge weights. A detailed discussion can be found in \mdf{Section~S1 of the Supplementary Material}.

 \begin{sidewaystable}[!htbp] \centering 
 \caption{The values of $D_{\tree}$ in \eqref{eq:D} and $S_{\tree_{\lambda}}$ in \eqref{eq:tree_star} for the $4$-variate H\"{u}sler--Reiss model in Example~\ref{eg:HR} and the $4$-variate asymmetric logistic model in Example~\ref{eg:AsyLogZero} based on trees with structures $\tree^{(a)}$ and $\tree^{(b)}$ in Figure~\ref{fig:treestructure}, where $(i,j,k,l)$ in the second column denotes the tree $\tree$ with $v_a=i$, $v_b=j$, $v_c=k$ and $v_d=l$ in tree structure $\tree^{(a)}$ or $\tree^{(b)}$.}
   \label{tab:4HR_SDH} 
 \begin{tabularx}{0.9\textwidth}{ccXccccccXcccccc} 
\toprule
  \multirow{3}{*}{Structure} & \multirow{3}{*}{Trees} &&
  \multicolumn{6}{c}{H\"{u}sler--Reiss} && \multicolumn{6}{c}{Asymmetric logistic}\cr
  \cmidrule(lr){4-9} \cmidrule(lr){11-16}
   & &&
  \multicolumn{2}{c}{$\bm{\Gamma}_{1}$} & \multicolumn{2}{c}{$\bm{\Gamma}_{2}$} & \multicolumn{2}{c}{$\bm{\Gamma}_{3}$} && \multicolumn{2}{c}{$\bm{\Psi}_{1}$} & \multicolumn{2}{c}{$\bm{\Psi}_{2}$}  & \multicolumn{2}{c}{$\bm{\Psi}_{3}$} \cr 
  \cmidrule(lr){4-5}  \cmidrule(lr){6-7} \cmidrule(lr){8-9} \cmidrule(lr){11-12} \cmidrule(lr){13-14}  \cmidrule(lr){15-16} 
  & && $S_{\tree_{\lambda}}$ & $D_{\tree}$ & $S_{\tree_{\lambda}}$ & $D_{\tree}$  & $S_{\tree_{\lambda}}$ & $D_{\tree}$ && $S_{\tree_{\lambda}}$ & $D_{\tree}$ & $S_{\tree_{\lambda}}$ & $D_{\tree}$ & $S_{\tree_{\lambda}}$ & $D_{\tree}$\cr
  \midrule
 \multirow{12}{*}{$\tree^{(a)}$} 
 & $(1,2,3,4)$ && $0.632$ & $0.638$ & $0.952$ & $0.038$ & $0.686$ & $0.194$ && $1.3$ & $0.376$ & $0.9$ & $0.490$ & $1.7$  &  $0$  \\
 & $(1,2,4,3)$ && $0.632$ & $0.638$ & $0.792$ & $0.384$ & $0.553$ & $0.200$ && $1.1$ & $0.716$ & $0.8$ & $0.630$ & $0.9$  &  $1.2$  \\
 & $(1,3,2,4)$ && $0.632$ & $0.638$ & $0.632$ & $0.488$ & $0.138$ & $0.699$ && $1.0$ & $0.568$ & $0.8$ & $0.560$ & $1.7$  &  $0$  \\
 & $(1,3,4,2)$ && $0.632$ & $0.638$ & $0.632$ & $0.564$ & $0.268$ & $0.514$ && $0.8$ & $1.076$ & $0.7$ & $0.750$ & $0.9$  &  $1.2$  \\
 & $(1,4,2,3)$ && $0.632$ & $0.638$ & $0.520$ & $0.686$ & $0.300$ & $0.469$ && $0.8$ & $0.908$ & $0.7$ & $0.700$ & $1.4$  &  $0.8$  \\
 & $(1,4,3,2)$ && $0.632$ & $0.638$ & $0.680$ & $0.435$ & $0.242$ & $0.441$ && $0.8$ & $1.076$ & $0.7$ & $0.750$ & $1.4$  &  $0.8$  \\
 & $(2,1,3,4)$ && $0.792$ & $0.346$ & $0.792$ & $0.384$ & $0.287$ & $0.313$ && $1.3$ & $0.344$ & $0.9$ & $0.466$ & $1.2$  &  $0.6$ \\ 
 & $(2,1,4,3)$ && $0.792$ & $0.346$ & $0.680$ & $0.567$ & $0.288$ & $0.182$ && $1.1$ & $0.704$ & $0.8$ & $0.616$ & $0.9$  &  $1.4$ \\
 & $(2,3,1,4)$ && $0.792$ & $0.346$ & $0.520$ & $0.686$ & $0.865$ & $0.156$ && $1.0$ & $0.548$ & $0.8$ & $0.540$ & $1.7$  &  $0.2$ \\
 & $(2,4,1,3)$ && $0.792$ & $0.346$ & $0.360$ & $0.919$ & $0.333$ & $0.220$ && $0.8$ & $0.888$ & $0.7$ & $0.680$ & $0.9$  &  $1.4$ \\
 & $(3,2,1,4)$ && $0.792$ & $0.346$ & $0.680$ & $0.435$ & $0.230$ & $0.862$ && $1.3$ & $0.296$ & $0.9$ & $0.450$ & $1.7$  &  $0.2$\\
 & $(3,1,2,4)$ && $0.792$ & $0.346$ & $0.632$ & $0.564$ & $1.087$ & $0.513$ && $1.3$ & $0.284$ & $0.9$ & $0.446$ & $1.2$  &  $0.6$ \\
\midrule
 \multirow{4}{*}{$\tree^{(b)}$} 
 & $(1,2,3,4)$ && $0.952$ & $0$     & $0.520$ & $0.669$ & $0.549$ & $0.245$ && $1.3$ & $0.160$ & $0.9$ & $0.350$ & $1.2$  &   $0.7$ \\
 & $(2,3,4,1)$ && $0.632$ & $0.580$ & $0.792$ & $0.272$ & $0.269$ & $1.007$ && $1.3$ & $0.240$ & $0.9$ & $0.420$ & $1.7$  &  $0$ \\ 
 & $(3,4,1,2)$ && $0.632$ & $0.580$ & $0.792$ & $0.272$ & $0.828$ & $0.582$ && $1.0$ & $0.660$ & $0.8$ & $0.560$ & $1.7$  &  $0$  \\
 & $(4,1,2,3)$ && $0.632$ & $0.580$ & $0.520$ & $0.669$ & $0.357$ & $0.479$ && $0.6$ & $1.200$ & $0.6$ & $0.800$ & $0.6$  &  $1.6$  \\
 \bottomrule
 \end{tabularx} 
 \end{sidewaystable} 
\end{example}

\section{Estimation and modelling}
\label{sec:estim}

In this section, we assume the data are drawn independently from a distribution that, up to increasing marginal transformations, is attracted by a simple max-stable distribution $G$, not necessarily tree-structured. Given such data, we propose a three-step procedure to construct a tree-structured dependence model for extremes. The three steps were given in the introduction and are detailed below.

\subsection{Structure learning for the Markov tree}
\label{subsec:LOTS}

To construct the model, a tree structure is required. So the first crucial step is to choose an appropriate tree structure from all possible candidates. In the framework of structure learning for graphical models, a commonly used way is to give each edge a weight in the complete graph and then choose the tree with maximum or minimum sum of weights along its edges; see for instance \citet{C68}, \citet{H20} and \citet{EV20}. We follow this method and try to find the tree structure which can retain the dependence of $G$ as much as possible. 

An intuitive way to define the edge weight of a pair of variables is by the value of a bivariate dependence measure.
Here we consider the upper tail dependence coefficient $\lambda$ defined in \eqref{eq:TDC}.
Another possible edge weight, advocated in \citet{EV20}, is the extremal variogram. In that paper, the upper tail dependence coefficient is used as edge weight as well but under the name `extremal correlation': they propose to find the minimal spanning tree with respect to $\rho = -\log(\lambda)$. This corresponds to maximizing $\sum_{(a,b) \in E} \log(\lambda_{ab}) = \log \left( \prod_{(a,b) \in E} \lambda_{ab} \right)$ and thus the product of the tail dependence coefficients rather than their sum which we used in the following.
Note that, in contrast to \citet{EV20}, we do not assume that $G$ is tree-structured.

Let $\omega_{ab}$ denote the weight associated to $(a,b) \in V \times V$
. We define the set of maximum (tail) dependence trees by 
 \begin{equation}
 \label{eq:tree_star}
    \bm{T}_{\omega}^{\star} = \argmax_{\tree} S_{\tree_{\omega}} \quad 
    \text{where} ~ S_{\tree_{\omega}} = \sum_{(a, b) \in E} \omega_{ab} ~ \text{for} ~ \tree = (V, E),
\end{equation}   
the maximum being taken over all possible trees over the set $V$, 
and with $\omega_{ab}=\lambda_{ab}$ the tail dependence coefficient
of the data-generating distribution for the pair of variables indexed by $(a, b) \in V^2$. Similarly to the definition of $D_{\tree}$, only one of $\omega_{ab}$ and $\omega_{ba}$ for $(a,b)\in E$ is added to $S_{\tree_{\omega}}$ since they are equal. The notation $\tree_{\omega}^{\star}$ will be used for a generic member of the set $\bm{T}_{\omega}^{\star}$, which will often be a singleton.

We clarify again that we do not assume $G$ is tree-structured throughout the construction of our model. However, for the special case where $G$ is a tree-based max-stable distribution, following the same line of proof for Proposition 5 in~\citet{EV20}, a result similar to that proposition can be obtained based on the inequality in Proposition~\ref{pro:TDC_ieq}. If $G$ has a dependence structure linked to a tree $\tree$, then $\tree$ belongs to the set of maximum dependence trees $\bm{T}_{\lambda}^{\star}$ weighted by the upper tail dependence coefficients. If, moreover, the latter set is a singleton, then $\tree$ can be identified as the unique maximum dependence tree. A sufficient condition is that the inequality on the right-hand side of Proposition~\ref{pro:TDC_ieq} is strict.

\begin{proposition}
\label{pro:uniqueness}
Let $G$ be a max-stable distribution with dependence structure linked to a tree $\tree=(V,E)$ as in Definition~\ref{def:maxstabletree}. Then the tree $\tree$ belongs to $\bm{T}_{\lambda}^{\star}$ in \eqref{eq:tree_star} with edge weights equal to the bivariate upper tail dependence coefficients of $G$, that is, $\sum_{(a, b) \in E} \lambda_{ab} \ge \sum_{(a', b') \in E'} \lambda_{a'b'}$ for any tree $\tree' = (V, E')$. The set $\bm{T}_{\lambda}^{\star}$ is a singleton and thus equal to $\{ \tree \}$ if for every triple of distinct nodes $a, u, b \in V$ with $u$ on the path between $a$ and $b$ we have $\lambda_{ab} < \min(\lambda_{au}, \lambda_{ub})$; in this case, $\tree$ is the unique tree corresponding to $G$.
\end{proposition}

By way of example, consider the $4$-variate H\"{u}sler--Reiss distribution $G$ with variogram matrix $\bm{\Gamma}_{1}$ in \eqref{eq:Gamma123}. This $G$ is structured along the tree $\tree^{(b)}$ on the right-hand side of Figure~\ref{fig:treestructure} with $(a,b,c,d)=(1,2,3,4)$. According to Table~\ref{tab:4HR_SDH}, the value of $S_{\tree_{\lambda}}$ is maximized uniquely at the true tree, with value $0.952$. In general, however, the collection of trees $\bm{T}_{\lambda}^{\star}$ in Proposition~\ref{pro:uniqueness} need not be singleton, not even if $0 < \lambda_{ab} < 1$ for all $(a, b) \in V \times V$ with $a \ne b$. For a counter-example, see the trivariate regularly varying Markov tree after Proposition~\ref{pro:TDC_ieq}, for which all pairwise tail dependence coefficients are equal to $1/2$ and thus all trees have the same sum of edge weights.

Given an arbitrary max-stable distribution $G$, with relation to the discussion in the fourth paragraph of Section~\ref{subsec:ApproErr}, we would also wonder whether the converse side is true: the tree $\tree$ that minimizes $D_{\tree}$, does it necessarily belong to $\bm{T}_\lambda^{\star}$? We know this is the case for a tree-structured $G$ from Proposition~\ref{pro:uniqueness}. However, the answer is no in general as illustrated by the $4$-variate H\"{u}sler--Reiss distribution with variogram matrix $\bm{\Gamma}_3$ in Example~\ref{eg:HR} and Table~\ref{tab:4HR_SDH}: the tree 2--3--1--4 with line tree structure $\tree^{(a)}$ minimizes $D_{\tree}$ but does not belong to $\bm{T}_\lambda^{\star}$. But as we have discussed in Section~\ref{sec:Examples} (see also \mdf{Lemma~S1 and~S3 in Section~S1 of the Supplementary Material}), the trivariate H\"{u}sler--Reiss distributions and all the asymptotic logistic models are examples where this property holds. Somehow, it provides a kind of support and motivation for us to use the maximum dependence tree as the basis for model construction.

Consider again the general case where $G$ is not necessarily tree-structured.
In practice, the true max-stable distribution $G$ is unknown and the edge weights must be estimated from the data. Assume we observe independent copies $\bm{\xi}^{i}=(\xi_{i,v},v\in V)$, $i=1,\ldots, n$, of the $d$-variate random vector $\bm{\xi}=(\xi_{v},v\in V)$ with marginal distributions $F_{v}$ for $v\in V$ such that the standardized vector $(1/(1-F_{v}(\xi_{v})), v \in V)$ belongs to the max-domain of attraction of $G$. Equivalently, we have relation \eqref{eq:stdf} for $\bm{\xi}$ with $\ell$ equal to the stable tail dependence function of the random vector $\bm{Z}$ with distribution function $G$. Note that $\bm{\xi}$ does not necessarily belong to the domain of attraction of a multivariate max-stable distribution since the margins are not required to be attracted by univariate extreme value distributions. 
From Corollary~\ref{pro:G_M}, the Markov tree $\bm{Y}$ constructed in Definition~\ref{def:Y} has the same stable tail dependence function as the random vector $\bm{Z}_{M}$ with distribution $G_{M}$ in Corollary~\ref{pro:GM_Markov} since $\bm{Y}\in D(G_{M})$. Moreover, for all $(a,b)\in E$, the pairs $(Y_{a},Y_{b})$ and $(Z_{a},Z_{b})$ are equal in distribution by construction. Consequently, the stable dependence functions of all adjacent pairs on $\tree$ of $\bm{\xi}$, $\bm{Z}$, $\bm{Y}$ and $\bm{Z}_{M}$ are identical. The same statement holds also for the upper tail dependence coefficients. Note that the increasing component-wise transformations $x_{v} \mapsto 1/(1-F_{v}(x_{v}))$ have no effect on $\lambda$. 
Thus we can estimate the edge weights directly from samples drawn from $\bm{\xi}$.

Let $\I(\cdot)$ denote the indicator function. For the upper tail dependence coefficient defined in \eqref{eq:TDC}, we consider the empirical estimator of $\lambda_{ab}$ for $a,b\in V$ defined by
\begin{equation}
\label{eq:empestimator_TDC}
    \hat{\lambda}_{ab}=\frac{1}{k_{\lambda}}\sum_{i=1}^{n} \I\left\{
        \hat{F}_{na}(\xi_{i,a})\geq 1-\frac{k_{\lambda}}{n}, \,
        \hat{F}_{nb}(\xi_{i,b})\geq 1-\frac{k_{\lambda}}{n}
    \right\},
\end{equation}
where $\hat{F}_{na}(x) = n^{-1} \sum_{j=1}^{n} \I(\xi_{j,a}\le x)$ is the empirical distribution function of $\xi_{a}$ and $k_{\lambda}=k_{\lambda}(n)$ is an intermediate sequence such that $k_{\lambda}\to\infty$ and $k_{\lambda}/n\to 0$ as $n\to\infty$. Its consistency and asymptotic normality were shown, for example, in \citet{SS06}. Consequently, 
we define
 \begin{equation}
 \label{eq:tree_star_est}
    \hat{\bm{T}}_{\omega}^{\star}
    =\argmax_{\tree}\hat{S}_{\tree_{\omega}}
    =\argmax_{\tree=(V,E)}\left(\sum_{(a,b)\in E} \hat{\omega}_{ab}\right)
\end{equation} 
to be the collection of maximum dependence trees with respect to the estimated edge weights $\hat{\omega}_{ab}$, where $\hat{\omega}$ is the estimate of $\omega$.
A generic tree in the set $\hat{\bm{T}}_{\omega}^{\star}$ will be denoted by $\hat{\tree}_{\omega}^{\star}$.

The following proposition is a direct consequence of the consistency of estimates of the upper tail dependence coefficients. It states that, with probability tending to one, the tree estimate $\hat{\tree}_{\lambda}^{\star}$ is an element of $\bm{T}^{\star}_{\lambda}$. If, further, $\bm{T}^{\star}_{\lambda}$ is a singleton with element $\tree_{\lambda}^{\star}$, then $\hat{\tree}_{\lambda}^{\star}$ is a consistent estimator of $\tree_{\lambda}^{\star}$.

\begin{proposition}
\label{pro:consitency_lambda}
Let $\bm{T}^{\star}_{\lambda}$ be the set of maximum dependence trees in~\eqref{eq:tree_star} weighted by the tail dependence coefficients and let $\hat{\bm{T}}_{\lambda}^{\star}$ in \eqref{eq:tree_star_est} be its sample equivalent based on the empirical tail dependence coefficients in \eqref{eq:empestimator_TDC}. If $k_{\lambda}=k_{\lambda}(n)$ is an intermediate sequence such that $k_{\lambda}\to\infty$ and $k_{\lambda}/n\to 0$ as $n\to\infty$, then
\[
    \pr\left(\hat{\bm{T}}_{\lambda}^{\star} \subseteq \bm{T}_{\lambda}^{\star}\right)\to 1, \qquad n\to\infty. 
\]
Moreover, if $\bm{T}^{\star}_{\lambda}$ is a singleton, with unique element $\tree_{\lambda}^{\star}$, then, with probability tending to one, $\hat{\bm{T}}_{\lambda}^{\star}$ is a singleton too and its unique element $\hat{\tree}_{\lambda}^{\star}$ satisfies
\[
    \pr\left(\hat{\tree}_{\lambda}^{\star}=\tree_{\lambda}^{\star}\right)\to 1, \qquad n\to\infty.
\]
\end{proposition}

We stress that Proposition~\ref{pro:consitency_lambda} holds for a general max-stable distribution $G$ which can be either tree-based or not. In case $G$ has a dependence structure linked to the tree $\tree$ and the set $\bm{T}_{\lambda}^{\star}$ is a singleton, then there is a true tree structure $\tree$ and by Proposition~\ref{pro:uniqueness} and~\ref{pro:consitency_lambda} we know that the sample-based tree estimator $\hat{\tree}_{\lambda}$ is a consistent estimator of the true tree $\tree$. If $G$ is not tree-based, we then are estimating the elements of $\bm{T}_{\lambda}^{\star}$ based on which we proposed to construct the max-stable distribution $G_M$ and in this case Proposition~\ref{pro:consitency_lambda} ensures that the tree estimate we obtain is, with probability tending to one, an element of $\bm{T}_{\lambda}^{\star}$. If, further, $\bm{T}_{\lambda}^{\star}$ is a singleton with element $\tree_{\lambda}^{\star}$, then $\hat{\tree}_{\lambda}^{\star}$ is a consistent estimator of $\tree_{\lambda}^{\star}$.

Once all edge weights $\hat{\omega}_{ab}$ have been computed, Prim's algorithm is used to find a global optimizer $\hat{\tree}_{\omega}^{\star}$ in \eqref{eq:tree_star_est}, the solution being unique if all edge weights are distinct. The time complexity of this procedure to find a tree with $d$ nodes is $\mathcal{O}(d^{2})$, see \citet{P57} and \citet[Theorem~12.2]{P98}. The algorithm is given in \mdf{Section~S2 of the Supplementary Material}.

\begin{remark}
Another possible edge weight, although it is not commonly used in the field of extreme value theory, is Kendalls' $\tau$. The upper tail dependence coefficient of a random vector is the same as the one of the max-stable distribution to which it is attracted, and it provides motivation for using this coefficient when given data sampled from a distribution in the domain of attraction of a max-stable distribution. For Kendall's $\tau$, however, this argument does not hold: its value can differ considerably between the data-generating distribution and the max-stable distribution to which it is attracted. Nevertheless, the method based on Kendall's $\tau$ performs well in the simulations and in the case study in Sections~\ref{sec:Sim} and~\ref{sec:app}, and this is why we discuss it here too. 

For $(a,b)\in E$, the value $\tau_{ab}$ for $(\xi_a,\xi_b)$ can be estimated by the empirical estimator
\begin{equation*}
\hat{\tau}_{ab}=\frac{2}{n(n-1)} \sum_{1\le i<j\le n} \left\{\sgn(\xi_{i,a}-\xi_{j,a}) \sgn(\xi_{i,b}-\xi_{j,b})\right\},    
\end{equation*}
which, by classical theory of $U$-statistics, is consistent and asymptotically normal as $n\to\infty$; see, e.g., \citet[Theorem~12.3 and Example~12.5]{V00}. Based on the estimates of $\tau_{ab}$, we can also define the collection of maximum dependence trees $\hat{\bm{T}}_{\tau}^{\star}$ with respect to the estimated edge weights $\hat{\tau}_{ab}$ according to~\eqref{eq:tree_star_est}.

However, as also kindly pointed out by the two reviewers and the associated editor, this approach cannot work in great generality, especially for data sets where the dependence structure in the center differs greatly from that in the tails.  
\end{remark}

\subsection{Estimation of pairwise dependence structures}
\label{subsec:EstSTDF}

Once a tree structure $(V, E)$ is specified, we are able to construct a Markov tree according to Definition~\ref{def:Y}. Actually, there may be several maximum dependence trees in the set $\bm{T}_{\lambda}^{\star}$ if it is not a singleton. However, 
we can only obtain one estimate in practice and by Proposition~\ref{pro:consitency_lambda}, it would be an element of the set $\bm{T}_{\lambda}^{\star}$ with probability tending to one, thus equal to a fixed member in $\bm{T}_{\lambda}^{\star}$. For simplicity, we assume henceforth that the tree is fixed and non-random.

Since $G_{M}$ in Corollary~\ref{pro:GM_Markov} is determined by the tree and the distributions of $M_{e}$ with $e\in E$, it is sufficient to estimate the bivariate distributions of the adjacent pairs. Ultimately, this comes down to the estimation of bivariate dependence structures, as the univariate marginal distributions are standardized. We will estimate the bivariate stable tail dependence function $\ell_{ab}$ of each adjacent pair $(Z_{a},Z_{b})$ with $(a,b)\in E$, which also equals the stable tail dependence function $\ell_{ab}^{M}$ of $(Y_{a},Y_{b})$ and of $(Z_{a}^{M},Z_{b}^{M})$ as discussed in Subsection~\ref{subsec:LOTS}.

The estimation of the stable tail dependence function dates back to \citet{H92} and 
\citet{D98} and has been extensively investigated in the literature, see for instance \citet{P08}, \citet{EKS12}, \citet{B16}, \citet{E16}, \citet{K18}, and the references therein. 

In our framework, for each $e=(a,b)\in E$, we assume that the stable tail dependence function $\ell_e$ of $(Z_{a},Z_{b})$ belongs to a parametric family $\{(x,y) \mapsto \ell(x,y;\bm{\beta}_{e}),\bm{\beta}_{e}\in \bm{B}_{e}\}$, where the parameter $\bm{\beta}_{e}=(\beta_{e,1},\ldots,\beta_{e,p_{e}})^{\top}$ belongs to some parameter space $\bm{B}_{e} \subseteq \mathbb{R}^{p_{e}}$ of dimension $p_{e}\ge 1$. The stable tail dependence functions for different pairs are allowed to belong to different parametric families. Assume that the true parameter $\bm{\beta}_{e}^{0}$ of $\ell_e$ belongs to the interior $\bm{B}_{e}^{o}$ of $\bm{B}_{e}$. Let $\bm{\beta}_{0}=\left(\bm{\beta}_{e}^{0}, e\in E\right)$ be the $(\sum_{e\in E}p_{e})\times 1$ column vector of true parameters. 
It is to be noted that Definition~\ref{def:maxstabletree} of a tree-structured extreme-value distribution does not pose any constraints whatsoever on the combination of parametric families and parameter values. This flexibility greatly facilitates model construction and statistical inference.

For each $e=(a,b)\in E$, we estimate the parameter $\bm{\beta}_{e}^{0}$ using only the data in components $a$ and $b$. Define 
\begin{equation*}
    \hat{\bm{\beta}}_{n}=(\hat{\bm{\beta}}_{n,e}, e\in E)
\end{equation*} 
to be the estimator of $\bm{\beta}_{0}$, where $\hat{\bm{\beta}}_{n,e}$ is any consistent estimator of $\bm{\beta}_{e}^{0}$ for $e\in E$. Since the convergence in probability of a sequence of vectors is equivalent to the convergence of every one of the component sequences separately \citep[Theorem~2.7]{V00}, the estimator $\hat{\bm{\beta}}_{n}$ is jointly consistent provided the estimators of $\bm{\beta}_{e}^{0}$ for $e\in E$ are all consistent. The analogous statement for convergence in distribution is, however, not true. Still, we show that the joint estimator is asymptotically normal if the edge-wise estimators have a certain asymptotic expansion in terms of the empirical stable tail dependence function, which is recalled below. The expansion is shared by the estimators proposed in \cite{E08}, \cite{EKS12} and \cite{E18}, as we will show as well.

Given an independent random sample $\bm{\xi}^{i}=(\xi_{i,v},v\in V)$, $i=1,\ldots, n$, let $R_{v}^{i}$ denote the rank of $\xi_{i,v}$ among $\xi_{1,v},\ldots,\xi_{n,v}$ for $v\in V$. The empirical stable tail dependence function $\hat{\ell}_{n,k}$ is
\begin{equation*}
    \hat{\ell}_{n,k}(\bm{x})=\frac{1}{k}\sum_{i=1}^{n} \I\left(R_{1}^{i}>n+\frac{1}{2}-k x_{1}~\text{or}~\ldots~\text{or}~R_{d}^{i}>n+\frac{1}{2}-k x_{d}\right), \qquad \bm{x}\in [0,\infty)^{d},
\end{equation*} 
where $k=k(n)$ is an intermediate sequence such that $k\to\infty$ and $k/n\to 0$ as $n\to\infty$ \citep{H92,D98}. The asymptotic distribution of $\hat{\ell}_{n,k}$ was studied in \citet{EKS12} and \citet{BSV14} under the following assumption.

\begin{assumption}
\label{con:SedOrdstdf}
Let $k=k(n)$ be an intermediate sequence such that $k\to\infty$ and $k/n\to 0$ as $n\to\infty$. There exists $\zeta > 0$ such that the following two asymptotic relations hold:
\begin{enumerate}[(i)]
\item
    $t^{-1} \pr\left[1-F_{1}(\xi_{1}) \le t x_{1} \text { or } \ldots \text { or } 1-F_{d}(\xi_{d}) \le t x_{d} \right] - \ell(\bm{x}) = O(t^{\zeta})$ as $t \downarrow 0$, uniformly in $\bm{x}\in \Delta_{d-1} = \left\{ \bm{w} \in [0,1]^{d} : w_{1}+\cdots+w_{d} = 1 \right\}$;
\item 
    $k=o\left(n^{2\zeta/(1+2\zeta)}\right)$ as $n\to\infty$.
\end{enumerate}
\end{assumption}

Let $W$ be a mean-zero Gaussian process on $[0, \infty)^d$ with continuous trajectories and covariance function 
\[ 
    \E[W(\bm{x}) W(\bm{y})] = \ell(\bm{x}) + \ell(\bm{y}) - \ell(\bm{x} \vee \bm{y}),
    \qquad \bm{x}, \bm{y} \in [0, \infty)^d, 
\]
where $\bm{x} \vee \bm{y} = (\max\{x_j, y_j\})_{j=1}^d$. Further, let $\dot{\ell}_j(x)$ denote the right-hand partial derivative of $\ell$ with respect to the $j$-th coordinate, for $j \in \{1,\ldots,d\}$. Since $\ell$ is convex, this one-sided partial derivative always exists and is continuous in points where the ordinary two-sided partial derivative exists. Consider the stochastic process $\alpha$ on $[0, \infty)^d$ defined by
\begin{equation}
\label{eq:limproc}
    \alpha(\bm{x}) = W(\bm{x}) - \sum_{j=1}^d \dot{\ell}_j(\bm{x}) W(x_j \bm{e}_j), \qquad \bm{x} \in [0, \infty)^d,
\end{equation}
where $\bm{e}_j$ denotes the $j$-th canonical unit vector in $\mathbb{R}^d$. In \citet[Theorem~4.6]{EKS12}, the asymptotic distribution of $\sqrt{k} (\hat{\ell}_{n,k} - \ell)$ is established with respect to the topology of uniform convergence on $[0, 1]^d$. An additional condition needed on $\ell$ is that $\dot{\ell}_j$ is continuous in $\bm{x}$ whenever $x_j > 0$. For certain models, however, this condition fails; a counter-example is the asymmetric logistic distribution in Example~\ref{eg:AsyLogZero}. A more general result is given in \citet[Theorem~5.1]{BSV14}, who establish weak convergence of $\sqrt{k} (\hat{\ell}_{n,k} - \ell)$ in a slightly weaker topology but without additional differentiability conditions on $\ell$. The topology is based on epi- and hypographs of functions, whence the name hypi-topology. What concerns us here is that the topology is still strong enough to yield weak convergence in certain $L^p$-spaces and thus of certain integrals of the empirical stable tail dependence function. This will be crucial when studying parameter estimates based on $\hat{\ell}_{n,k}$.

Recall that we observe an independent random sample $\bm{\xi}^1,\ldots,\bm{\xi}^n$ from a distribution $F$ such that, after marginal transformation, the standardized random vector belongs to the max-domain attraction of a max-stable distribution $G$ which is not necessarily tree-structured. In the coming theorem we show that the certain integrals of the empirical stable tail dependence function converge weakly to a centered multivariate normal distributed random vector, which is crucial for the proof of asymptotic normality of $\bm{\beta}_n$ in Corollary~\ref{cor:AN}.

\begin{theorem}
\label{thm:estdf:AN}
Let $\mu_1,\ldots,\mu_q$ be finite Borel measures on $[0, 1]^d$ such that for all $j \in \{1,\ldots,d\}$ and $r \in \{1,\ldots,q\}$, the set of points $\bm{x}$ with $x_j > 0$ in which $\dot{\ell}_j$ is not continuous is a $\mu_r$-null set. Let $\psi_r \in L^2(\mu_r)$ for all $r \in \{1,\ldots,q\}$. Then under Assumption~\ref{con:SedOrdstdf}, we have as $n \to \infty$ that
\begin{equation}
\label{eq:estdf:AN}
    \left( 
        \int_{[0, 1]^d} 
            \sqrt{k} \left( \hat{\ell}_{n,k}(\bm{x}) - \ell(\bm{x}) \right)
            \psi_r(\bm{x}) \, 
        \mu_r(\d \bm{x}) 
    \right)_{r=1}^{q}
    \dto
    \left(
        \int_{[0, 1]^d} \alpha(\bm{x}) \, \psi_r(\bm{x}) \, \mu_r(\d \bm{x})
    \right)_{r=1}^{q}.
\end{equation}
The limit is centered multivariate normal with covariance matrix $\bm{\Sigma} = (\sigma_{rs})_{r,s=1}^q$ having entries
\[
    \sigma_{rs} = 
    \int_{[0, 1]^d} \int_{[0, 1]^d} 
        \E[\alpha(\bm{x}) \alpha(\bm{y})] \,
        \psi_r(\bm{x}) \, \psi_s(\bm{y}) \, 
    \mu_r(\d \bm{x}) \, \mu_s(\d \bm{y}).
\]
\end{theorem}

\begin{remark}
\label{rk:estdf:AN}
The measures $\mu_r$ in Theorem~\ref{thm:estdf:AN} can be discrete, in which case \eqref{eq:estdf:AN} states the asymptotic normality of the finite-dimensional distributions of $\sqrt{k} \{ \hat{\ell}_{n,k}(\bm{x}) - \ell(\bm{x}) \}$, at least in points $\bm{x} \in [0, 1]^d$ such that for every $j \in \{1,\ldots,d\}$ we have either $x_j = 0$ or $\dot{\ell}_j$ is continuous in $\bm{x}$.
\end{remark}

In the bivariate case, the non-parametric estimator $\hat{\ell}_{n,k,e}(x_{a},x_{b})$ of $\ell_{e}(x_{a},x_{b})$ for $e=(a,b)\in E$ becomes
\begin{equation*}
    \hat{\ell}_{n,k,e}(x_a,x_b) 
    = \frac{1}{k}\sum_{i=1}^{n} \I \left(
        R_{a}^{i} > n + \frac{1}{2} - k x_a
        ~\text{or}~
        R_{b}^{i} > n + \frac{1}{2} - k x_b
    \right)
    = \hat{\ell}_{n,k}(x_a \bm{e}_a + x_b \bm{e}_b).
\end{equation*} 
As the formula indicates, it arises from $\hat{\ell}_{n,k}$ by setting $x_j = 0$ whenever $j \ne a$ and $j \ne b$. 

\begin{assumption}
\label{ass:expansion}
There exist an intermediate sequence $k = k(n)$ and, for every $e = (a, b) \in E$, a finite Borel measure $\nu_{e}$ on $[0, 1]^2$ and a vector of functions $\bm{\psi}_e = (\psi_{e,1},\ldots,\psi_{e,p_{e}})^{\top} : [0, 1]^2 \to \mathbb{R}^{p_e}$ in $L^2(\nu_{e})$ such that the following two properties hold:
\begin{enumerate}[(i)]
\item
As $n \to \infty$, we have
\begin{equation}
\label{eq:expansion}
    \sqrt{k} \left( \hat{\bm{\beta}}_{n,e} - \bm{\beta}_{e}^{0} \right)
    = \int_{[0,1]^2} 
        \sqrt{k} \left( \hat{\ell}_{n,k,e}(x_a,x_b) - \ell_e(x_a,x_b) \right)
        \bm{\psi}_{e}(x_a,x_b) \, 
    \nu_{e}(\d (x_a,x_b)) + o_{p}(\bm{1}).
\end{equation}
\item
For $c \in \{a, b\}$, the set of points $(x_{a}, x_{b})$ such that $x_{c} > 0$ and such that the first-order partial derivative of $\ell_{e}$ with respect to $x_{c}$ does not exist is a $\nu_{e}$-null set.
\end{enumerate}
\end{assumption}

\begin{corollary}
\label{cor:AN}
Under Assumptions~\ref{con:SedOrdstdf} and~\ref{ass:expansion}, for the same intermediate sequence $k(n)$, we have
\[
    \sqrt{k} \left( \hat{\bm{\beta}}_{n} - \bm{\beta}_{0} \right)
    \dto
    \left( 
        \int_{[0,1]^2} 
            \alpha(x_{a} \bm{e}_{a} + x_{b} \bm{e}_{b}) \,
            \bm{\psi}_{e}(x_{a}, x_{b}) \,
        \nu_{e}(\d (x_{a}, x_{b}))
    \right)_{e = (a, b) \in E},
    \qquad n \to \infty,
\]
with $\alpha$ as in \eqref{eq:limproc}. The limit vector is centered multivariate normal with covariance matrix $\bm{\Sigma} = (\bm{\Sigma}_{e e'})_{e, e' \in E}$ whose block indexed by $e = (a, b)$ and $e' = (a', b')$ in $E$ is of dimension $p_{e} \times p_{e'}$ and is given by
\begin{multline*}
    \bm{\Sigma}_{e e'} =
    \int_{[0, 1]^2} \int_{[0, 1]^2}
        \E\left[
            \alpha(x_{a} \bm{e}_{a} + x_{b} \bm{e}_{b}) \,
            \alpha(x_{a'} \bm{e}_{a'} + x_{b'} \bm{e}_{b'}) 
        \right] \\
        \cdot 
        \bm{\psi}_{e}(x_{a},x_{b}) \,
        \bm{\psi}_{e'}(x_{a'},x_{b'})^{\top} \,
    \nu_{e}(\d(x_{a}, x_{b})) \, \nu_{e'}(\d(x_{a'}, x_{b'})).
\end{multline*}
\end{corollary}

Corollary~\ref{cor:AN} allows for a combination of different estimators for the edge parameters $\bm{\beta}_{e}$, as long as each of them satisfies Assumption~\ref{ass:expansion}. 
Three such estimators are the moments estimator in \citet{E08}, the M-estimator in \citet{EKS12} and the weighted least squares estimator in \citet{E18}. Details are given in Appendix~\ref{subsec:three_estimators}.

\subsection{Description of the approximation error}

As soon as the bivariate margins of all the adjacent pairs on $\bm{Y}$ are estimated, an estimate of distribution function $G_{M}$ can be obtained from Corollary~\ref{pro:GM_Markov} by replacing the unknown quantities in there by their estimates. Likewise, a model-based estimate $\hat{\lambda}_{ab}^{M}$ of the upper tail dependence coefficient $\lambda_{ab}^{M}$ can be obtained from \eqref{eq:TDC_GM}. Thus we take 
\begin{equation}
\label{eq:D_est}
\hat{D}_{\tree}=\sum_{a,b\in V,(a,b)\notin E}|\hat{\lambda}_{ab}^{M}-\hat{\lambda}_{ab}|
\end{equation} 
as a consistent estimator of $D_{\tree}$ defined in \eqref{eq:D} to measure the approximation error, where $\hat{\lambda}_{ab}$ is the empirical estimator of $\lambda_{ab}$ given in \eqref{eq:empestimator_TDC}.

Under the null hypothesis that $G$ is a tree-structured max-stable distribution along a tree, the asymptotic null distribution of $\hat{D}_{\tree}$ could be derived, although this would be tedious. But since we view $G_M$ as an approximation of $G$, we are not really interested in testing the tree-structure hypothesis on $G$. Instead, we use $\hat{D}_{\tree}$ as a descriptive statistic, a large value of which means that the approximation is inaccurate and needs to be improved.

In line with the measure $\tilde{D}_{\tree}$ defined Section~\ref{subsec:ApproErr}, one can consider other error measures such as
\[
    \hat{\tilde{D}}_{\tree} =
    \max_{a,b \in V; a \ne b; (a,b) \notin E} 
    \sup_{(x_a,x_b) \in [0, 1]^2}
    \left| 
        \hat{\ell}_{ab}^M(x_a,x_b) - \hat{\ell}_{n,k,ab}(x_a,x_b) 
    \right|
\]
where $\hat{\ell}_{ab}^M$ is the bivariate stable tail dependence function implied by the fitted model, while $\hat{\ell}_{n,k,ab}$ is the empirical tail dependence function. However, this quantity will be difficult to calculate. Rather than actually quantifying the approximation error via $\hat{D}_{\tree}$, we could also simply plot the pairs $(\hat{\lambda}_{ab}^{M}, \hat{\lambda}_{ab})$ as in Figure~8 in \citet{hentschel2022statistical}. Such a scatterplot can be used as a kind of diagnostic plot: the closer the points are concentrated around the main diagonal, the better the fit.

\section{Simulation study}
\label{sec:Sim}

To evaluate the performance of the proposed approach, we perform our modelling procedure on random samples from the H\"{u}sler--Reiss and asymmetric logistic max-stable distributions in Examples~\ref{eg:HR} and~\ref{eg:AsyLogZero} respectively. 
To achieve a more realistic situation, we add some lighter-tailed noise.
Specifically, starting from random vectors $\bm{Z}^{i}$, for $i=1,\ldots,n$, drawn independently from a simple H\"{u}sler--Reiss or asymmetric logistic distribution,
we construct the samples from a random vector $\bm{\xi}$ by 
\begin{align}
    \label{eq:sample}
    \bm{\xi}^{i}=\bm{Z}^{i}+\bm{\varepsilon}^{i}, \qquad i=1,\ldots,n,
\end{align}
where 
$\bm{\varepsilon}^{i}$, for $i=1,\ldots,n$, are independent of $\bm{Z}^{i}$ and are themselves independent and identically distributed random vectors with independent Fr\'echet($2$) distributed entries, i.e., $\pr(\varepsilon^{i}_{v}\le x)=\exp(-1/x^{2})$ for $x > 0$ and $v\in V$. 

The measure $2\hat{D}_{\tree}/[(d-1)(d-2)]$ with $\hat{D}_{\tree}$ in~\eqref{eq:D_est} is calculated to assess the model's goodness-of-fit. The resulting tree-based models are applied to the estimation of the rare event probabilities in~\eqref{eq:rareevent}.
To do that, we need take account of both the dependence structure and the margins of the sample. Thus it is more convenient to work with the assumption that a general random vector $\bm{\xi}$ is in the domain of attraction of a multivariate max-stable distribution $H$ with margins $H_{v}$ for $v\in V$. This is equivalent to the assumption $(1/(1-F_{v}(\xi_{v})),v\in V)\in D(G)$ with $G$ the simple multivariate max-stable distribution that relates to $H$ through $G(\bm{x})=H(H_{v}^{\leftarrow}(e^{-1/x_{v}}),v\in V)$ (the superscript arrow denoting the functional inverse), together with the 
condition that $F_{v}$ is in the domain of attraction of the univariate extreme value distribution $H_{v}$ for $v\in V$ respectively. Then by equation~(8.81) in \citet{B04}, we have
\begin{align}
\nonumber
\pr\left(\xi_{1}>u_{1}~\text{or}~\ldots ~\text{or}~\xi_{d}>u_{d}\right) 
    &\approx 1- G(-1/\log{F_{v}(u_{v})},\, v\in V) \\
\label{eq:rareeventappxi}
    &\approx 1- G_{M}(-1/\log{F_{v}(u_{v})},\, v\in V)
\end{align}
for thresholds $u_{v}$ such that all probabilities $F_{v}(u_{v})$ are close to unity. 
 
For the random vector $\bm{\xi}$ with copies constructed in \eqref{eq:sample}, we know that $\bm{\xi}\in D(G)$ with $G$ equal to the common distribution of $\bm{Z}^{i}$, $i=1,\ldots,n$. Moreover, the marginal distributions $F_{v}$, $v \in V$, belong to the domain of attraction of the unit-Fr\'{e}chet distribution $G_v$, so that $-\log{F_{v}(u_{v})}\approx 1/u_{v}$ provided $u_{v}$ is sufficiently large. Therefore, the rare event probability on the left-hand side of~\eqref{eq:rareeventappxi} can be approximated by $1-G_{M}(\bm{u})$ by plugging the marginal approximations into the quantity on the right-hand side. We will show boxplots of the logarithm of the relative approximation error
\begin{align}
\label{eq:lograreevent}
    \operatorname{AE}
    =\log \frac{1-\hat{G}_{M}(\bm{u})}{1-F(\bm{u})},
\end{align}
where $\hat{G}_{M}$ is the estimate of $G_{M}$ based on the three-step procedure in Section~\ref{sec:estim} and $F$ is the joint cumulative distribution function of $\bm{\xi}$. For simplicity, three components are considered simultaneously.  
In each experiment, we take $u_{j}$, for $j\in J$ with $J\subset V$ and $|J|=3$, as the $0.999$ empirical marginal quantiles and $u_{j}=\infty$ for $j\in V\setminus J$.
All simulations are done in \texttt{R}. Realizations of samples of the H\"{u}sler--Reiss distribution are simulated using the package \texttt{graphicalExtremes} \citep{EHG19}. Samples of the asymmetric logistic distribution are generated through \eqref{eq:ML}. 
The true probability $1-F(\bm{u})$ is calculated by numerical integration as implemented in the package \texttt{cubature} \citep{cubature} based on the convolution of $\bm{Z}$ and $\bm{\varepsilon}$: $F(\bm{u})=\int_{[\bm{0},\bm{u}]} G(u_{1}-x_{1},\ldots,u_{d}-x_{d}) \prod_{v\in V} \left(2x_{v}^{-3}\exp(-1/x_v^{2})\right) \d \bm{x}$.

\subsection{H\"{u}sler--Reiss distribution}
\label{subsec:Sim_HR}

The H\"{u}sler--Reiss distribution may have a dependence structure linked to a tree, see Example~\ref{eg:HR}. Random samples can be generated from either a general H\"{u}sler--Reiss distribution or a structured one. Both cases are considered and experiments are based on $300$ repetitions with samples of size $n=1000$. The procedure is detailed in \mdf{Section~S2 of the Supplementary Material}.   

For random samples generated from the $10$-dimensional H\"{u}sler--Reiss distribution with variogram matrix $\bm{\Gamma}_{4}$ in \mdf{Table~S1 and dependence structure linked to the tree given in Figure~S1 of the Supplementary Material},
the proportion of wrongly estimated edges
\[
    1-\frac{|E_{\hat{\tree}_{\omega}^{\star}}\cap E_{\tree}|}{d-1}
\]
for $\omega\in \{\lambda, \tau\}$ is recorded. For Kendall's tau, all trees $\hat{\tree}_{\tau}^{\star}$ found in the $300$ replications are the same as the true one, while for the upper tail dependence coefficient, there are about $23.2\%$ of replications in which $\hat{\tree}_{\lambda}^{\star}$ does not agree with the true tree in one out of nine edges and $1.3\%$ of replications where two out of nine edges are wrong. This is partly caused by the selected $k_{\lambda}=0.1n$ for the estimation of $\lambda$. \mdf{Figure~S2 of the Supplementary Material} shows an interesting tendency in the influence of $k_{\lambda}$ on the number of replications where $\hat{\tree}_{\lambda}^{\star}$ does not coincide with the true tree: the best proportion  $k_{\lambda}/n$ does not appear in the region close to zero but around $0.5$, even if the estimation bias is large for such $k_{\lambda}$. This provides a practical argument for using Kendall's tau to learn the tree structure, even though $\tau$ is not a measure of tail dependence.
    
Recall that the limiting distribution $G_{M}$ is also a H\"{u}sler--Reiss distribution with variogram matrix $\bm{\Gamma}_{M}$ which can be recovered from the estimated parameters of bivariate stable tail dependence functions and the selected tree structure, see Example~\ref{eg:HR}. A H\"{u}sler--Reiss distribution is totally determined by its variogram matrix. Therefore, the distance between the variogram matrix $\hat{\bm{\Gamma}}_{M}=(\hat{\gamma}^{M}_{ij})_{i,j\in V}$ of the tree-based model $\hat{G}_{M}$ and the true variogram matrix $\bm{\Gamma}=(\gamma_{ij})_{i,j\in V}$ can also be used to measure the approximation error. We compare the variogram matrices through the Frobenius norm of their difference. Boxplots of the relative distance of $\hat{\bm{\Gamma}}_{M}$ to $\bm{\Gamma}$ given by
\begin{equation}
\label{eq:Variodis}
\frac{\Vert \hat{\bm{\Gamma}}_{M}-\bm{\Gamma} \Vert_{F}}{\Vert\bm{\Gamma} \Vert_{F}}=\frac{\left[\sum_{i=1}^{d}\sum_{j=1}^{d}(\hat{\gamma}^{M}_{ij}-\gamma_{ij})^2\right]^{1/2}}{\left[\sum_{i=1}^{d}\sum_{j=1}^{d}\gamma_{ij}^2\right]^{1/2}}
\end{equation}
based on $300$ replications are given in Figure~\ref{fig:10-HR-tree_STDFbias}. The parameters of the stable tail dependence functions are estimated using the moments estimator, the M-estimator and the weighted least squares estimator discussed in Subsection~\ref{subsec:EstSTDF}. The moments estimator and M-estimator turned out to have quite similar performances, which is why we only show the results for the M-estimator and weighted least squares estimator. The two types of edge weights seem comparable in the sense of approximation distance even though $\lambda$ is not as good as $\tau$ in recovering the true tree structure. Compared with the M-estimator, the weighted least squares estimator is more sensitive to the choice of $k$.
Boxplots of the average of $\hat{D}_{\tree}$ basically correspond to those of the relative distance of $\hat{\bm{\Gamma}}_{M}$ to $\bm{\Gamma}$, supporting the use of $\hat{D}_{\tree}$ as a goodness-of-fit measure.

The model based on 
$\hat{\tree}_{\lambda}^{\star}$
is used to estimate the probability in \eqref{eq:rareevent}. From Figure~\ref{fig:10-HR-tree-RareeventBias} we see that the proposed model underestimates the probability in \eqref{eq:rareevent}. This is mainly caused by the underestimation of the variogram matrix. As suggested by a referee, the lengths of the paths connecting the chosen components $\xi_j$ for $j\in J$ which we consider in \eqref{eq:rareevent} may impact the approximation error of the threshold exceedance probability. For illustration, we considered two options for $J$: $J=\{1,2,3\}$ and $J=\{1,3,5\}$, where in the former case one of the nodes, node~2, is on the path from node~1 to node~3 in \mdf{Figure~S1 of the Supplementary Material} but in the latter case the three selected nodes are connected in a looser way. As shown in Figure~\ref{fig:10-HR-tree-RareeventBias}, it turns out that the approximation is more accurate if the considered components are ``closer'' on the tree.

\begin{figure}
    \centering
    \includegraphics[width=0.47\textwidth]{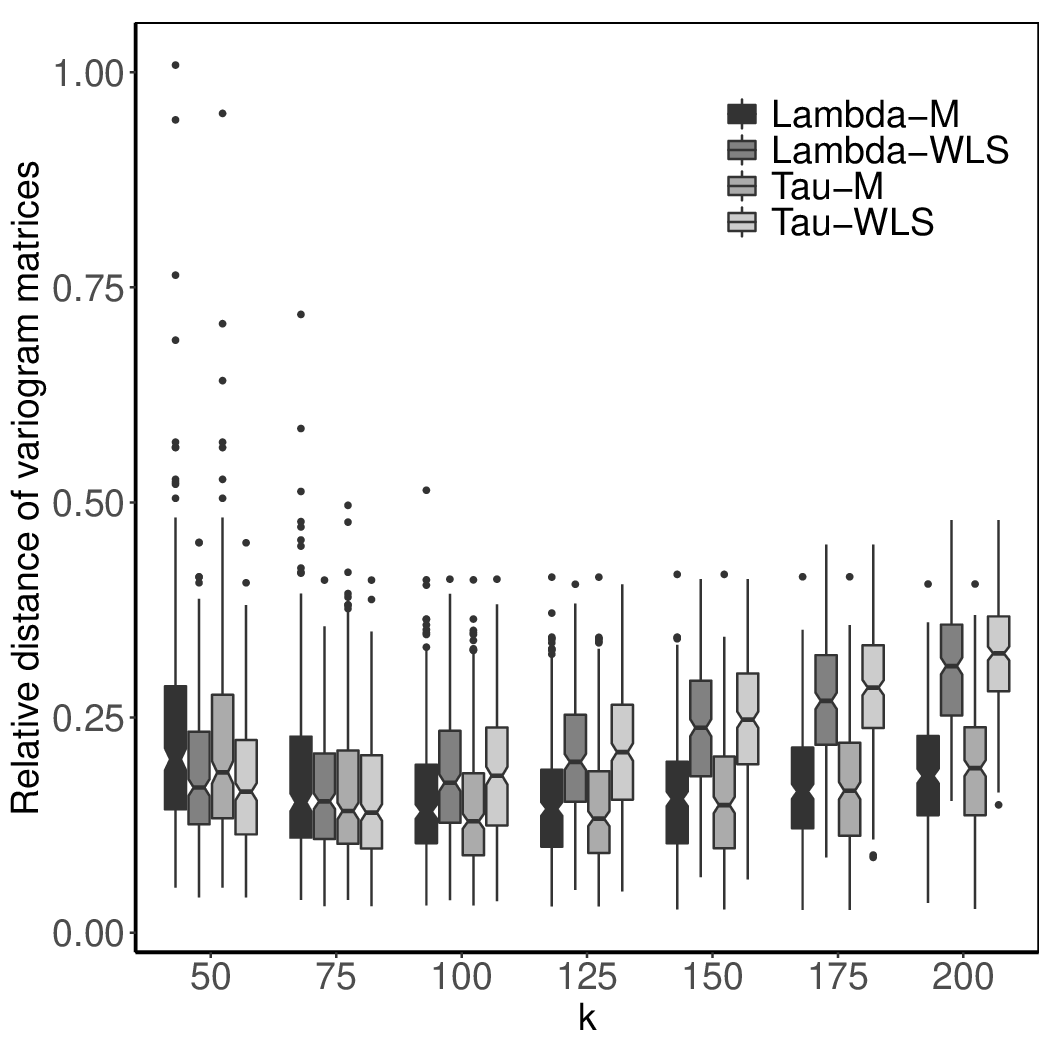}
    \includegraphics[width=0.47\textwidth]{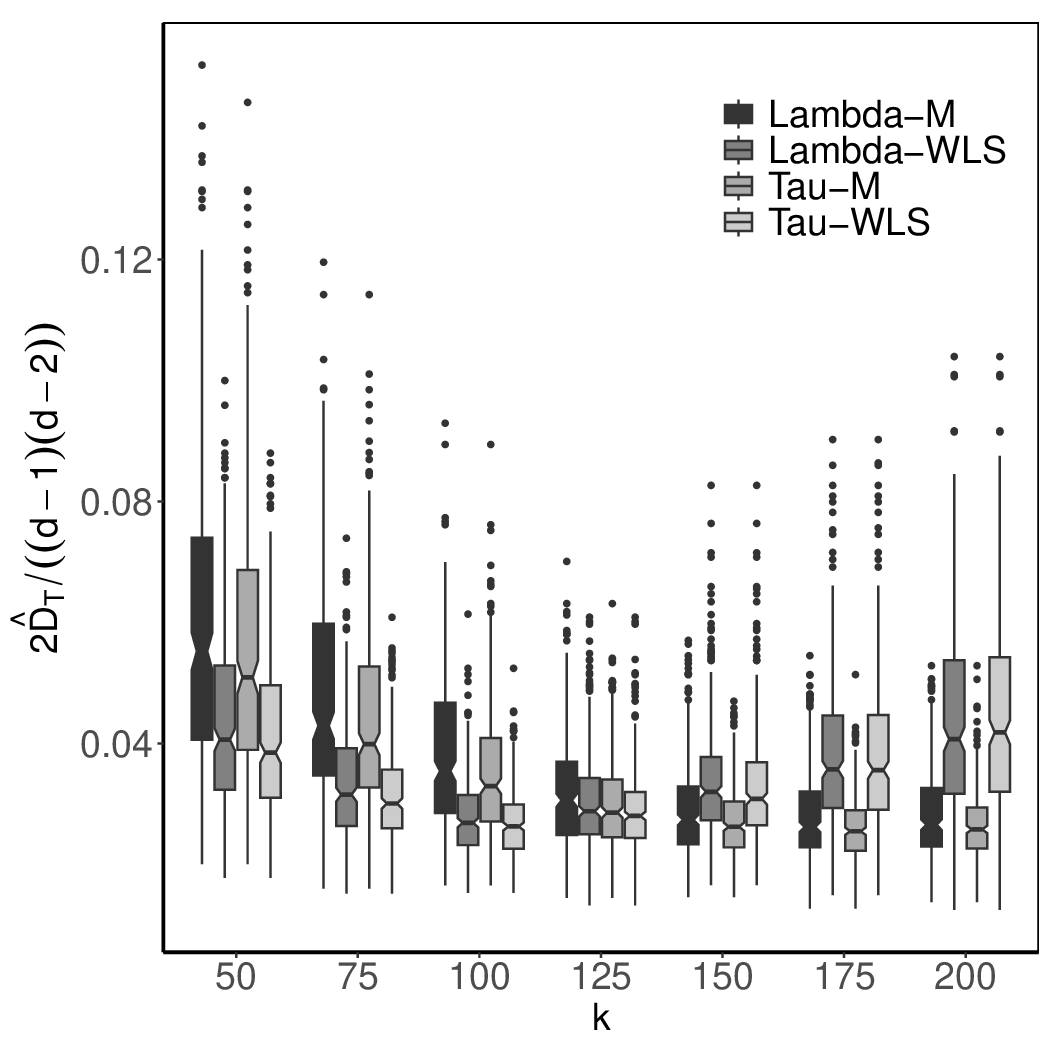}
    \caption{Boxplots of $\frac{\Vert \hat{\bm{\Gamma}}_{M}-\bm{\Gamma} \Vert_{F}}{\Vert \bm{\Gamma} \Vert_{F}}$ (left) and $2\hat{D}_{\tree}/[(d-1)(d-2)]$ with \mdf{$\hat{D}_{\tree}$} in \eqref{eq:D_est} (right) for $10$-dimensional H\"{u}sler--Reiss distribution with variogram matrix $\bm{\Gamma}_{4}$ based on $\hat{\tree}_{\lambda}^{\star}$ (Lambda) and $\hat{\tree}_{\tau}^{\star}$ (Tau), where bivariate stable tail dependence functions are fitted based on M-estimator (M) and weighted least squares estimator (WLS) at $k=50,75,\ldots,200$ ($300$ samples of size $n=1000$)}
    \label{fig:10-HR-tree_STDFbias}
\end{figure}

\begin{figure}
    \centering
    \includegraphics[width=0.47\textwidth]{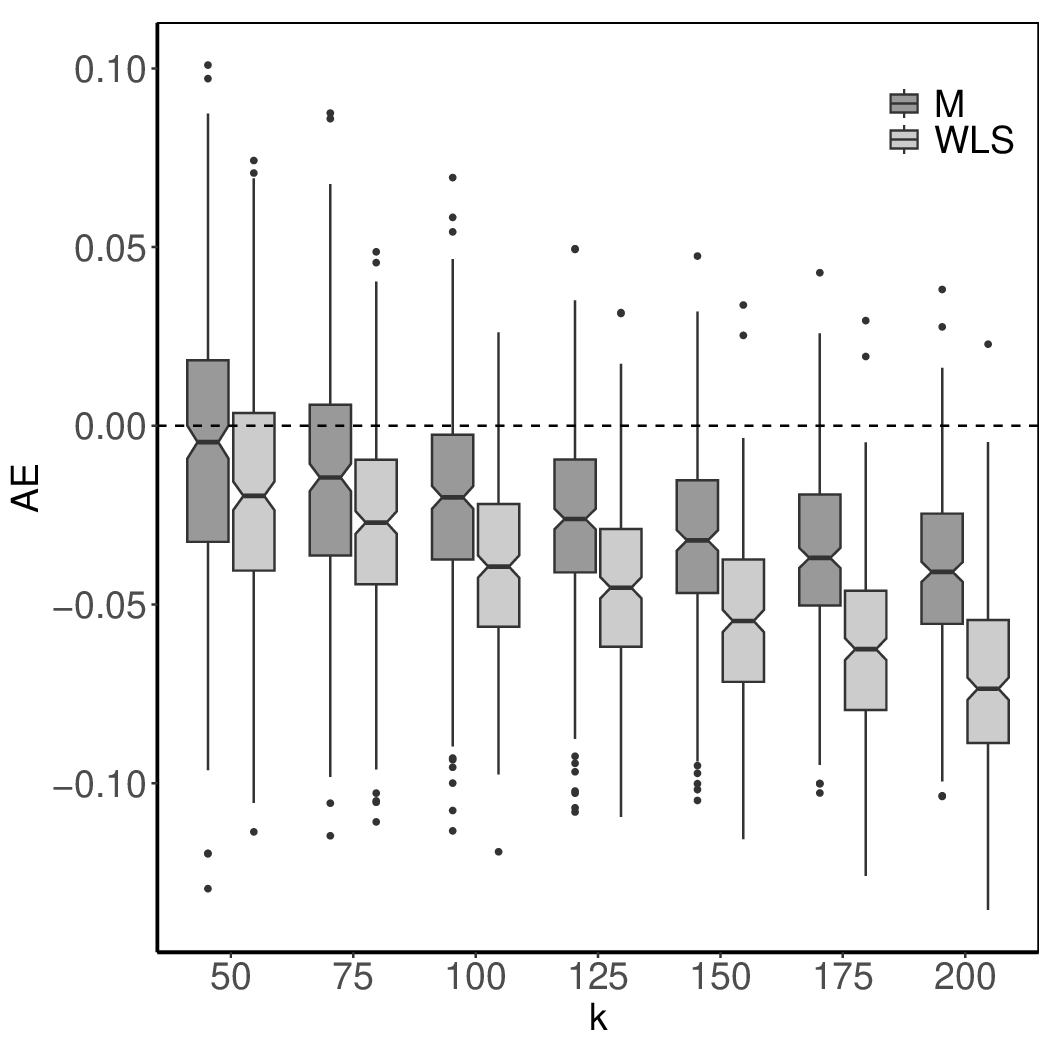}
    \includegraphics[width=0.47\textwidth]{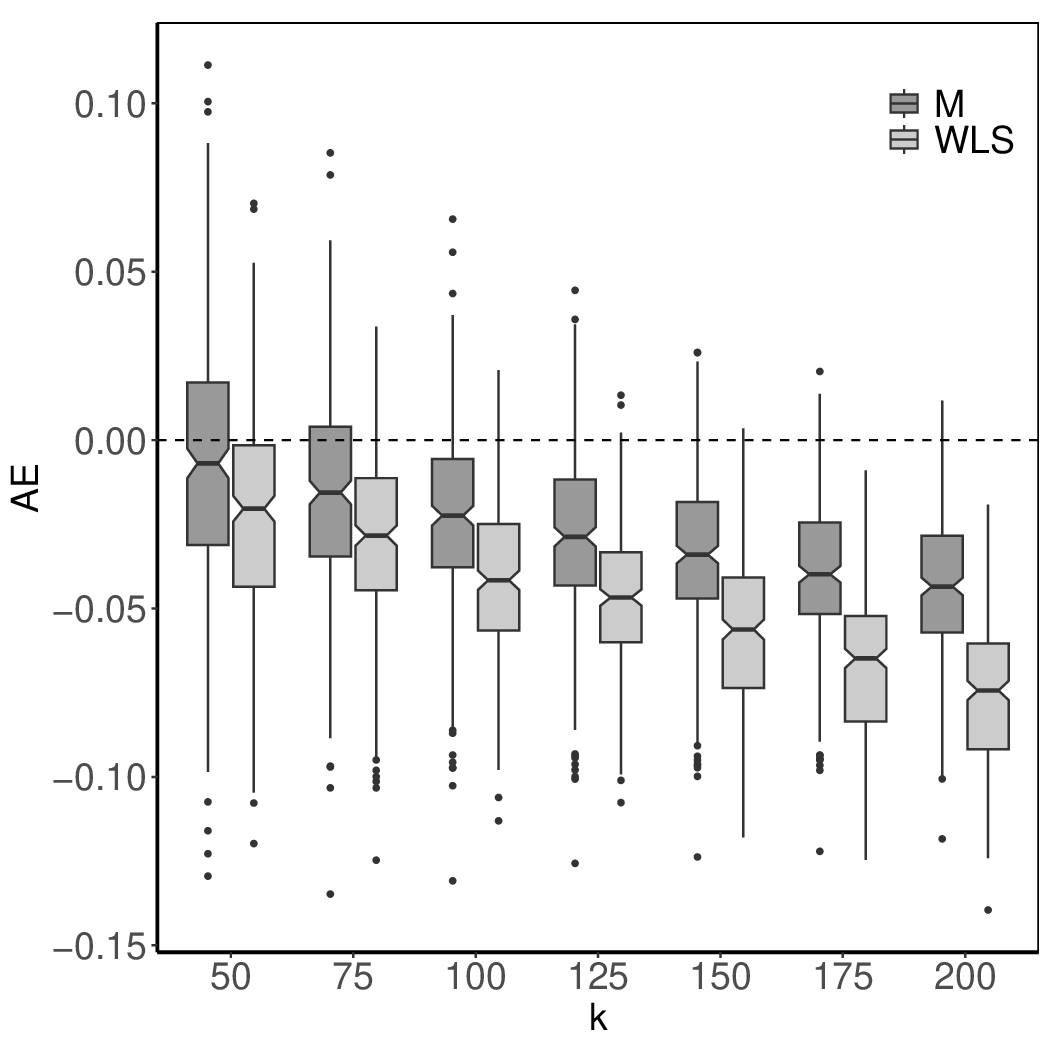}
    \caption{Boxplots of the approximation errors of the rare event probability in~\eqref{eq:lograreevent} based on $\hat{\tree}_{\lambda}^{\star}$ for $10$-dimensional H\"{u}sler--Reiss distribution with variogram matrix $\bm{\Gamma}_{4}$, where bivariate stable tail dependence functions are fitted using the M-estimator (M) and the weighted least squares estimator (WLS) at $k=50,75,\ldots,200$, the levels $u_{j}$ with $j=1,2,3$ (left) and $j=1,3,5$ (right) are taken as the $0.999$ empirical marginal quantiles and $u_{j}=\infty$ for the rest components ($300$ samples of size $n=1000$)}
    \label{fig:10-HR-tree-RareeventBias}
\end{figure}

For samples generated from the $10$-dimensional H\"{u}sler--Reiss distribution with variogram matrix $\bm{\Gamma}_{5}$ which is not necessarily tree-structured, the simulation results have similar patterns as those in the first case but with larger errors. The variogram matrix $\bm{\Gamma}_{5}$ and boxplots can be found in \mdf{Table~S1, Figures~S3 and~S4 in the Supplementary Material}.

\subsection{Asymmetric logistic distribution}
\label{subsec:asymlog}

In reality, we do not know the true distribution $G$. The bivariate distribution family we use to model adjacent pairs may be different from the true one. To see the estimation error in this scenario, we generate samples from the $5$-dimensional asymmetric logistic distribution in Example~\ref{eg:AsyLogZero} (a special max-linear model) with randomly generated parameter $\bm{\Psi}_{4}=(0.763, 0.835, 0.602, 0.747, 0.859)$. Still, we model the adjacent pairs on the constructed Markov tree by bivariate H\"{u}sler--Reiss distributions. 

According to Figure~\ref{fig:5-ML}, the two types of edges weights, $\lambda$ and $\tau$, still exhibit a similar performance. Larger $k$ leads to a higher accuracy for both estimators in terms of $\hat{D}_{\tree}$, while the weighted least squares estimator behaves better than the M-estimator for the chosen $k$, as can also be seen from the approximation error on the right-hand panel in the figure. However, compared with the experiments for the H\"{u}sler--Reiss distribution above, the approximation errors are much larger.

\begin{figure}
    \centering
    \includegraphics[width=0.47\textwidth]{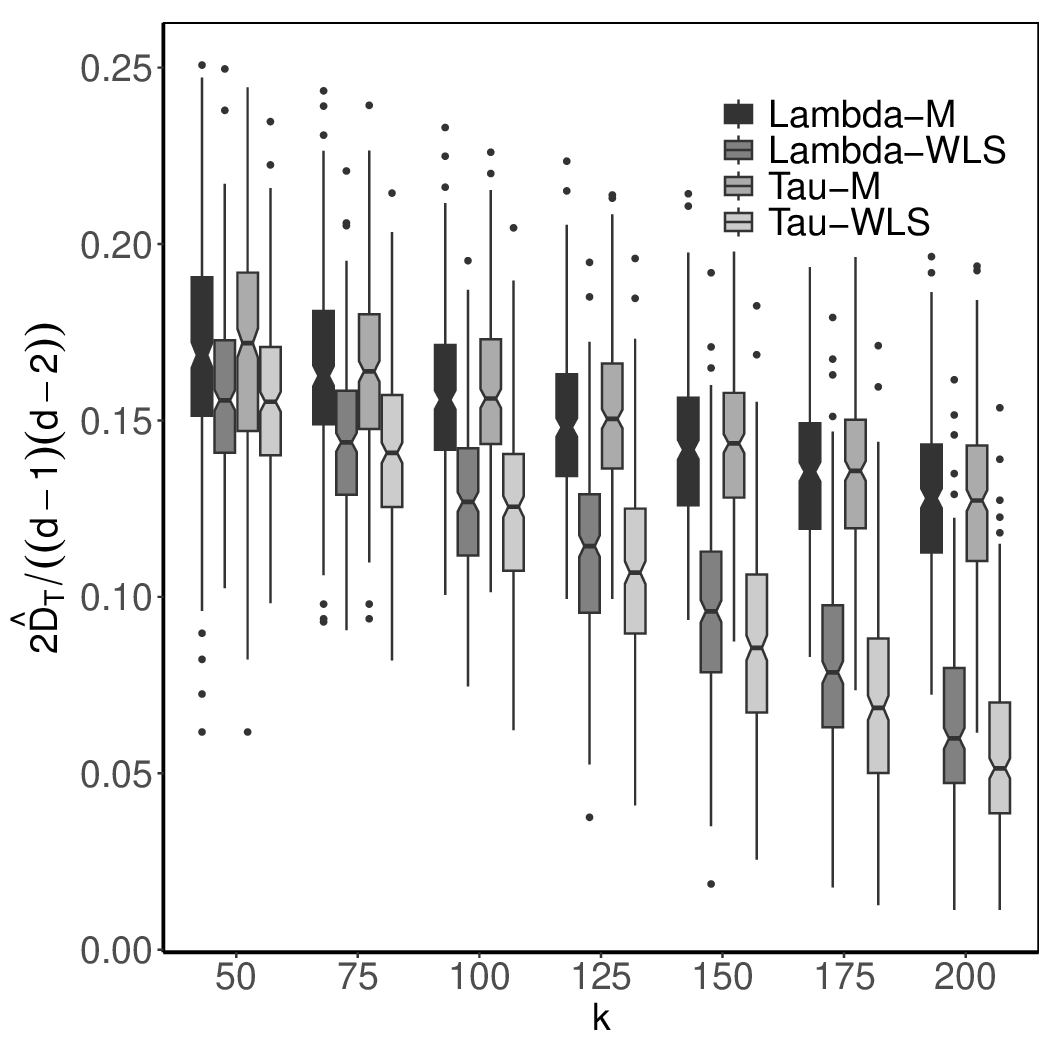}
   \includegraphics[width=0.47\textwidth]{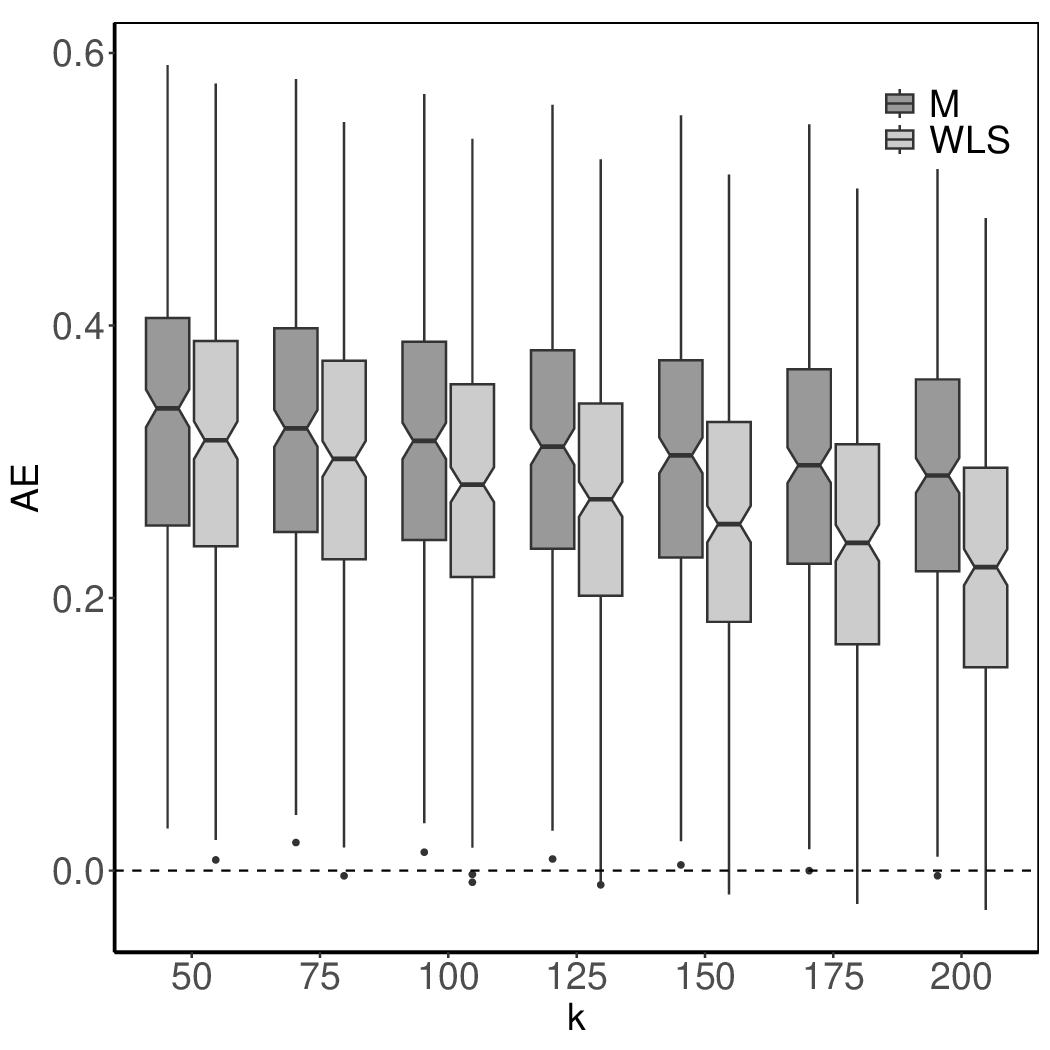}
    \caption{Boxplots of $2\hat{D}_{\tree}/[(d-1)(d-2)]$ (left) with \mdf{$\hat{D}_{\tree}$} in \eqref{eq:D_est} based on $\hat{\tree}_{\lambda}^{\star}$ (Lambda) and $\hat{\tree}_{\tau}^{\star}$ (Tau) and the approximation errors (right) of the rare event probability in \eqref{eq:lograreevent} based on 
    $\hat{\tree}_{\lambda}^{\star}$ for $5$-dimensional asymmetric logistic distribution with parameter vector $\bm{\Psi}_{4}$ in Section~\ref{subsec:asymlog}, where bivariate H\"{u}sler--Reiss stable tail dependence functions are fitted with M-estimator (M) and weighted least squares estimator (WLS) at $k=50,75,\ldots,200$ ($300$ samples of size $n=1000$)}
    \label{fig:5-ML}
\end{figure}

\section{Application}
\label{sec:app}

In this section we consider the average daily discharges recorded at $31$ gauging stations in the upper Danube basin covering parts of Germany, Austria and Switzerland. The topographic map can be found in Figure~1 of \citet{A15} and the data are available in the supplementary materials of that paper. The series at individual stations have lengths from $54$ to $113$ years, with $51$ years of data for all stations from $1960$ to $2010$. The average daily discharges were measured in $m^3/s$, ranging from $20$ to $1400$. Following \citet{A15} and \citet{E20}, we only consider the daily discharges in the summer months (June, July and August) to get rid of the seasonality. The information for summer discharges at each station is given in \mdf{Table~S2 of the Supplementary Material}. 

For $i=1,\ldots,n$, denote the daily mean discharges at the station $j$ on day $i$ by $\xi_{i,j}$ and assume the distribution function of $\bm{\xi}^{i}=(\xi_{i,1},\ldots,\xi_{i,d})$ belongs to the domain of attraction of a multivariate max-stable distribution function. We are concerned about the probability that there will be a flood above a certain (high) threshold for at least one site within the considered stations, precisely, the probability in~\eqref{eq:rareevent}. 
In view of \eqref{eq:rareeventappxi}, the modelling can be split into two parts: fitting the marginal distributions and estimating the simple tree-based multivariate extreme distribution function $G_{M}$.

Extreme discharges often occur in clusters. To remove the clusters and create a sequence of independent and identically distributed variables, we follow the procedure in \citet{A15} for each of the $51$ summer periods. It ranks the data within each series and starts from the day with the highest rank of the period across all series (only one series when extracting univariate events) and takes a window of $9$ days. An event is formed by taking the largest observation within this window. Then the data in this window are deleted and the process is repeated to form a new event. All events are found when no window of $9$ consecutive days remains. 

\subsection{Marginal fitting}

Under the assumption that the distribution of $\bm{\xi}^{i}$ belongs to the domain of attraction of a multivariate max-stable distribution, the marginal distributions are also attracted by univariate extreme value distributions.
In line with asymptotic theory, we model high-threshold excesses by the generalized Pareto distribution to obtain the margins $F_{j}$ for $j=1,\ldots,d$. Since the extremal behavior from any available earlier data does not change relative to the $51$ common years, see \citet{A15}, all the available data for each station are used to do the marginal fitting.

For $j=1,\ldots,d$ and $0 < p < 1$, let $q_{j,p}$ be the empirical $p$-quantile of $\xi_{i,j}$ and put $\mathcal{I}_{j}=\{i\in \{1,\ldots,n\}: \xi_{i,j}>q_{j,p}\}$.
For each station $j$, let $\hat{\sigma}_{j} > 0$ and $\hat{\vartheta}_{j} \in \mathbb{R}$ be the maximum likelihood estimates of the scale and shape parameters of the generalized Pareto distribution fitted to the excesses $\xi_{i,j} - q_{j,p}$ for $i\in \mathcal{I}_{j}$. Let $N_{j}$ and $n_{j}$ denote the sample size at station $j$ and the number of elements in $\mathcal{I}_{j}$, respectively. Then the tail distribution function $1-F_j$ is estimated by
\begin{equation}
    \label{eq:MDFFitting}
     1-\hat{F}_{j}(x) =
    \frac{n_{j}}{N_{j}} \left\{ 
        1 - \operatorname{GPD}(x-q_{j,p};\hat{\sigma}_{j},\hat{\vartheta}_{j})
    \right\},
\end{equation}
where $\operatorname{GPD}(\,\cdot\,;\sigma,\vartheta)$ denotes the cumulative distribution function of the generalized Pareto distribution with scale parameter $\sigma > 0$ and shape parameter $\vartheta \in \mathbb{R}$.

In pursuit of a good fit, an appropriate threshold $q_{j,p}$ has to be found.
By inspection of the mean residual life plot, it was decided to set the threshold for margin $j$ equal to $q_{j,p}$ with $p = 0.9$.
The number of declustered excesses used for the estimation at each station as well as the corresponding thresholds $q_{j,p}$ and maximum likelihood estimates of $\hat{\sigma}_{j}$ and $\hat{\vartheta}_{j}$ of fitted generalized Pareto distribution for $j=1,\ldots,31$ are given in \mdf{Tables~S3 and~S4 of the Supplementary Material} respectively.

\subsection{Estimation of the dependence structure}

In this subsection, we approximate the dependence structure between extremes of the $31$ gauging stations by a multivariate max-stable distributions $G_{M}$ related to a Markov tree. 
After declustering, there are $N=428$ independent observations from the $51$ years of data for all stations. The tree structure corresponding to the physical flow connections at these stations is presented in \mdf{Figure~S5}, while the tree structures $\hat{\tree}_{\lambda}^{\star}$ and $\hat{\tree}_{\tau}^{\star}$~\eqref{eq:tree_star_est} are given in \mdf{Figure~S6 of the Supplementary Material}. We construct the max-stable distribution $G_{M}$ by fitting the H\"{u}sler--Reiss distribution to all pairs of connected variables and estimate their parameters by the moments estimator, the M-estimator and the weighted least squares estimator with $k=65$. 

For model assessment, we show the scatterplots of the pairs $(\hat{\lambda}_{ab}^M, \hat{\lambda}_{ab})$, $(a, b)\in V\times V$ and $a\neq b$, for models based on $\hat{\tree}_{\lambda}^{\star}$ and $\hat{\tree}_{\tau}^{\star}$ in Figure~\ref{fig:ScatterLambda} and \ref{fig:ScatterTau}, comparing the model-based and empirical tail dependence coefficients. The approximation errors $\hat{D}_{\hat{\tree}_{\lambda}^{\star}}$ with bivariate stable tail dependence functions estimated by moments estimator, M-estimator and weighted least squares estimator are $42.10$, $41.41$, $40.34$ while those for $\hat{D}_{\hat{\tree}_{\tau}}$ are $38.18$, $37.67$, $44.98$ respectively. We see that all the models, although with a slight overestimation of the bivariate tail dependence coefficients, seem to behave well in quantifying the bivariate dependence. The estimated GPD-based marginal probability in~\eqref{eq:MDFFitting} and the probability $\pr\left(\xi_{4}>u_{4}~\text{or}~\xi_{7}>u_{7}~\text{or}~\xi_{13}>u_{13}\right)$ of flooding at one or more stations $4$, $7$ and $13$, which are connected differently on the maximum dependence trees $\hat{\tree}_{\lambda}^{\star}$ and $\hat{\tree}_{\tau}^{\star}$ (\mdf{Figure~S6 in the Supplementary Material}),
are given in Tables~\ref{tab:MTDFitting} and~\ref{tab:Danube_Rareevents} respectively, where $u_{j}$ for $j=4,7,13$ are taken as the $0.95$, $0.99$, $0.995$ and $0.999$ empirical quantiles of the $51$-year samples. Compared with the empirical probabilities, the tree-based model yields higher estimates of the flooding probabilities. As shown in Table~\ref{tab:MTDFitting}, this may be partly due to the difference between the non-parametric and GPD-based estimates of the marginal probabilities $F_{j}(u_{j})$ for $j=4,7,13$.

\begin{figure}
    \centering
    \includegraphics[width=0.3\textwidth,height=0.3\textwidth,angle=-90]{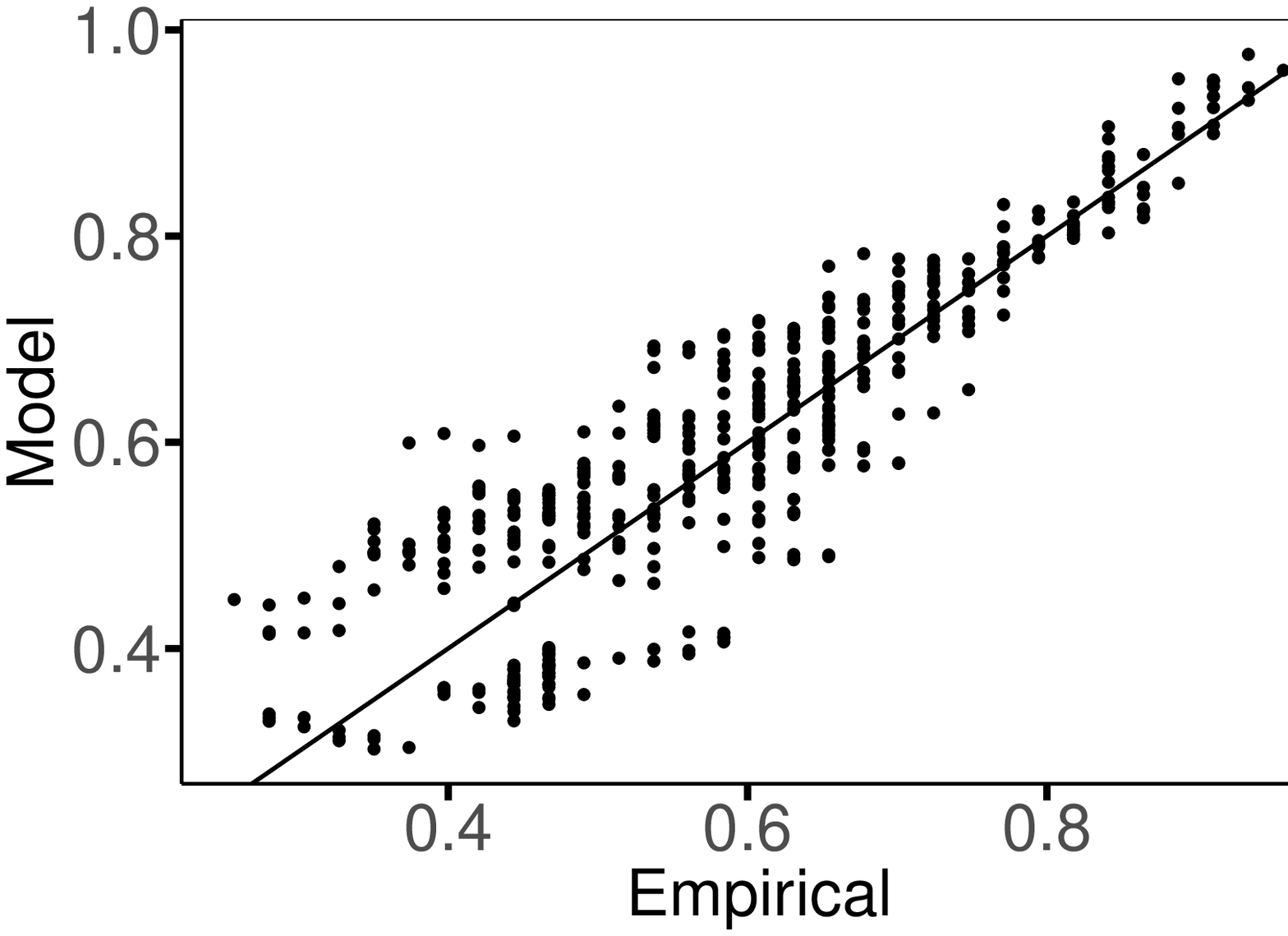}
    \includegraphics[width=0.3\textwidth,height=0.3\textwidth,angle=-90]{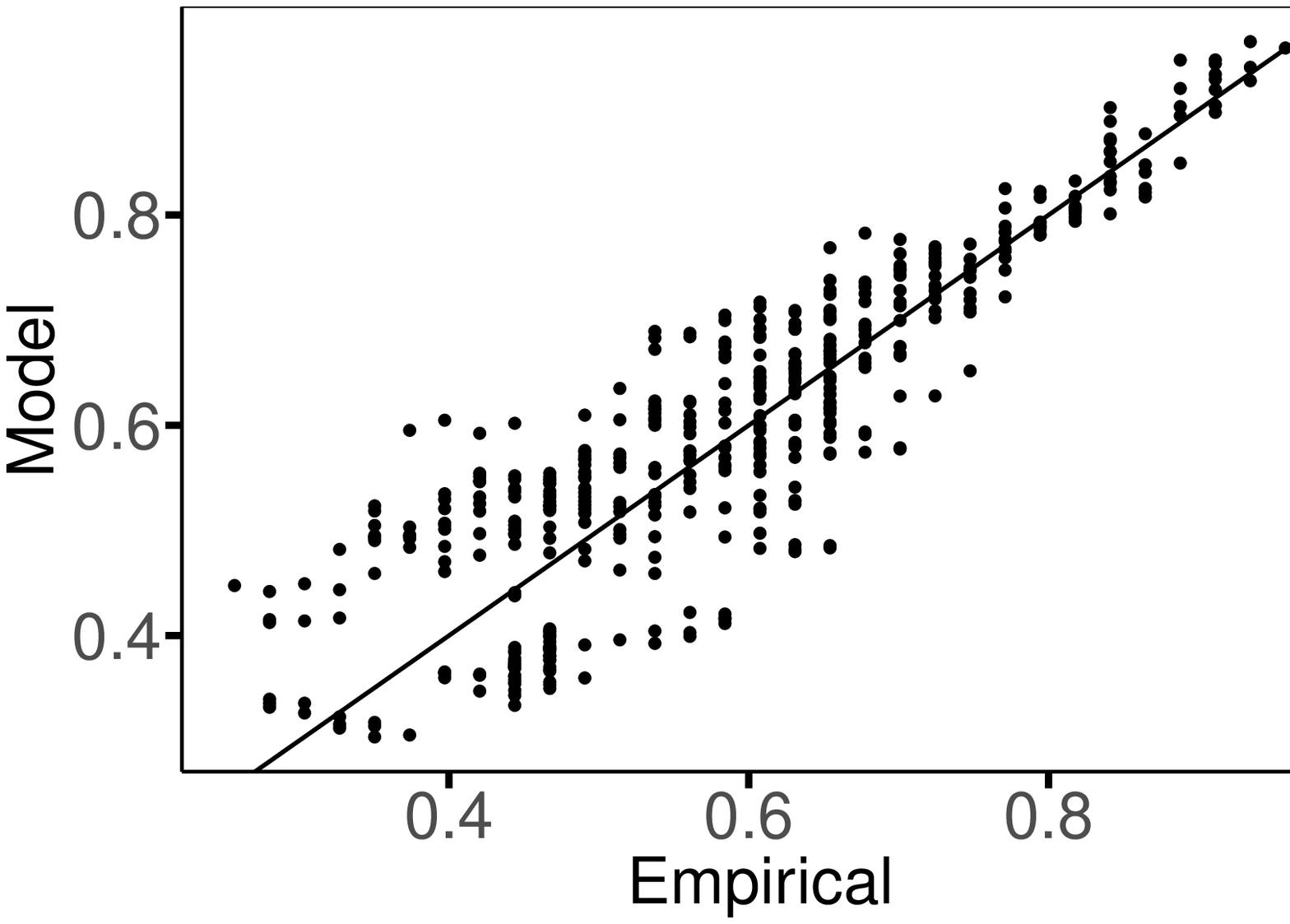}
   \includegraphics[width=0.3\textwidth,height=0.3\textwidth,angle=-90]{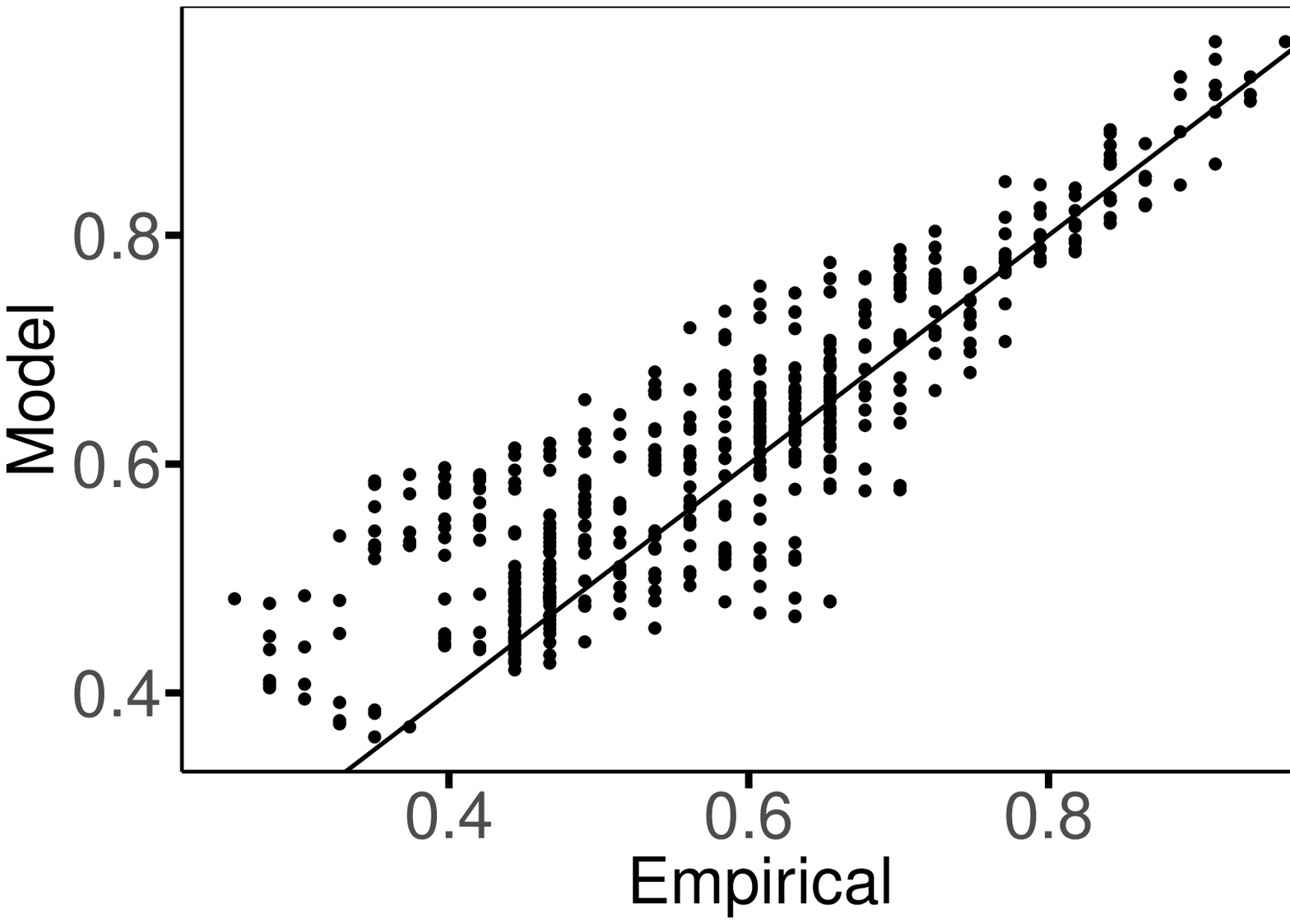}
    \caption{Scatterplots of the empirical upper tail dependence coefficients and model-based estimates of upper tail dependence coefficients $(\lambda_{ab},\lambda_{ab}^M)$ for all $a,b\in V$ and $a\neq b$, where the models are constructed based on tree $\tree_{\lambda}^{\star}$ and the bivariate tail dependence functions are estimated by moments estimator (left), M-estimator (middle) and weighted least squares estimator (right) at $k=65$}
    \label{fig:ScatterLambda}
\end{figure}

\begin{figure}
    \centering
    \includegraphics[width=0.3\textwidth,height=0.3\textwidth,angle=-90]{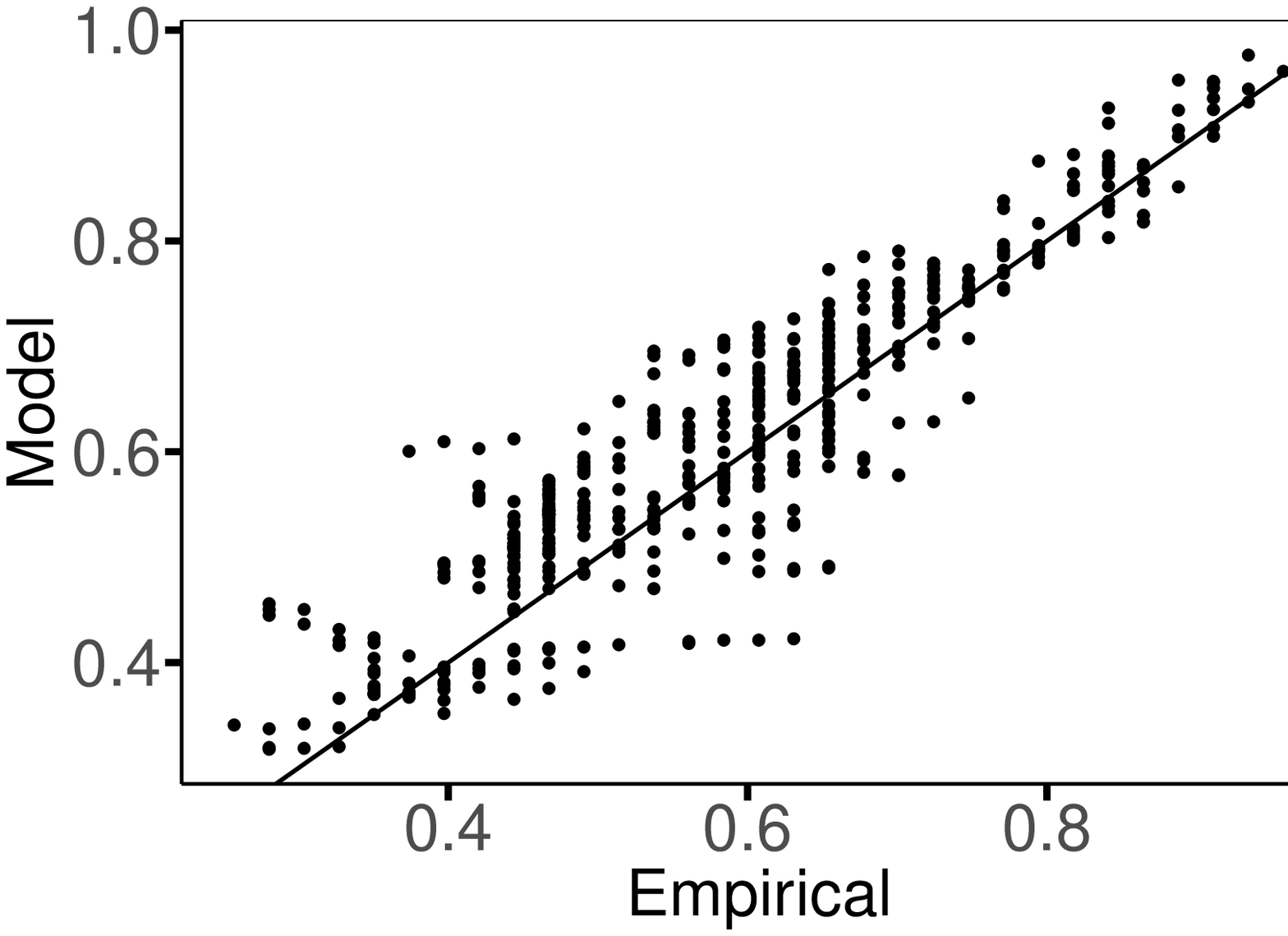}
    \includegraphics[width=0.3\textwidth,height=0.3\textwidth,angle=-90]{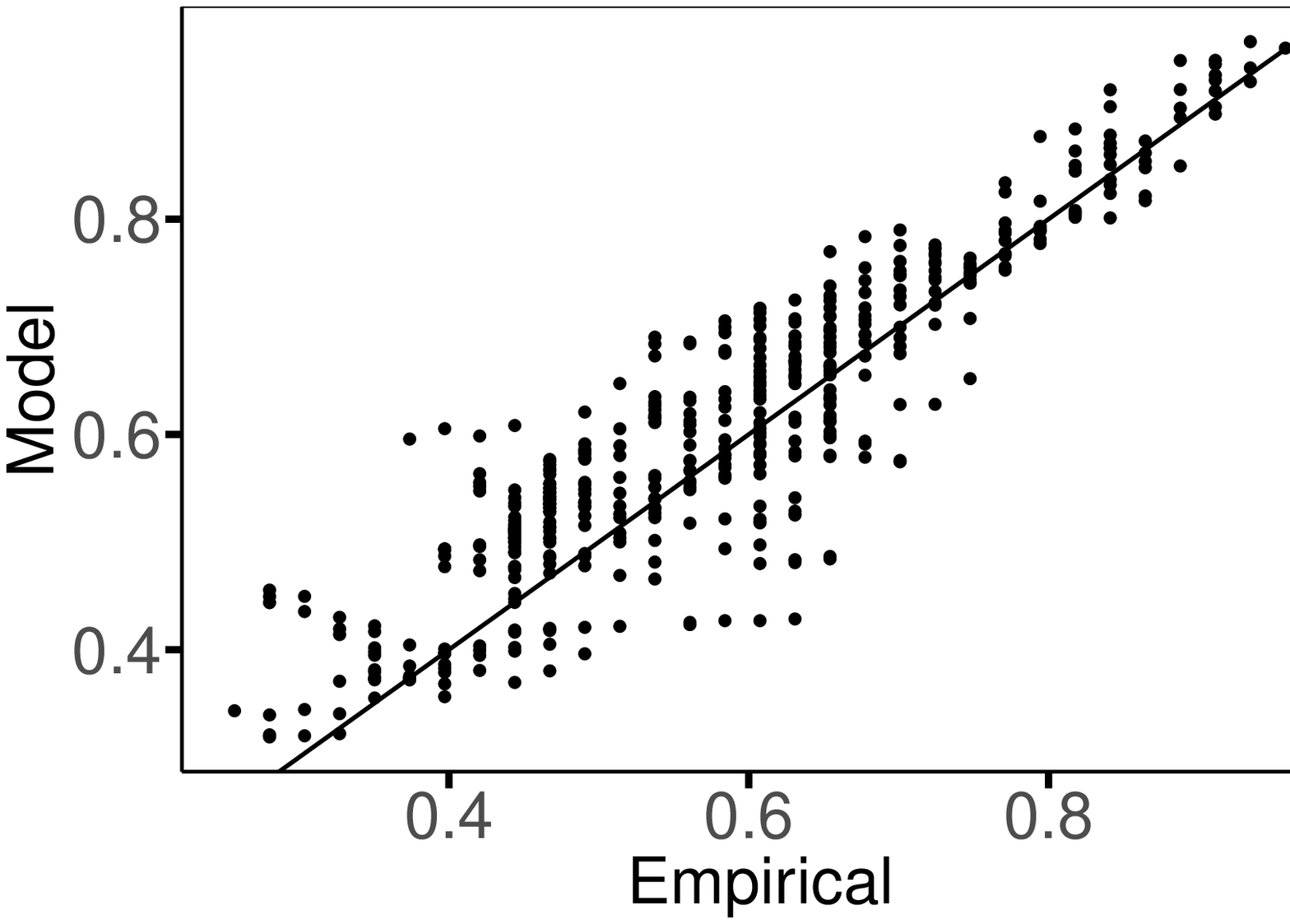}
   \includegraphics[width=0.3\textwidth,height=0.3\textwidth,angle=-90]{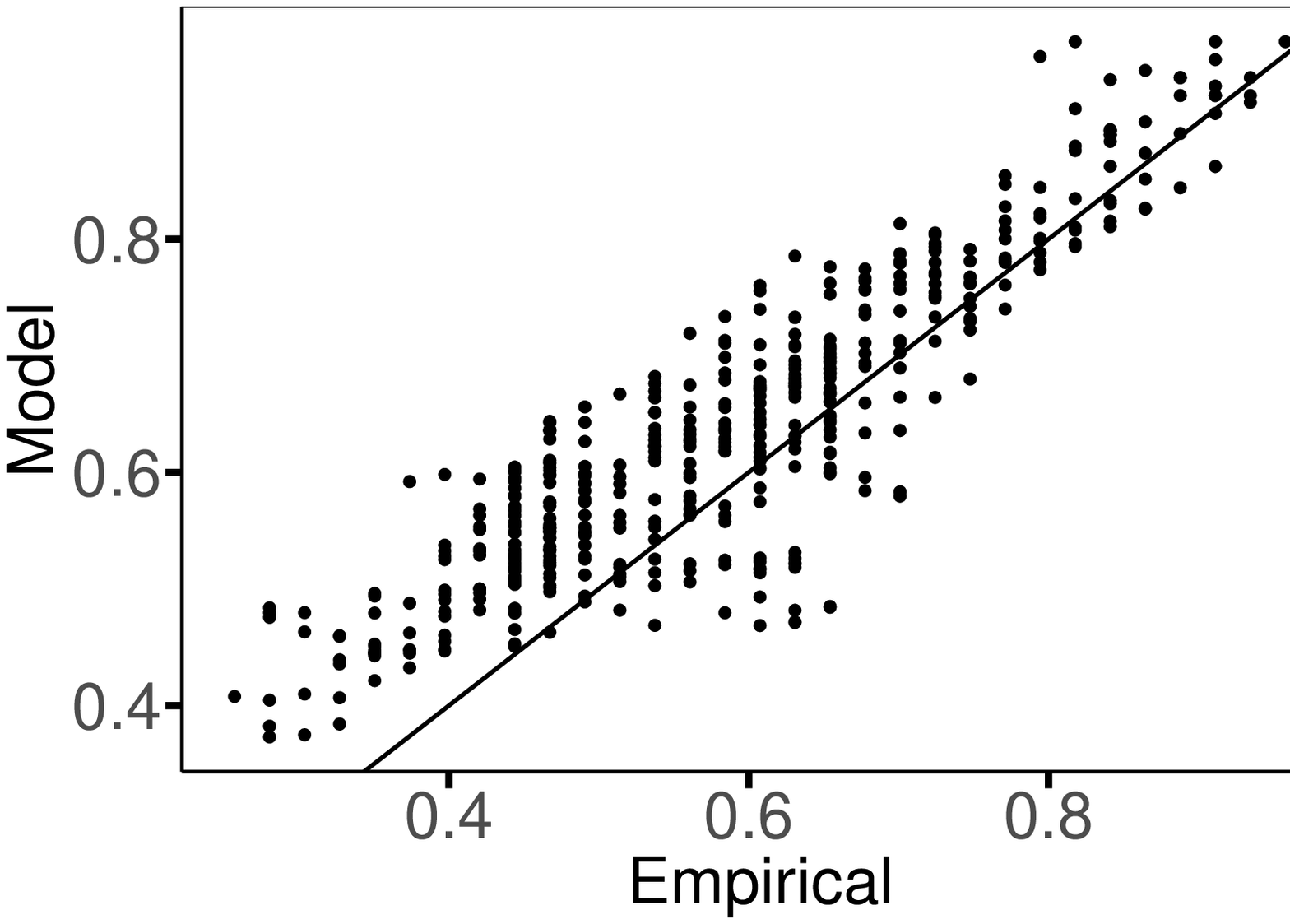}
    \caption{Scatterplots of the empirical upper tail dependence coefficients and model-based estimates of upper tail dependence coefficients $(\lambda_{ab},\lambda_{ab}^M)$ for all $a,b\in V$ and $a\neq b$, where the models are constructed based on tree $\tree_{\tau}^{\star}$ and the bivariate tail dependence function are estimated by moments estimator (left), M-estimator (middle) and weighted least squares estimator (right) at $k=65$}
    \label{fig:ScatterTau}
\end{figure}

\begin{table}[!htbp] \centering 
\caption{GPD-based estimates $1-\hat{F}_{j}(u_{j})$ of marginal tail probability in~\eqref{eq:MDFFitting} with $u_{j}$ taken as the $0.95$, $0.99$, $0.995$ and $0.999$ empirical quantiles of the $51$-year samples for $j=4,7,13$ }
\label{tab:MTDFitting} 
\begin{tabular}{lccccc} 
\toprule
\multicolumn{1}{c}{$u_j$} & $q_{j,0.95}$ & $q_{j,0.99}$ & $q_{j,0.995}$ & $q_{j,0.999}$ \\
\midrule
$1-\hat{F}_{4}$  &  $6.52\%$ & $1.87\%$ & $0.99\%$ & $0.22\%$ \\  
$1-\hat{F}_{7}$  &  $6.66\%$ & $1.69\%$ & $1.03\%$ & $0.15\%$ \\ 
$1-\hat{F}_{13}$ &  $7.02\%$ & $1.90\%$ & $1.18\%$ & $0.33\%$ \\ 
\bottomrule
\end{tabular} 
\end{table}

\begin{table}[!htbp] \centering 
   \caption{Empirical and model-based estimates of the flooding probability \eqref{eq:rareevent} over the $0.95$, $0.99$, $0.995$ and $0.999$ empirical quantiles of the $51$-year samples at stations $4$, $7$ and $13$ in the upper Danube basin, where bivariate H\"{u}sler--Reiss stable tail dependence functions of connected pairs on maximum dependence tree $\hat{\tree}_{\lambda}^{\star}$ and $\hat{\tree}_{\tau}^{\star}$ are fitted with moments estimator (MM), M-estimator (M) and weighted least squares estimator (WLS) at $k=65$}
   \label{tab:Danube_Rareevents} 
\begin{tabular}{clcccc} 
 \toprule
  & $u_j$ & $q_{j,0.95}$ & $q_{j,0.99}$ & $q_{j,0.995}$ & $q_{j,0.999}$ \\
 \midrule
  \multicolumn{2}{c}{Empirical} & $8.44\%$ & $1.88\%$ & $1.02\%$ & $0.21\%$ \\ 
 \midrule
 \multirow{4}{*}{$\hat{\tree}_{\lambda}^{\star}$} & 
   MM &   $9.71\%$ & $2.65\%$ & $1.58\%$ & $0.38\%$ \\  
 & M & $9.72\%$ & $2.65\%$ & $1.58\%$ & $0.38\%$ \\ 
 & WLS & $9.62\%$ & $2.62\%$ & $1.56\%$ & $0.37\%$ \\ 
 \midrule
 \multirow{4}{*}{$\hat{\tree}_{\tau}^{\star}$} & 
   MM & $9.54\%$ & $2.61\%$ & $1.55\%$ & $0.37\%$ \\ 
 & M & $9.55\%$ & $2.61\%$ & $1.55\%$ & $0.37\%$ \\ 
 & WLS & $9.26\%$ & $2.53\%$ & $1.50\%$ & $0.37\%$ \\ 
 \bottomrule
\end{tabular} 
\end{table}

\section*{Acknowledgments}
We would like to thank the two Referees and the Associate Editor for their constructive suggestions and comments, which helped us greatly to enhance the manuscript.
The research of Shuang Hu is financially supported by the State Scholarship Fund (CSC No.202106990036) from the China Scholarship Council.

\mdf{\section*{Supporting information}
Additional supporting information can be found online in the Supplementary Material including detailed proofs and discussions for Sections~\ref{subsec:ApproErr} and~\ref{sec:Examples}, examples, and additional
tables and figures for the simulation and application studies.}

\bibliography{sample}

\section*{Address}
\Address

\appendix
\clearpage

\setcounter{equation}{0}
\renewcommand{\theequation}{A\arabic {equation}} 

\section{Appendix}
\label{sec:SM}

\subsection{Proofs}
\label{subsec:LemProf}

\begin{proof}[Proof of Proposition~\ref{pro:G_M}]
By assumption, there exists a regularly varying Markov tree $\bm{Y}$ satisfying Condition~\ref{condition} and \eqref{eq:RVMar} such that $\bm{Y}\in D(G)$, where $G$ is a simple multivariate extremes value distribution transformed from $H$. Moreover, $G$ and $H$ have the same stable tail dependence functions. Thus it is sufficient to derive the result in terms of $G$.
Let $v_{1},\ldots,v_{d}$ be an arbitrary permutation of $1,\ldots,d$. Since $\bm{Y}$ satisfies Condition~\ref{condition}, it follows from~\eqref{eq:RVMar} and Theorem~2 in \citet{S20} that
\begin{align}
\label{eq:cvg_stdf}
\nonumber
\lefteqn{n\pr\left(\bigcup_{v=1}^{d} \{Y_{v}>n y_{v}\}\right)} \\
\nonumber
&= n\sum_{i=1}^{d-1} \pr( Y_{v_{i}}>n y_{v_{i}}, Y_{v_{j}}\leq n y_{v_{j}}, j=i+1,\ldots,d)+n \pr(Y_{v_{d}}>n y_{v_{d}})\\
\nonumber
&= \sum_{i=1}^{d-1} n\pr(Y_{v_{i}}>n y_{v_{i}}) \pr(Y_{v_{j}}\leq n y_{v_{j}}, j=i+1,\ldots,d \mid Y_{v_{i}}>n y_{v_{i}})+n \pr(Y_{v_{d}}> n y_{v_{d}})\\
&\to \sum_{i=1}^{d} y_{v_{i}}^{-1}  \pr\left(W \Theta_{v_{i},v_{j}}\leq y_{v_{j}}/y_{v_{i}}, j=i+1,\ldots, d\right)
\end{align}
as $n\to\infty$, where the probability in the sum \eqref{eq:cvg_stdf} is interpreted as $1$ for $i = d$. Recall that $W$ is unit-Pareto distributed and is independent of $\bm{\Theta}_{v_{i}}$. By the fact that $1/W$ is uniformly distributed on $[0,1]$, we have
\begin{align}
\label{eq:expection}
\nonumber
\lefteqn{ y_{v_{i}}^{-1} \pr\left(W \Theta_{v_{i},v_{j}}\leq y_{v_{j}}/y_{v_{i}}, j=i+1,\ldots, d\right)} \\
\nonumber
&= y_{v_{i}}^{-1} \int_{0}^{1} \pr\left\{\Theta_{v_{i},v_{j}}\leq (y_{v_{j}}u) /y_{v_{i}} , j=i+1,\ldots, d\right\} \d{u}\\
\nonumber
&= y_{v_{i}}^{-1} \int_{0}^{1} \pr\left\{\max_{j=i+1,\ldots,d} \left(y_{v_j}^{-1} \Theta_{v_{i},v_{j}} \right)\leq y_{v_{i}}^{-1} u\right\} \d{u}\\
\nonumber
&= \int_{0}^{1/y_{v_i}} \pr\left\{\max_{j=i+1,\ldots,d} \left(y_{v_j}^{-1} \Theta_{v_{i},v_{j}} \right)\leq u\right\} \d{u}\\
\nonumber
&=  \int_{0}^{\infty} \pr\left\{\max_{j=i+1,\ldots, d} \left(y_{v_j}^{-1} \Theta_{v_{i},v_{j}} \right) \leq u \leq y_{v_{i}}^{-1} \Theta_{v_{i},v_{i}}\right\} \d{u}\\
&= \E\left\{\max_{j=i,\ldots, d} \left(y_{v_j}^{-1} \Theta_{v_{i},v_{j}} \right)-\max_{j=i+1,\ldots, d}\left(y_{v_j}^{-1} \Theta_{v_{i},v_{j}} \right)\right\},
\end{align}
where the last step follows from $\int_0^\infty \pr(A \le u \le B) \d u = \E[\int_0^\infty \I(A \le u \le B) \d{u}] = \E[\max(A,B)-A]$ for non-negative random variables $A$ and $B$.
By \eqref{eq:G_M} with $a_{n,v}=n$ and $b_{n,v}=0$, \eqref{eq:cvg_stdf}--\eqref{eq:expection} and \eqref{eq:Maxstable_stdf}, we have
\begin{align}
\label{eq:stdf_G_M_theta}
\nonumber
\ell(\bm{y})
&=-\log G(1/y_1,\ldots,1/y_d)=\lim_{n\to\infty} n \pr\left(\bigcup_{v=1}^{d} \{Y_{v}>n y_{v}\}\right) \\
&= \sum_{i=1}^{d} \E\left\{\max_{j=i,\ldots, d} \left(y_{v_{j}} \Theta_{v_{i},v_{j}}\right)-\max_{j=i+1,\ldots, d} \left(y_{v_{j}} \Theta_{v_{i},v_{j}}\right)\right\},
\end{align}
where the maximum over the empty set is defined as zero. The identity $\Theta_{v_{i},v_{j}}=\prod_{e\in \pth{v_{i}}{v_{j}} } M_{e}$ yields the desired expression of $\ell$, which is totally determined by the distributions of the increments $M_{e}$ with $e\in E$. 

Assume $\E (M_{ab})=1$ for all $(a,b)\in E$. Then $\E (\Theta_{u,v})=\prod_{(a,b)\in \pth{u}{v}}\E (M_{ab})=1$ for $u,v\in V$ by the independence of $M_{e}$ for $e\in E$. Thus it follows from Corollary~3 in \citet{S20} that the root-change formula 
\begin{equation*}
    \E\left\{g(\bm{\Theta}_{u})\right\}=\frac{\E\left\{g\left(\bm{\Theta}_{v}/\Theta_{v,u}\right) \Theta_{v,u} \right\}}{\E(\Theta_{v,u})}
\end{equation*}
holds for any Borel measurable function $g:[0,\infty)^{d}\setminus\{\bm{0}\}\rightarrow [0,\infty]$, where $\bm{0}$ denotes the $d$-dimensional vector whose components are zero. For fixed $\bm{y}=(y_{v},v\in V)$, taking $g_{i}(\bm{x})=\max_{j=i,\ldots,d}\left(y_{v_{j}}x_{v_{j}}\right)$ for $i=2,\ldots,d$ and using the root-change formula we have
\begin{align*}
\E\left\{\max_{j=i,\ldots, d} \left( y_{v_{j}}\Theta_{v_{i},v_{j}}\right)\right\}
&= \E\left[\max_{j=i,\ldots, d} \left\{ y_{v_{j}}\Theta_{v_{i},v_{j}} \I (\Theta_{v_{i},v_{i-1}}>0) \right\}\right]
\\
&= \E\left[\Theta_{v_{i},v_{i-1}} \max_{j=i,\ldots,d}\left\{ \left(y_{v_{j}} \Theta_{v_{i},v_{j}}\right)/ \Theta_{v_{i},v_{i-1}} \right\}\right]
\\
&= \E\left\{\max_{j=i,\ldots,d} \left( y_{v_{j}}\Theta_{v_{i-1},v_{j}}\right)\right\} \E\left( \Theta_{v_{i},v_{i-1}}\right)\\
&= \E\left\{\max_{j=i,\ldots,d}\left( y_{v_{j}}\Theta_{v_{i-1},v_{j}}\right)\right\},
\end{align*}
where the first step holds since $\E(\Theta_{v_{i-1},v_{i}})=1$ implies $\pr(\Theta_{v_{i},v_{i-1}}>0)=1$ by Corollary~3 in \citet{S20}. Consequently, in view of \eqref{eq:stdf_G_M_theta} and the above equality we obtain
\begin{align*}
\ell(\bm{y})
&=\sum_{i=1}^{d} \E\left\{\max_{j=i,\ldots, d} \left( y_{v_{j}}\Theta_{v_{i},v_{j}}\right)\right\}-\sum_{i=2}^{d} \E\left\{\max_{j=i,\ldots, d} \left( y_{v_{j}}\Theta_{v_{i-1},v_{j}}\right)\right\}\\
&= \E\left\{\max_{j=1,\ldots, d} \left( y_{v_{j}}\Theta_{v_{1},v_{j}}\right)\right\}.
\end{align*}
The final identity holds because $v_1,\ldots,v_d$ was an arbitrary permutation of $1,\ldots,d$.

For the converse, since $H_{v}$, $v\in V$, are univariate extreme value distributions, any max-stable distribution $H$ can be transferred to a simple max-stable distribution by monotone marginal transformations. For brevity, we assume $H_{v}(x)=1-\exp(-1/x)$ for $x>0$ and all $v\in V$ here. In the following, we show that we can construct a regularly varying Markov tree $\bm{Y}$ on the tree $\tree=(V,E )$ in~\eqref{eq:stdf_G} such that $\bm{Y}$ is in the domain of attraction of $H$ and satisfies Condition~\ref{condition} and equation~\eqref{eq:RVMar}. 

Given nonnegative $M_{ab}$ with $\E (M_{ab})\le 1$, define for every $(a,b)\in E$ the function
\begin{equation}
\label{eq:CstSTDF}
\ell_{ab}(x,y)=x+y-\E\left[\min(x,y M_{ab})\right], \quad x,y\ge 0.
\end{equation}
Since $M_{ab}$ and $M_{ba}$ satisfy the relation in~\eqref{eq:changeDir}, we have $\ell_{ab}(x,y)=\ell_{ba}(y,x)$. Moreover, one can easily check that for each $(a,b)\in E$ the properties (L1)--(L4) in \citet[page~257]{B04} hold for $\ell_{ab}(x,y)$ and thus it is a well-defined stable tail dependence function. As a consequence, the function $H_{ab}(x,y):=\exp[-\ell_{ab}(1/x,1/y)]$ a bivariate max-stable distribution with unit-Fr\'{e}chet marginal distributions.

Let $\bm{Y}=(Y_v,v\in V)$ be a Markov tree on $\tree=(V,E)$ satisfying the global Markov property with respect to $\tree$. Assume the distribution function of the adjacent pair $(Y_a,Y_b)$ for every $(a,b)\in E$ is given by $H_{ab}(x,y)$. Then by Example~3 in~\citet{S20} and the discussion before this proposition, $\bm{Y}$ is a regularly varying Markov tree satisfying Condition~\ref{condition} and Equation~\eqref{eq:RVMar} and belongs to the domain of attraction of a multivariate max-stable distribution, which is $H$ with stable tail dependence function given by~\eqref{eq:stdf_G} from the proof of the first part statement of this proposition.
\end{proof}

\begin{proof}[Proof of Proposition~\ref{pro:bi_TDC}]
For $y_{a},y_{b}\ge 0$, setting $d=2$ in \eqref{eq:stdf_G} and using the relation $\max(x,y)=x+y-\min(x,y)$ yields
\begin{align*}
\ell_{ab}(y_{a},y_{b})&=\E\left\{\max \left(y_{a}, y_{b} \prod_{e\in \pth{a}{b} } M_{e}\right)- y_{b} \prod_{e\in \pth{a}{b} } M_{e}\right\}+ y_{b}\\
&= y_{a}+y_{b}-\E\left\{ \min\left(y_{a},  y_{b} \prod_{e\in \pth{a}{b} } M_{e}\right)\right\}.
\end{align*}
The desired result for $\lambda_{ab}$ follows since $\lambda_{ab}=2-\ell_{ab}(1,1)$. The proof is complete. 
\end{proof}

\begin{proof}[Proof of Proposition~\ref{pro:TDC_ieq}]
Let $\bm{Y}=(Y_{v},v\in V)$ be the Markov tree on $\tree$ associated to $H$ in Definition~\ref{def:maxstabletree}. For each pair $(a,b)\in E$, the tail dependence coefficient of $(Y_{a},Y_{b})$ equals the one of $(X_{a},X_{b})$. Thus it is sufficient to derive the result in terms of $\bm{Y}$. 

For $\delta>0$, note that
\begin{equation}
\label{eq:ConUpperBound}
    n\pr(Y_{a}>n\delta,Y_{b}>n)\le n\pr(Y_{b}>n)\to 1
\end{equation}
as $n\to\infty$ by \eqref{eq:RVMar}. From the equality 
\[
    \int_{0}^{x}\pr(X>u)\d u=\E\{\min(X,x)\}
\]
for a non-negative random variable $X$ and for $x>0$, we find by \eqref{eq:RVMar}, \eqref{eq:W_i} and the fact that $1/W$ is uniformly distributed on $[0,1]$ that
\begin{align*}
n\pr(Y_{a}>n\delta,Y_{b}>n)
&=n\pr(Y_{a}>n\delta) \cdot \pr(Y_{b}>n\mid Y_{a}>n\delta)\\
&\to \delta^{-1} \pr(W M_{ab}> 1/\delta) = \delta^{-1} \pr(\delta M_{ab} > 1/W ) \\
&= \delta^{-1} \E\{\min(\delta M_{ab}, 1)\} = \E\{ \min(M_{ab},1/\delta)\}
\end{align*}
as $n\to\infty$. Combining the above display with \eqref{eq:ConUpperBound} and letting $\delta \to 0$ we obtain 
\begin{equation}
\label{eq:EMab}
    \E(M_{ab})\le 1, \qquad (a,b)\in E. 
\end{equation}

The concavity of the function $y \mapsto \min(1, xy)$ implies that
\begin{equation}
\label{eq:min1xyconcave}
    \min(1, xy) \le \min(1, x) - \left(1-y\right) x \I(x < 1), 
    \qquad (x, y) \in [0, \infty)^2,
\end{equation}
as can be confirmed by a case-by-case analysis.
Let $u$ be a node on the path from $a$ to $b$ and write $A = \prod_{e\in\pth{a}{u}}M_{e}$ and $B = \prod_{e\in\pth{u}{b}}M_{e}$. We have $\E[B] \le 1$ by \eqref{eq:EMab} and the independence of the increments $M_{e}$ for $e \in E$. Using Proposition~\ref{pro:bi_TDC}, we have, since $A$ and $B$ are independent and non-negative,
\begin{align}
\label{eq:AB}
    \lambda_{ab}
    = \E\{ \min(1, AB) \}
    &\le \E\{ \min(1, A) \} - \left\{1 - \E(B)\right\} \E\{ A \I(A < 1) \} \\
\nonumber
    &\le \E\{ \min(1, A) \}
    = \lambda_{au}.
\end{align}
Interchanging the roles of $A$ and $B$ yields $\lambda_{ab} \le \lambda_{ub}$.
The inequality \eqref{eq:min1xyconcave} is strict whenever $x > 1 > xy$ or $x < 1 < xy$. If $\pr(1-\eps < M_e < 1) > 0$ and $\pr(1 < M_{e} < 1+\eps) > 0$ for every $e \in E$ and every $\eps > 0$, then also $\pr(A > 1 > AB) > 0$ and $\pr(A < 1 < AB) > 0$. But then $\min(1, AB) < \min(1, A) - \left(1-B\right) A \I(A < 1)$ with positive probability, yielding a strict inequality in \eqref{eq:AB}.

As for the lower bound, the inequality $\min(1,s)\min(1,t)\le\min(1,st)$ for non-negative reals $s$ and $t$ and the independence of $M_{e}$, $e\in E$, yields, with $A$ and $B$ as above,
\[
    \lambda_{ab} 
    = \E\{\min(1,AB)\} 
    \ge \E\{\min(1,A)\min(1,B) \}
    = \E\{\min(1,A)\} \E\{\min(1,B)\} 
    = \lambda_{au} \lambda_{ub},
\]
as required. The proof is complete.
\end{proof}

\begin{proof}[Derivation of equation~\eqref{eq:stdf_G_M_ML}]
For $x_{v}\in \mathbb{R}$ such that $v\in V$, we have, by the maximum-minimums identity
\[
\max_{v\in V}x_{v}=\sum_{\emptyset \ne U\subseteq V} (-1)^{|U|+1} \min_{v\in U}x_{v}. 
\]
Recall that $V_{i}=\{i,\ldots,d\}$ and $e=(e_{p},e_{s})$ for $e\in E$. Without loss of generality, let $v_i=i$ for the permutation $v_{1},\ldots,v_{d}$ of $1,\ldots,d$ in Corollary~\ref{pro:GM_Markov} and \eqref{eq:stdf_G}. Then, in view of the distribution of $M_{e}$ for $e\in E$ given in equation~\eqref{eq:increase_ML} and by applying the above display, we have from Corollary~\ref{pro:GM_Markov} that
\begin {align*}
\ell^{M}(\bm{y})=&\sum_{i=1}^{d} \E\left\{\max_{j=i,\ldots,d} \left(y_{j} \prod_{e\in \pth{i}{j}} M_{e}\right)\right\}-\sum_{i=2}^{d} \E\left\{\max_{j=i,\ldots,d} \left(y_{j} \prod_{e\in \pth{(i-1)}{j}} M_{e}\right)\right\}\\
=& \sum_{i=1}^{d} \sum_{\emptyset \ne U\subseteq V_{i}} (-1)^{|U|+1} \E\left\{\min_{j\in U}\left(y_{j} \prod_{e\in \pth{i}{j}} M_{e}\right)\right\}\\
&- \sum_{i=2}^{d} \sum_{\emptyset \ne U\subseteq V_{i}} (-1)^{|U|+1} \E\left\{\min_{j\in U}\left(y_{j} \prod_{e\in \pth{(i-1)}{j}} M_{e}\right)\right\}\\
 =& \sum_{i=1}^{d} \psi_{i}^{-1} \sum_{\emptyset \ne U\subseteq V_{i}} (-1)^{|U|+1} \left(\prod_{e\in \bigcup_{j\in U}\pth{i}{j}} \psi_{e_{p}}\right)  \left\{\min_{j\in U}\left(y_{j} \psi_{j}\right)\right\}\\
&- \sum_{i=2}^{d} \psi_{i-1}^{-1} \sum_{\emptyset \ne U\subseteq V_{i}} (-1)^{|U|+1} \left(\prod_{e\in \bigcup_{j\in U}\pth{(i-1)}{j}} \psi_{e_{p}}\right) \left\{\min_{j\in U}\left(y_{j} \psi_{j}\right)\right\}.
\qedhere
\end {align*}
\end{proof}

\begin{proof}[Proof of Proposition~\ref{pro:uniqueness}]
Let $\tree' = (V, E')$ be any tree with the same node set $V$ as the true tree $\tree = (V, E)$. As in the proof of Proposition~5 in \citet{EV20}, an application of Hall's marriage theorem \citep{H35} yields a bijection $f: E \to E'$ with the property that every edge $(u, v) \in E$ is mapped to an edge $f(u,v) = (u', v') \in E'$ such that $(u, v)$ belongs to the path $\pth{u'}{v'}$ in $\tree$. By Proposition~\ref{pro:TDC_ieq}, $\lambda_{uv} \ge \lambda_{f(u,v)}$ and thus
\[
    S_{\tree_{\lambda}} 
    = \sum_{(u,v) \in E} \lambda_{uv}
    \ge \sum_{(u,v) \in E} \lambda_{f(u,v)}
    = \sum_{(u',v') \in E'} \lambda_{u'v'}
    = S_{\tree_{\lambda}'}.
\]
It follows that $\tree$ belongs to the collection $\bm{T}_{\lambda}^{\star}$ of maximum tail dependence trees.

Suppose in addition that $\lambda_{ab} < \min(\lambda_{au}, \lambda_{ub})$ for any triple of distinct nodes $a,u,b$ such that $u$ lies on the path between $a$ and $b$. If $\tree' = (V, E')$ is different from $\tree$, then there is at least one edge $(u, v) \in E$ such that $f(u,v) = (u',v')$ is different from $(u,v)$. By the additional condition, we then have $\lambda_{uv} > \lambda_{u'v'}$. The inequality in the above display is thus strict, implying that $\tree$ is the unique maximizer of the sum of edge weights. Moreover, $\tree=(V,E)$ must be the unique tree associated with $G$. Otherwise, assume $\tree'=(V,E')$ is another tree consistent with $G$. Then there is at least one edge $(a,b)\in E$ such that $a$ and $b$ are nonadjacent on $\tree'$, that is, there exists at least one node $c$ on the path $\pth{a}{b}$ of $\tree'$. Thus $\lambda_{ab}<\min(\lambda_{ac}, \lambda_{bc})$ since $\pth{a}{c}\subset \pth{a}{b}$ on $\tree'$. On the other hand, $c$ must be nonadjacent with $a$ or $b$ on $\tree$ as there is only one path from $u$ to $v$ for each pair of nodes $u,v\in V$; hence we have $\lambda_{ac}<\lambda_{ab}$ or $\lambda_{bc}<\lambda_{ab}$ since $(a,b)\in \pth{a}{c}\cup \pth{b}{c}$. This leads to a contradiction in view of the conclusion $\lambda_{ab}<\min(\lambda_{ac}, \lambda_{bc})$ that we obtained previously.
\end{proof}

\begin{proof}[Proof of Proposition~\ref{pro:consitency_lambda}] 
Let $\tree^{\star}=(V,E^{\star})$ be an arbitrary tree in $\bm{T}_{\lambda}^{\star}$. We show that for any tree $\tree'=(V,E')$ not in $\bm{T}_{\lambda}^{\star}$, we have
\begin{align}
\label{eq:MDTLambdaCon}
    \pr\left(\hat{S}_{\tree_{\lambda}^{\prime}} < \hat{S}_{\tree_{\lambda}^{\star}}\right)\to 1,
    \quad n \to \infty.
\end{align}
By definition of $\hat{\bm{T}}_{\lambda}^{\star}$, this implies
\[ 
    \pr \left( \tree^{\prime} \in \hat{\bm{T}}_{\lambda}^{\star} \right) \to 0, 
    \quad n \to \infty. 
\]
As this is true for any $\tree'$ not in $\bm{T}_{\lambda}^{\star}$, we can then conclude that $\pr \left( \hat{\bm{T}}_{\lambda}^{\star} \subseteq \bm{T}_{\lambda}^{\star} \right) \to 1$ as $n \to \infty$.

To show \eqref{eq:MDTLambdaCon}, note that, by definition of $\bm{T}_{\lambda}^{\star}$, we have
\[
    S_{\tree_{\lambda}^{\prime}} < S_{\tree_{\lambda}^{\star}}.
\]
The consistency of the empirical tail dependence coefficients $\hat{\lambda}_e$ for $e \in E$ \citep{D98,SS06} implies
\[
\hat{S}_{\tree_{\lambda}^{\prime}}-\hat{S}_{\tree_{\lambda}^{\star}}
=\sum_{e\in E^{\prime}} \hat{\lambda}_{e}- \sum_{e\in E^{\star}} \hat{\lambda}_{e}
=\sum_{e\in E^{\prime}} \lambda_{e}- \sum_{e\in E^{\star}} \lambda_{e} +o_{p}(1)
= S_{\tree_{\lambda}^{\prime}}-S_{\tree_{\lambda}^{\star}}+o_{p}(1),
\quad n \to \infty,
\]
yielding the convergence in \eqref{eq:MDTLambdaCon}, as required.

If $\bm{T}_{\lambda}^{\star}$ is a singleton, say $\{\tree_{\lambda}^{\star}\}$, then the relation $\hat {\bm{T}}_{\lambda}^{\star} \subseteq \bm{T}_{\lambda}^{\star}$ implies that $\hat{\bm{T}}_{\lambda}^{\star}$ (which is non-empty since it is the set of maximizers over a finite domain) is a singleton too and that its only element, say $\hat{\tree}_{\lambda}^{\star}$, is equal to $\tree_{\lambda}^{\star}$.
\end{proof}

\begin{proof}[Proof of Theorem~\ref{thm:estdf:AN}]
By Theorem~5.1 in \citet{BSV14}, we have weak convergence of $\alpha_n = \sqrt{k} (\hat{\ell}_{n,k} - \ell)$ in the so-called hypi-topology on the set of bounded functions on $[0, 1]^d$ to a process which, thanks to the condition on the $\mu_r$-almost everywhere continuity of $\dot{\ell}_j$, is $\mu_r$-almost everywhere equal to $\alpha$ in \eqref{eq:limproc}. By Corollary~3.3 in the same article and again by the assumption on $\dot{\ell}_j$, we obtain the weak convergence of $\alpha_n$ to $\alpha$ in $L^2(\mu_r)$; indeed, with probability one, the process $\alpha$ is continuous $\mu_{r}$-almost everywhere, thanks to the assumption on $\dot{\ell}_j$ and the fact that $W(\bm{0}) = 0$. The linear functional $L^2(\mu_r) \to \mathbb{R} : f \mapsto \int f \psi_r \, \d \mu_r$ for $r \in \{1,\ldots,q\}$ being continuous, an application of the continuous mapping theorem yields the result. The limit in~\eqref{eq:estdf:AN} is a linear transformation of a Gaussian process and therefore Gaussian as well. The formula for the covariance matrix follows by Fubini's theorem.

An alternative proof is to proceed by a Skorohod construction as in the proof of Proposition~7.3 in \citet{EKS12}. The term involving the partial derivatives $\dot{\ell}_j$ is handled via the dominated convergence theorem: the pointwise convergence for $\mu_r$-almost every $\bm{x} \in [0, 1]^d$ holds true thanks to the condition on the partial derivatives $\dot{\ell}_j$.
\end{proof}

\begin{proof}[Proof of Corollary~\ref{cor:AN}]
For $e = (a, b) \in E$, the map
\[
    T_{e} : [0, 1]^2 \to [0, 1]^d : 
    (x_a, x_b) \mapsto x_{a} \bm{e}_a + x_{b} \bm{e}_{b}
\]
pushes the measure $\nu_{e}$ on $[0, 1]^2$ forward to the measure $\mu_{e} = T_{e}\#\nu_{e}$ on $[0, 1]^d$ given by
\[
    \mu_{e}(A) = \nu_{e}\left(\left\{
        (x_a, x_b) \in [0, 1]^2 : x_{a} \bm{e}_a + x_{b} \bm{e}_{b} \in A 
    \right\}\right)
\]
for Borel sets $A \subseteq [0, 1]^d$.
The support of $\mu_{e}$ is contained in $\{\bm{x} \in [0, 1]^d : x_j = 0 ~\text{if}~ j \ne a ~\text{and}~ j \ne b\}$. By the change-of-variables formula, integrals with respect to $\mu_{e}$ and $\nu_{e}$ are related through $\int_{[0, 1]^d} f(\bm{x}) \, \mu_{e}(\d \bm{x}) = \int_{[0, 1]^2} f(T_{e}(x_{a}, x_{b})) \, \nu_{e}(\d(x_{a}, x_{b}))$. 

Put $\alpha_{n} = \sqrt{k} (\hat{\ell}_{n,k} - \ell)$. By Slutsky's lemma, the $o_p(\bm{1})$ term in Assumption~\ref{ass:expansion} does not matter for the convergence in distribution. Since $\hat{\ell}_{n,k,e}(x_a, x_b) = \hat{\ell}_{n,k}(T_{e}(x_{a}, x_{b}))$ and similarly for $\ell_{e}$ and $\ell$, the main term on the right-hand side of \eqref{eq:expansion} can be written as
\[
    \int_{[0, 1]^d} 
        \alpha_n(\bm{x}) \,
        \bm{\psi}_{e}(x_{a}, x_{b}) \,
    \mu_{e}(\d \bm{x}).
\]
This is a random vector of dimension $p_{e}$, each component of which has the same form as the entries of the random vector on the left-hand side of \eqref{eq:estdf:AN}. Indeed, the measure $\mu_{r}$ in \eqref{eq:estdf:AN} is equal to the measure $\mu_{e}$ here, while the function $\psi_{r}$ in \eqref{eq:estdf:AN} is the equal to the function $\bm{x} \mapsto \psi_{e,j}(x_{a}, x_{b})$ here, where $\psi_{e,j}$ is one of the $p_{e}$ component functions of $\bm{\psi}_{e}$. The conclusion now follows from Theorem~\ref{thm:estdf:AN}; Assumption~\ref{ass:expansion}(ii) ensures that the smoothness condition in Theorem~\ref{thm:estdf:AN} is fulfilled.
\end{proof}

\subsection{Relation to the extremal tree models}
\label{sec:TSMS2ETM}

The definition of the tree-structured max-stable distribution in Definition~\ref{def:maxstabletree} actually only emphasizes the limitations on the dependence structures, while the marginal distributions can be any of the three types of univariate extreme value distributions. We illustrate the relation between the tree-structured multivariate max-stable distribution in Definition~\ref{def:maxstabletree} and the extremal tree model in~\citet{EV20} in terms of a simple max-stable distribution. The general case can be derived in the same way by adjusting the corresponding normalizing constants, see~\citet{RT2006}. Related discussions can be found in Section A.1 of~\citet{A21} and S.5 in the Supplementary Material to~\citet{EV20}.

Under the setting of Section~\ref{subsec:MEVT}, if $\bm{\xi}$ is a random vector belonging to the domain of attraction of a simple multivariate max-stable distribution $H$, then we have 
\[
\lim_{t\to\infty} \left[\pr(\xi_{v}\le t x_v, \, v\in V)\right]^t=H(\bm{x}),\quad \bm{x}\in (0,\infty)^d.
\]
It is equivalent to
\begin{equation}
\label{eq:Mpareto}
\lim_{t\to\infty} \pr\left(\frac{\xi_v}{t}\le x_v, v\in V \, \Bigg| \, \max_{v\in V} \xi_{v}>t\right)=\tilde{H}(\bm{x}), \quad \bm{x} \in [0,\infty)^d\setminus{\bm{0}},
\end{equation}
where
\[
\tilde{H}(\bm{x})=\frac{\log H(\min(1,x_v), v\in V)-\log H(\bm{x})}{\log H(1,\ldots, 1)}
\]
with supporting set $\mathcal{L}=\{ \bm{x} \in [0,\infty)^d\setminus \{\bm{0}\}: \max_{v\in V}x_v >1\}$. The limiting distribution $\tilde{H}$ is called the multivariate Pareto distribution and $\bm{\xi}$ satisfying \eqref{eq:Mpareto} is said to be in the domain of attraction of $\tilde{H}(\bm{x})$. The distributions $H(\bm{x})$ and $\tilde{H}(\bm{x})$ are called associated since they can be recovered from each other.

Let $\tilde{\bm{Y}}=(\tilde{Y}_v, v\in V)$ be a random vector with distribution $\tilde{H}$. Then \eqref{eq:Mpareto} means that
\[
\left(\xi_v/t, v\in V\right) \mid \max_{v\in V} \xi_v >t \xrightarrow{d} \tilde{\bm{Y}},\quad \mbox{as}~t\to \infty.
\]
For $m\in V$, let $\tilde{\bm{Y}}^{m}=\tilde{\bm{Y}}\mid {\tilde{Y}_m>1}$ be the random vector $\tilde{\bm{Y}}$ conditioned on the event $\tilde{Y}_m>1$ with support on the space $\mathcal{L}^{m}=\left\{\bm{x}\in \mathcal{L}:x_m>1\right\}$. The tree graphical models in~\citet{EV20} \citep[see also][]{E20} are defined via the random vectors $\tilde{\bm{Y}}^{m}$ based on the following notion of conditional independence: For disjoint subsets $A,B,C\subset V$, they say that $\tilde{\bm{Y}}_{A}$ is conditionally independent of $\tilde{\bm{Y}}_C$ given $\tilde{\bm{Y}}_B$ if 
\[
\forall m\in \{1,\ldots,d\}, \quad \tilde{\bm{Y}}_A^m \idp \tilde{\bm{Y}}_B^m \mid \tilde{\bm{Y}}_C^m.
\]
An extremal tree model $\tilde{\bm{Y}}$ is a multivariate Pareto distributed random vector satisfying the global Markov property with respect to a tree $\tree=(V,E)$ in the sense of conditional independence defined above. 

In view of Proposition~1 in~\citet{EV20} and Proposition~\ref{pro:G_M}, the extremal tree model $\tilde{\bm{Y}}$ is exactly the multivariate Pareto distribution associated with the tree-structured multivariate max-stable distribution with respect to the same undirected tree $\tree$, stated by the coming lemma.

\begin{lemma}
$H$ is a tree-structured max-stable distribution linked to an undirected tree $\tree=(V,E)$ if and only if the associated multivariate Pareto distribution is an extremal tree model with respect to $\tree$.

Moreover, the stable tail dependence function of $H$ is given by~\eqref{eq:stdf_G}. For arbitrary $m\in V$,
$\tilde{\bm{Y}}^{m}$ is $\bm{W}_{m}$ in~\eqref{eq:W_i} having representation
\begin{equation}
\label{eq:rp}
\tilde{Y}^{m}_v\overset{d}=\begin{cases}
W, & \text{if $v=m$,} \\
W \prod_{e\in \pth{m}{v}} M_e, & \text{if $v\in V\setminus \{m\}$},
\end{cases}
\end{equation}
where $W$ is a unit-Pareto distributed random variable.
\end{lemma}

\begin{proof}
Assume $H$ is a simple tree-structured multivariate max-stable distribution. Then there exists a regularly varying Markov tree $\bm{Y}$ on a tree $\tree=(V,E)$ according to Definition~\ref{def:maxstabletree} such that $\bm{Y}$ satisfies Condition~\ref{condition}, \eqref{eq:RVMar} and $\bm{Y}\in D(H)$. Thus by~\eqref{eq:W_i} we know that $\tilde{\bm{Y}}^{m}=\bm{W}_{m}$ and it has the representation in~\eqref{eq:rp}. By Corollary 6 in~\citet{S20}, $M_{ab}$ and $M_{ba}$ satisfy Equation~(7) of~\citet{EV20}.
Since $m$ is arbitrary, it then follows from Proposition~1 in~\citet{EV20} that $\tilde{\bm{Y}}$ is an extremal tree model with respect to the tree $\tree$.

Conversely, if $\tilde{\bm{Y}}$ is an extremal tree model on $\tree$ with distribution function $\tilde{H}$, for any arbitrary but fixed node $m$, there exists a set of independent random variables $\{M_{e}, e\in E^{m}\}$ called the extremal function of $\tilde{\bm{Y}}$ such that $\tilde{\bm{Y}}^m$ is equal in distribution to the right-hand side of~\eqref{eq:rp}, where $E^m$ is the set of all edges $e\in E$ of the tree $\tree$ pointing away from node $m$, see~\citet{EV20}. We construct a Markov tree $\bm{Y}$ rooted at $m$ the same as the one in the proof of Proposition~\ref{pro:G_M} (or Definition~\ref{def:Y}). It follows from Example~3 and Corollary 5 in~\citet{S20} that all equivalent statements of Theorem 2 in that paper hold for $\bm{Y}$ with $I=\{m\}$. Since $m$ is arbitrary, we can derive from Proposition 1 in S.1 of~\citet{EV20} that $\bm{Y}$ is in the domain of attraction of $\tilde{H}$. Since we have shown in Proposition~\ref{pro:G_M} that $\bm{Y}$ is also in the domain of attraction of $H$ with stable tail dependence function given by~\eqref{eq:stdf_G}, it follows that $H$ is the associated multivariate max-stable distribution with $\tilde{H}$. The proof is complete. 
\end{proof}

\subsection{Examples of parameter estimators}
\label{subsec:three_estimators}

We list three examples of estimators of extremal dependence parameters satisfying Assumption~\ref{ass:expansion}.

\begin{example}[M-estimator]
	\label{eg:M}
	The idea behind the M-estimator proposed in \citet{EKS12} is that of matching weighted integrals of the empirical stable tail dependence function with their model-based values. For integer $q_{e} \ge p_{e}$, let $\bm{g}_{e} =\left(g_{e,1}, \ldots, g_{e,q_{e}}\right)^\top: [0,1]^{2} \rightarrow \mathbb{R}^{q_{e}}$ be a column vector of Lebesgue square-integrable functions such that the map $\bm{\varphi}_{e}: \bm{B}_{e} \rightarrow \mathbb{R}^{q_{e}}$ defined by
	\[
	\bm{\varphi}_{e}(\bm{\beta}) = 
	\int_{[0,1]^{2}} 
	\ell_{e}(x_{a},x_{b};\bm{\beta}) \, 
	\bm{g}_{e}(x_{a},x_{b}) 
	\d x_{a} \d x_{b}
	\]
	is a homeomorphism between $\bm{B}_{e}^{o}$ and $\bm{\varphi}_{e}\left(\bm{B}_{e}^{o}\right)$.
	For $e=(a,b)\in E$, the M-estimator $\hat{\bm{\beta}}_{n,k,e}^{\mathrm{M}}$ of $\bm{\beta}_{e}^{0}$ is a minimizer of the criterion function
	\begin{equation*}
		Q_{n,k,e}(\bm{\beta}) = \sum_{m=1}^{q_{e}} \left[
		\int_{[0,1]^{2}}
		\left\{
		\hat{\ell}_{n,k,e}(x_{a},x_{b}) - \ell_{e}(x_{a},x_{b};\bm{\beta})
		\right\} \,
		g_{e,m}(x_{a},x_{b})
		\d x_{a} \d x_{b} 
		\right]^{2}.
	\end{equation*}
	Suppose that $\bm{\varphi}_{\bm{e}}$ is twice continuously differentiable and that the  $q_{e} \times p_{e}$ Jacobian matrix $\nabla \bm{\varphi}_{e}(\bm{\beta}_{e}^{0})$ has rank $p_{e}$.
	Following the proof of Theorem~4.2 in \citet{EKS12}, it can be shown that expansion~\eqref{eq:expansion} holds with $\nu_{e}$ the two-dimensional Lebesgue measure and with
	\begin{equation}
		\label{eq:psiab:M}
		\bm{\psi}_{e}(x_{a}, x_{b})
		= \left\{\nabla{\bm{\varphi}}_{e}(\bm{\beta}_{e}^{0})^\top \nabla{\bm{\varphi}}_{e}(\bm{\beta}_{e}^{0})\right\}^{-1} \nabla{\bm{\varphi}}_{e}(\bm{\beta}_{e}^{0})^\top
		\bm{g}_{e}(x_{a}, x_{b}).
	\end{equation}
	By convexity, the function $\ell_{e}$ is differentiable Lebesgue almost everywhere. Corollary~\ref{cor:AN} thus applies with $\nu_{e}$ and $\bm{\psi}_{e}$ as stated.
\end{example}

\begin{example}[Method of moments]
	\label{eg:MM}
	The moments estimator of \cite{E08} arises as the special case $q_{e} = p_{e}$ of the M-estimator of \cite{EKS12} discussed in Example~\ref{eg:M} above. Rather than as a minimizer of a criterion function, the estimator $\hat{\bm{\beta}}_{n,k,e}^{\mathrm{MM}}$ is defined as the solution of the equation
	\[
	\int_{[0,1]^{2}} 
	\hat{\ell}_{n,k,e}(x_{a},x_{b}) \,
	\bm{g}_{e}(x_{a}, x_{b}) 
	\d x_{a} \d x_{b}
	=
	\bm{\varphi}_{e}\left(\hat{\bm{\beta}}_{n,k,e}^{\mathrm{MM}}\right),
	\]
	The Jacobian matrix $\nabla \bm{\varphi}_{e}(\bm{\beta}_{e}^{0})$ being an invertible square $p_{e} \times p_{e}$ matrix, the function $\bm{\psi}_{e}$ in \eqref{eq:psiab:M} simplifies to
	\[
	\bm{\psi}_{e}(x_{a}, x_{b}) = 
	\left\{ \nabla \bm{\varphi}_{e}(\bm{\beta}_{e}^{0}) \right\}^{-1} 
	\bm{g}_{e}(x_{a}, x_{b}).
	\]
\end{example}

\begin{example}[Weighted least squares estimators]
	The weighted least squares estimator studied in \citet{E18} avoids the integration of stable tail dependence functions, facilitating computations. Let $q_{e} \ge p_{e}$ and let $\bm{c}_{e}^{1}, \ldots, \bm{c}_{e}^{q_{e}}$ be $q_{e}$ points in $(0, 1]^2$.
	For $\bm{\beta}_{e} \in \bm{B}_{e}$, set
	\begin{align*}
		\hat{\bm{L}}_{n,k,e}
		&=
		\left(\hat{\ell}_{n,k,e}\left(\bm{c}_{e}^{m}\right), m=1,\ldots,q_{e}\right), &
		\bm{L}_{e}(\bm{\beta}_{e})
		&=
		\left(\ell_{e}\left(\bm{c}_{e}^{m}; \bm{\beta}_{e}\right), m=1,\ldots, q_{e}\right).
	\end{align*}
	Assume $\bm{L}_{e}$ is a homeomorphism from $\bm{B}_{e}$ to $\bm{L}_{e}(\bm{B}_{e})$.
	Let $\bm{\Omega}_{e}(\bm{\beta}_{e})$ be a $q_{e} \times q_{e}$ dimensional symmetric, positive definite matrix, twice continuously differentiable as a function of $\bm{\beta}_{e}$, and with a smallest eigenvalue that remains bounded away from zero uniformly in $\bm{\beta}_{e}$.
	The weighted least squares estimator $\hat{\bm{\beta}}_{n, k,e}^{\mathrm{WLS}}$ of $\bm{\beta}_{e}^{0}$ is defined as
	\[
	\hat{\bm{\beta}}_{n, k,e}^{\mathrm{WLS}}
	=
	\argmin_{\bm{\beta}_{e} \in \bm{B}_{e}} 
	\left(\hat{\bm{L}}_{n, k,e}-\bm{L}_{e}(\bm{\beta}_{e})\right)^\top
	\bm{\Omega}_{e}(\bm{\beta}_{e}) 
	\left(\hat{\bm{L}}_{n, k,e}-\bm{L}_{e}(\bm{\beta}_{e})\right).
	\]
	In the numerical examples in Sections~\ref{sec:Sim} and~\ref{sec:app}, $\bm{\Omega}_{e}(\bm{\beta}_{e})$ is set to be the $q_{e} \times q_{e}$ identity matrix. 
	
	Assume that $\bm{L}_{e}$ is twice continuously differentiable on a neighbourhood of $\bm{\beta}_{e}^{0}$ and that the $q_{e} \times p_{e}$ Jacobian matrix $\nabla \bm{L}_{e}(\bm{\beta}_{e}^{0})$ has rank $p_{e}$. 
	Define the $p_{e} \times q_{e}$-dimensional matrix
	\[
	\bm{A}_{e} = \left\{
	\nabla{\bm{L}}_{e}(\bm{\beta}_{e}^{0})^\top
	\bm{\Omega}_{e}(\bm{\beta}_{e}^{0})
	\nabla{\bm{L}}_{e}(\bm{\beta}_{e}^{0})
	\right\}^{-1} 
	\nabla{\bm{L}}_{e}(\bm{\beta}_{e}^{0})^\top
	\bm{\Omega}_{e}(\bm{\beta}_{e}^{0}).
	\]
	Assume that the first-order partial derivatives of $\ell_{e}$ exist in all $q_{e}$ points $\bm{c}_{e}^{m}$. Then the vector $\sqrt{k} \left( \hat{\bm{L}}_{n, k,e}-\bm{L}_{e}(\bm{\beta}_{e}^{0}) \right)$ is asymptotic normal by Theorem~\ref{thm:estdf:AN} and Remark~\ref{rk:estdf:AN}. By Theorem~2 in \cite{E18}, we have
	\begin{align*}
		\sqrt{k}\left(\hat{\bm{\beta}}_{n,k,e}^{\mathrm{WLS}} - \bm{\beta}_{e}^{0}\right)
		= \bm{A}_{e} \sqrt{k} \left( \hat{\bm{L}}_{n, k,e}-\bm{L}_{e}(\bm{\beta}_{e}^{0}) \right) + o_p(\bm{1}), \qquad n \to \infty.
	\end{align*}
	This expansion is of the form \eqref{eq:expansion} with $\nu_{e}$ the counting measure on the points $\bm{c}_{e}^{1},\ldots,\bm{c}_{e}^{q_{e}}$ and with $\bm{\psi}_{e} = (\psi_{e,1},\ldots,\psi_{e,p_{e}})^\top$ having component functions
	\[
	\psi_{e,j}(x_{a}, x_{b}) = (\bm{A}_{e})_{j,m} 
	\quad \text{if} \quad
	(x_{a}, x_{b}) = \bm{c}_{e}^{m}.
	\]
	Corollary~\ref{cor:AN} thus applies.
\end{example}

\end{document}